\documentclass[a4paper,12pt,reqno]
{amsart}
\pdfoutput=1
\usepackage{a4}
\usepackage{amsthm}
\usepackage{amssymb}
\usepackage{hyperref}
\usepackage{mathrsfs, amsmath,amsfonts,amssymb}
\newtheorem{definition}{Definition}

\newtheorem{lemma}[definition]{Lemma}
\newtheorem{remark}[definition]{Remark}
\newtheorem{corollary}[definition]{Corollary}
\newtheorem{proposition}[definition]{Proposition}
\begin{document}
	\title
[Conformal super Virasoro algebra]
{Conformal super Virasoro algebra: matrix model and quantum deformed algebra}

\author{Fridolin Melong$^{\dagger,\ddagger}$}

\address{ $\dagger$ Institut f\"ur Mathematik, Universit\"at Z\"urich,\\Winterthurerstrasse 190, CH-8057, Z\"urich, Switzerland
	{\itshape e-mail:} \normalfont  
	\texttt{fridomelong@gmail.com}}

\address{$\ddagger$ International Chair in Mathematical Physics
	and Applications
	(ICMPA-UNESCO Chair), 
	University of Abomey-Calavi,
	072 B.P. 50 Cotonou,   Benin Republic, 	\\
	and Centre International de Recherches et d'Etude Avanc\'{e}es en Sciences Math\'{e}matiques \& Informatiques et Applications (CIREASMIA), 072 B.P. 50 Cotonou, Republic of Benin.} 
	%{\itshape e-mail:} \normalfont
	%\texttt{norbert.hounkonnou@cipma.uac.bj, (with copy to hounkonnou@yahoo.fr)}}
\begin{abstract}
	In this paper, we construct the super Virasoro algebra with an arbitrary conformal dimension $\Delta$ from the generalized $\mathcal{R}(p,q)$-deformed quantum algebra and investigate 
	the $\mathcal{R}(p,q)$-deformed super Virasoro algebra with  the particular  conformal dimension $\Delta=1.$ Furthermore, we perform the $\mathcal{R}(p,q)$-conformal Virasoro $n$ algebra, the $\mathcal{R}(p,q)$-conformal super Virasoro $n$-algebra ($n$ even) and discuss a toy model for the $\mathcal{R}(p,q)$-conformal Virasoro constraints and $\mathcal{R}(p,q)$-conformal super Virasoro constraints. Besides, we generalized the notion of the $\mathcal{R}(p,q)$-elliptic hermitian matrix model with an arbitrary conformal dimension $\Delta$. Finally, we deduce relevant particular cases generated by quantum algebras known in the literature.
\end{abstract}
\maketitle
\tableofcontents
%%%%%%%%%%%%%%%%%%%%%%%%%%%%%ùùù
The notion of quantum  algebras and groups ( construction and representation theory), classical or  deformation, are  areas of mathematics, from mathematical physics of field theory and statistical mechanics. These topics were connected with many parts of mathematics and remain today an area of 
research activities. 
Many authors are investigated works related to the deformations of Virasoro algebra, see 
\cite{CILPP},\cite{CKL} (and references therein).  The Virasoro algebra with a conformal dimension $\Delta$ is related to the Korteweg-de-Vries (KdV) integrable systems, and  plays an important role in physics. This motivated the series of works devoted these last years  to its  deformation  and generalization \cite{CKL,CPP,Hounkonnou:2015laa}. 
Furthermore, the   
$q$-deformed Virasoro algebras with  conformal parameter $\Delta,$  multiplicative and comultiplication rule for  deformed generators were realized, the $q$-deformed central extension term , and related $q$-deformed KdV equation were determined for the particular cases of the conformal dimension $\Delta =0, 1/2, 1$ \cite{AS,CILPP,CEP,CPP}. Also, the  $q$-deformed energy-momentum tensor associated with the $q$-deformed central term was presented \cite{CILPP}. 
The two parameters deformation of the Virasoro algebra with conformal dimension and properties of  comultiplication were studied and investigated . Besides, the central charge term for the Virasoro algebra and the associated deformed nonlinear $(p,q)$-KdV-equation  (Korteweg-de Vries equation) were determined in \cite{CJ}.

Moreover,  the  generalized  Virasoro algebra, and related algebraic and  hydrodynamics properties were studied in \cite{Hounkonnou:2015laa}.

Now, we briefly recall definitions, notations and known results concerning $\mathcal{R}(p,q)$-calculus, quantum deformed algebras.

We consider  two positive real numbers $ p$ and $q,$ such that $ 0<q<p<1,$ and a 
meromorphic function $\mathcal{R}$ defined on $\mathbb{C}\times\mathbb{C}$ by \cite{HB}: \begin{eqnarray}\label{r10}
\mathcal{R}(s,t)= \sum_{u,v=-l}^{\infty}r_{uv}s^u\,t^v,
\end{eqnarray}
where $r_{uv}$ are complex numbers, $l\in\mathbb{N}\cup\left\lbrace 0\right\rbrace,$ $\mathcal{R}(p^n,q^n)>0,  \forall n\in\mathbb{N},$ and $\mathcal{R}(1,1)=0$ by definition. The bidisk $\mathbb{D}_{R}$ is defined by:  \begin{eqnarray*}
	\mathbb{D}_{R}
	&=&\left\lbrace w=(w_1,w_2)\in\mathbb{C}^2: |w_j|<R_{j} \right\rbrace,
\end{eqnarray*}
where $R$ is the convergence radius of the series (\ref{r10}) defined by Hadamard formula \cite{TN}:
\begin{eqnarray*}
	\lim\sup_{s+t \longrightarrow \infty} |r_{st}R^s_1\,R^t_2|^{\frac{1}{s+t}}=1
\end{eqnarray*}
and denote by  $\mathcal{O}(\mathbb{D}_{R})$ the set of holomorphic functions defined on $\mathbb{D}_{R}.$ 

The 
$\mathcal{R}(p,q)$-deformed numbers  \cite{HB}:
\begin{eqnarray}\label{rpqnumber}
[n]_{\mathcal{R}(p,q)}:=\mathcal{R}(p^n,q^n),\quad n\in\mathbb{N}\cup\{0\},
\end{eqnarray}
the
$\mathcal{R}(p,q)$-deformed factorials
\begin{eqnarray*}\label{s0}
	[n]!_{\mathcal{R}(p,q)}:=\left \{
	\begin{array}{l}
		1\quad\mbox{for}\quad n=0\\
		\\
		\mathcal{R}(p,q)\cdots\mathcal{R}(p^n,q^n)\quad\mbox{for}\quad n\geq 1,
	\end{array}
	\right .
\end{eqnarray*}
and the  $\mathcal{R}(p,q)$-deformed binomial coefficients
\begin{eqnarray*}\label{bc}
	\bigg[\begin{array}{c} m  \\ n\end{array} \bigg]_{\mathcal{R}(p,q)} := \frac{[m]!_{\mathcal{R}(p,q)}}{[n]!_{\mathcal{R}(p,q)}[m-n]!_{\mathcal{R}(p,q)}},\quad m,n\in\mathbb{N}\cup\{0\},\quad m\geq n.
\end{eqnarray*}

Consider the following linear operators defined on  $\mathcal{O}(\mathbb{D}_{R})$ by \cite{HB1}: 
\begin{align*}
\;Q:\psi\longmapsto Q\psi(z):&= \psi(qz),\\
\; P:\psi\longmapsto P\psi(z):&=\psi(pz),\\
\;{\mathcal D}_{p,q}:\psi\longmapsto {\mathcal D}_{p,q}\psi(z):&=\frac{\psi(pz)-\psi(qz)}{z(p-q)},
\end{align*}
and the $\mathcal{R}(p,q)$-derivative 
\begin{eqnarray*}\label{rpqdera}
	{\mathcal D}_{\mathcal{R}( p,q)}:={\mathcal D}_{p,q}\frac{p-q}{P-Q}\mathcal{R}( P,Q)
\end{eqnarray*}
or
\begin{eqnarray}\label{rpqder}
{\mathcal D}_{\mathcal{R}( p,q)}\psi(z):=\frac{p-q}{p^{P}-q^{Q}}\mathcal{R}(p^{P},q^{Q})\frac{\psi(pz)-\psi(qz)}{z(p-q)}.
\end{eqnarray}
The  algebra associated with the $\mathcal{R}(p,q)-$ deformation is a quantum algebra, denoted $\mathcal{A}_{\mathcal{R}(p,q)},$ generated by the set of operators $\{1, A, A^{\dagger}, N\}$ satisfying the following commutation relations\cite{HB1}:
\begin{eqnarray*}
	&& \label{algN1}
	\quad A A^\dag= [N+1]_{\mathcal {R}(p,q)},\quad\quad\quad A^\dag  A = [N]_{\mathcal {R}(p,q)}.
	\cr&&\left[N,\; A\right] = - A, \qquad\qquad\quad \left[N,\;A^\dag\right] = A^\dag
\end{eqnarray*}
with the realization on  ${\mathcal O}(\mathbb{D}_R)$ given by:
\begin{eqnarray*}\label{algNa}
	A^{\dagger} := z,\qquad A:=\partial_{\mathcal {R}(p,q)}, \qquad N:= z\partial_z,
\end{eqnarray*} 
where $\partial_z:=\frac{\partial}{\partial z}$ is the  derivative on $\mathbb{C}.$

Moreover, the  generalizations of the two parameters deformations  ($(p,q)$-deformed) Heisenberg algebras,  also called $\mathcal{R}(p,q)$-deformed quantum algebras were  investigated in \cite{HB1}. The importance of this work is to generalize  
$q$-and multi-parameter oscillator algebras well known in the physics literature. The  $\mathcal{R}(p,q)$-deformation   has the advantage to   the construction of   quantum groups and algebras.\cite{CK,CZ,G,SK}.  In the same idea, 
the $\mathcal{R}(p,q)$-deformed conformal Virasoro algebra was constructed, the $\mathcal{R}(p,q)$-deformed  Korteweg-de Vries equation for a conformal dimension $\Delta=1$ was derived, and the energy-momentum tensor from the $\mathcal{R}(p,q)$-deformed quantum algebra   for the conformal dimension $\Delta=2$ was presented by Hounkonnou and Melong \cite{HM}.

Recently, the  generalizations   of  Witt and Virasoro algebras were performed, and  the corresponding Korteweg-de Vries equations from the $\mathcal{R}(p,q)$-deformed quantum  algebras  were derived. Related relevant properties were investigated and discussed. Also,  the $\mathcal{R}(p,q)$-deformed Witt $n$-algebra was constructed, and  the Virasoro constraints for a toy model, which play an important role in the study of matrix models was presented \cite{HMM}. Following theses works, the super Virasoro $n$-algebras from the generalized quantum deformed algebras were constructed by Melong \cite{melong2022}. The generalization of the Heisenberg-Virasoro algebra and matrix models from the $\mathcal{R}(p,q)$-deformed quantum algebra were investigated by Melong and Wulkenhaar \cite{melongwulkenhaar}.  
Section $4$ describes the characterization of the $\mathcal{R}(p,q)$-conformal super Virasoro algebra, the $\mathcal{R}(p,q)$-conformal super Virasoro $n$-algebra  and a toy model for the $\mathcal{R}(p,q)$-conformal super Virasoro constraints. Particular cases are deduced. 
In Section $5$,  the conformal elliptic hermitian matrix model from the $\mathcal{R}(p,q)$-quantum algebra is obtained. We end with  concluding remarks in section $6.$
\section{$\mathcal{R}(p,q)$-conformal super Virasoro algebra $(\Delta\neq 0,1)$}
In this section, we derive  an $\mathcal{R}(p,q)$-extension of a conformal algebra with an arbitrary conformal dimension $\Delta$ \cite{CJ}. The $\mathcal{R}(p,q)$-generators are computed using the analogue of the $\mathcal{R}(p,q)$-Leibniz rule. The
$\mathcal{R}(p,q)$-super  Witt  and super  Virasoro algebras are built and discussed. A  result deduced by  Chaichian {\it et al} in \cite{CILPP} is recovered in a specific case.
\subsection{Constuction of the $\mathcal{R}(p,q)$-conformal generators}
Let $\mathcal{B}={\mathcal B}_0\oplus \mathcal{B}_1$ be  the super-commutative associative superalgebra such that ${\mathcal B}_0=\mathbb{C}\big[z,z^{-1}\big]$ and $\mathcal{B}_1=\theta\,\mathcal{B}_0,$ where $\theta$ is the Grassman variable with $\theta^2=0$ \cite{WYLWZ}:

We define the algebras endomorphism $\sigma$ on $\mathcal{B}$ as follows \cite{melong2022}:
\begin{eqnarray}\label{sigmarpq}
\sigma(t^n):=\tau_2^n\,t^n\quad\mbox{and}\quad \sigma(\theta):=\tau_2\,\theta,
\end{eqnarray}
where  $(\tau_i)_{i\in\{1,2\}}$  are functions  depending on the parameters $p$ and $q.$ 

The two linear maps $\partial_t$ and $\partial_{\theta}$ on $\mathcal{B}$ are defined by:
\begin{eqnarray*}
	\left \{
	\begin{array}{l}
		\partial_t(t^n):=[n]_{{\mathcal R}(p,q)}\,t^n\mbox{,}\quad \partial_t(\theta\,t^n):=[n]_{{\mathcal R}(p,q)}\,\theta\,t^n, \\
		\\
		\partial_{\theta}(t^n):=0\mbox{,}\quad \partial_{\theta}(\theta\,t^n):=\tau_2^n\,t^n.
	\end{array}
	\right .
\end{eqnarray*}
We assume that the linear map $\bar{\Delta}=\partial_{t}+\theta \partial_{\theta}$ on  ${\mathcal B}$ is an even $\sigma$-derivation. Then:
\begin{align}\label{deltaxyrpq}
\,\bar{\Delta}(x\,y)&=\bar{\Delta}(x)\,y+\sigma(x)\delta(y),\nonumber\\
\,\bar{\Delta}(t^n)&= [n]_{{\mathcal R}(p,q)}\,t^n\quad\mbox{and}\quad \bar{\Delta}(\theta\,t^n)= \big([n]_{{\mathcal R}(p,q)}+\tau^n_2\big)\,\theta\,t^n. 
\end{align}
%\begin{remark}
%	\begin{enumerate}
%		\item [(i)] 	Taking $\mathcal{R}(x,1)=(q-1)^{-1}(x-1)$ implies $\epsilon_1=1$ and  $\epsilon_2=q,$ we obtained the result given in \cite{WYLWZ}.
%		\item[(ii)] Putting $\mathcal{R}(x,y)=(p-q)^{-1}(x-y)$ implies $\epsilon_1=p$ and  $\epsilon_2=q,$ we have:
%		\begin{eqnarray*}
%			\left \{
%			\begin{array}{l}
%				\partial_t(t^n):=[n]_{p,q}\,t^n\mbox{,}\quad \partial_t(\theta\,t^n):=p\,[n]_{p,q}\,\theta\,t^n, \\
%				\\
%				\partial_{\theta}(t^n):=0\mbox{,}\quad \partial_{\theta}(\theta\,t^n):=q^n\,t^n
%			\end{array}
%			\right .
%		\end{eqnarray*}
%		and 
%		\begin{eqnarray*}
%			\,\Delta(x\,y)&=&\Delta(x)\,y+\sigma(x)\Delta(y),\nonumber\\
%			\,\Delta(t^n)&=& [n]_{p,q}\,t^n\quad\mbox{and}\quad \Delta(\theta\,t^n)= [n+1]_{p,q}\,\theta\,t^n. 
%		\end{eqnarray*}
%	\end{enumerate}
%\end{remark}

We consider $\phi_{\Delta}(z)$ be an arbitrary field with conformal dimension $\Delta.$ Thus, $\phi_{\Delta}(z)$ has the transformation property:
\begin{eqnarray}\label{eq1}
\phi_{\Delta}(z) \longrightarrow \big(\psi^{'}(z)\big)^{\Delta}\,\phi_{\Delta}\big(\psi(z)\big),
\end{eqnarray}
where $\psi(z)$ is a function of the conformal transformation, $"'"$ is the derivative and $\Delta$ is an arbitrary number.

The relation (\ref{eq1}) takes the infinitesimal form:
\begin{align*}
\,\psi(z)&= z+ \epsilon(z),\nonumber\\
\,\delta\big(\phi_{\Delta}(z)\big)&= \epsilon^{1-\Delta}\,\partial\,\epsilon^{\Delta}\,\phi_{\Delta}(z),
\end{align*}
where $\partial$ is the classical derivative. 

Taking $\epsilon=-z^{n+1},$ we obtain:
\begin{definition}
	The $\mathcal{R}(p,q)$-conformal super algebra is generated by bosonic and fermionic operators  $\mathcal{L}^{\Delta}_n$ of parity $0$ and $\mathcal{G}^{\Delta}_n$ of parity $1$  defined as follows:
	\begin{align*}
	\mathcal{L}^{(\Delta)}_n\phi_{\Delta}(z)&:=-z^{(n+1)(1-\Delta)}\,\bar{\Delta}(z^{(n+1)\Delta}\phi_{\Delta}(z)),\\
	\mathcal{G}^{(\Delta)}_n\phi_{\Delta}(z)&:=-\theta z^{(n+1)(1-\Delta)}\,\bar{\Delta}(z^{(n+1)\Delta}\phi_{\Delta}(z)).
	\end{align*}
\end{definition}
Note that, taking $\Delta=0,$ we recovered the $\mathcal{R}(p,q)$-generators given in \cite{melong2022}:
\begin{eqnarray*}
	\mathcal{L}_n\,\phi(z)=-z^{n}\,\bar{\Delta}\,\phi(z)\quad \mbox{and}\quad
	\mathcal{G}_n\,\phi(z)=-\theta z^{n}\,\bar{\Delta}\,\phi(z).
\end{eqnarray*}

We assume that for the field $\phi_{\Delta}(z),$ the linear $\bar{\Delta}$ can be written in the form of the $\mathcal{R}(p,q)$-number \eqref{rpqnumber} as follows:
\begin{eqnarray}\label{eq7}
\bar{\Delta}\,\phi_{\Delta}(z):=\frac{1}{z} [z\partial_z]_{\mathcal{R}(p,q)}\,\phi_{\Delta}(z).
\end{eqnarray}
Then, 
\begin{definition}
	The $\mathcal{R}(p,q)$-conformal super generators $\mathcal{L}^{(\Delta)}_n$  and $\mathcal{G}^{(\Delta)}_n$ are taking the following form:
	\begin{align}\label{scg}
	\mathcal{L}^{(\Delta)}_n\phi_{\Delta}(z)&=-[z\partial_z+\Delta(n+1)-n]_{\mathcal{R}(p,q)}\,z^n\,\phi_{\Delta}(z),\\
	\mathcal{G}^{(\Delta)}_n\phi_{\Delta}(z)&=-\theta\, [z\partial_z+\Delta(n+1)-n]_{\mathcal{R}(p,q)}\,z^n\,\phi_{\Delta}(z).
	\end{align}
\end{definition}
\begin{proposition}
	For the conformal dimenson $(\Delta \neq 0,1),$ the $\mathcal{R}(p,q)$-conformal super generators $\mathcal{L}^{(\Delta)}_n$ and $\mathcal{G}^{(\Delta)}_n$ satisfies the following algebraic structures:
	\begin{align}
	[\mathcal{L}^{(\Delta)}_n , \mathcal{L}^{(\Delta)}_m]_{X_{\Delta},Y_{\Delta}}
	&=(p-q)^{-1}\bigg\lbrace p^{N_{\Delta}}(X_{\Delta}p^{-n}-Y_{\Delta}p^{-m})\nonumber\\&- q^{N_{\Delta}}(X_{\Delta}q^{-n}-Y_{\Delta}q^{-m}) \bigg\rbrace \mathcal{L}^{(\Delta)}_{n+m},\label{p1a}\\
	[\mathcal{L}^{(\Delta)}_n , \mathcal{G}^{(\Delta)}_m]_{\tilde{X}_{\Delta},\tilde{Y}_{\Delta}}
	&=(p-q)^{-1}\bigg\lbrace p^{N_{\Delta}}(\tilde{X}_{\Delta}p^{-n}-\tilde{Y}_{\Delta}p^{-m})\nonumber\\& - q^{N_{\Delta}}(\tilde{X}_{\Delta}q^{-n}-\tilde{Y}_{\Delta}q^{-m}) \bigg\rbrace \mathcal{G}^{(\Delta)}_{n+m},\label{p1b}
	\end{align}
	and other commutators are zeros, with 
	\begin{equation}
	\left \{
	\begin{array}{l}
	{X_{\Delta}}=(pq)^{n}\frac{[n(\Delta-1)]_{p,q}[\Delta m]_{p,q}}{[n]_{p,q}[m]_{p,q}}\frac{f^{\Delta}_{n+m}(p,q)}{f^{\Delta}_n(p,q)f^{\Delta}_m(p,q)}  \\{\tilde{X}_{\Delta}}=X_{\Delta}\\
	{Y_{\Delta}}=(pq)^{m}\frac{[m(\Delta-1)]_{p,q}[\Delta n]_{p,q}}{[n]_{p,q}[m]_{p,q}}\frac{f^{\Delta}_{n+m}(p,q)}{f^{\Delta}_n(p,q)f^{\Delta}_m(p,q)}\\\\
	{\tilde{Y}_{\Delta}}=(pq)^{m+1}\frac{[(m+1)(\Delta -1)]_{p,q}[\Delta n]_{p,q}}{[n]_{p,q}[m+1]_{p,q}}\frac{f^{\Delta}_{n+m+1}(p,q)}{f^{\Delta}_n(p,q)f^{\Delta}_{m+1}(p,q)}\\\\
	f^{\Delta}_n(p,q)=\frac{p-q}{p^{\Delta(n+1)}-q^{\Delta(n+1)}}\mathcal{R}(p^{\Delta(n+1)}, q^{\Delta(n+1)})\\
	N_{\Delta} =z\partial_z + \Delta.
	\end{array}
	\right. 
	\end{equation}
\end{proposition}
\begin{proof}
	By a straightforward calculation.
\end{proof}
\begin{remark} The  conformal super generators and algebraic strucutres corresponding to quantum algebras in the literature are deduced as follows:
	\begin{enumerate}
		\item[(a)] Taking $\mathcal{R}(x)=\frac{x-x^{-1}}{q-q^{-1}},$ we obtain 
		the $q$-conformal super generators  associated to the {\bf Biedenharn-Macfarlane algebra} \cite{BC,M}
		\begin{eqnarray*}
			\mathcal{L}^{(\Delta)}_n\phi_{\Delta}(z)=-[z\partial_z+\Delta(n+1)-n]_{q}\,z^n\,\phi_{\Delta}(z),\\
			\mathcal{G}^{(\Delta)}_n\phi_{\Delta}(z)=-\theta\, [z\partial_z+\Delta(n+1)-n]_{q}\,z^n\,\phi_{\Delta}(z),
		\end{eqnarray*} 
		with
		\begin{eqnarray*}
			[n]_{q}=\frac{q^n-q^{-n}}{q-q^{-1}}
		\end{eqnarray*}	obeying the  algebraic structures\eqref{p1a} and \eqref{p1b} with the coefficients 
		\begin{equation}
		{X_{\Delta}}=\frac{[n(\Delta-1)]_{q}[\Delta m]_{q}}{[n]_{q}[m]_{q}}\mbox{,}\quad  {\tilde{X}_{\Delta}}=X_{\Delta}
		\end{equation}
		and 
		\begin{equation}
		{Y_{\Delta}}=\frac{[m(\Delta-1)]_{q}[\Delta n]_{q}}{[n]_{q}[m]_{q}},\quad 
		{\tilde{Y}_{\Delta}}=\frac{[(m+1)(\Delta -1)]_{q}[\Delta n]_{q}}{[n]_{q}[m+1]_{q}}. 
		\end{equation}
		\item[(b)] Putting $\mathcal{R}(x,y)=\frac{x-y}{p-q},$ we obtain the $(p,q)$-super  conformal generators corresponding to the deformed {\bf Jagannathan-Srinivasa algebra} \cite{JS}:
		\begin{eqnarray*}
			\mathcal{L}^{(\Delta)}_n\phi_{\Delta}(z)=-[z\partial_z+\Delta(n+1)-n]_{p,q}\,z^n\,\phi_{\Delta}(z),\\
			\mathcal{G}^{(\Delta)}_n\phi_{\Delta}(z)=-\theta\, [z\partial_z+\Delta(n+1)-n]_{p,q}\,z^n\,\phi_{\Delta}(z),
		\end{eqnarray*} 
		satisfiyng the following algebraic structures \eqref{p1a}
		and \eqref{p1b} with the coefficients
		\begin{equation}
		{X_{\Delta}}=(pq)^{n}\frac{[n(\Delta-1)]_{p,q}[\Delta m]_{p,q}}{[n]_{p,q}[m]_{p,q}},\quad  {\tilde{X}_{\Delta}}=X_{\Delta}
		\end{equation}
		and 
		\begin{equation}
		{Y_{\Delta}}=\frac{[m(\Delta-1)]_{p,q}[\Delta n]_{p,q}}{(pq)^{-m}[n]_{p,q}[m]_{p,q}},\quad
		{\tilde{Y}_{\Delta}}=\frac{[(m+1)(\Delta -1)]_{p,q}[\Delta n]_{p,q}}{(pq)^{-m-1}[n]_{p,q}[m+1]_{p,q}}. 
		\end{equation}
		\item[(c)]  For the  choice $\mathcal{R}(u,v)=(p^{-1}-q)^{-1}u^{-1}(1-uv),$ we deduce the conformal super generators induced by the \textbf{Chakrabarti-Jagannathan algebra} \cite{Chakrabarti&Jagan}:
		\begin{eqnarray*}
			\mathcal{L}^{(\Delta)}_n\phi_{\Delta}(z)=-[z\partial_z+\Delta(n+1)-n]_{p^{-1},q}\,z^n\,\phi_{\Delta}(z),\\
			\mathcal{G}^{(\Delta)}_n\phi_{\Delta}(z)=-\theta\, [z\partial_z+\Delta(n+1)-n]_{p^{-1},q}\,z^n\,\phi_{\Delta}(z),
		\end{eqnarray*} 
		obeying  the  algebraic structures \eqref{p1a} and \eqref{p1b} with 
		\begin{equation}
		{X_{\Delta}}=q^{n}\frac{[n(\Delta-1)]_{p^{-1},q}[\Delta m]_{p^{-1},q}}{p^n\,[n]_{p^{-1},q}[m]_{p^{-1},q}},\quad  {\tilde{X}_{\Delta}}=X_{\Delta},
		\end{equation} 
		\begin{equation}
		{Y_{\Delta}}=\frac{q^m\,[m(\Delta-1)]_{p^{-1},q}[\Delta n]_{p^{-1},q}}{p^{m}\,[n]_{p^{-1},q}[m]_{p^{-1},q}},
		\end{equation}
		and
		\begin{equation}
		{\tilde{Y}_{\Delta}}=\frac{q^{m+1}[(m+1)(\Delta -1)]_{p^{-1},q}[\Delta n]_{p^{-1},q}}{p^{m+1}[n]_{p^{-1},q}[m+1]_{p^{-1},q}}. 
		\end{equation}
		\item[(d)] Putting $\mathcal{R}(s,t)=((q-p^{-1})t)^{-1}(st-1)$, we obtain the conformal super generators associated to the \textbf{generalized $q$-Quesne} deformed  algebra \cite{HN}: 
		\begin{eqnarray*}
			\mathcal{L}^{(\Delta)}_n\phi_{\Delta}(z)=-[z\partial_z+\Delta(n+1)-n]^Q_{p,q}\,z^n\,\phi_{\Delta}(z),\\
			\mathcal{G}^{(\Delta)}_n\phi_{\Delta}(z)=-\theta\, [z\partial_z+\Delta(n+1)-n]^Q_{p,q}\,z^n\,\phi_{\Delta}(z),
		\end{eqnarray*}
		satisfying the commutation relations \eqref{p1a} and \eqref{p1b} with 
		\begin{equation}
		{X_{\Delta}}=p^{n}\frac{[n(\Delta-1)]^Q_{p,q}[\Delta m]^Q_{p,q}}{q^n\,[n]^Q_{p,q}[m]^Q_{p,q}},\quad {Y_{\Delta}}=\frac{p^m\,[m(\Delta-1)]^Q_{p,q}[\Delta n]^Q_{p,q}}{q^{m}[n]^Q_{p,q}[m]^Q_{p,q}} 
		\end{equation}
		and 
		\begin{equation}
		{\tilde{X}_{\Delta}}=X_{\Delta},\quad
		{\tilde{Y}_{\Delta}}=\frac{p^{m+1}[(m+1)(\Delta -1)]^Q_{p,q}[\Delta n]^Q_{p,q}}{q^{m+1}[n]^Q_{p,q}[m+1]^Q_{p,q}}. 
		\end{equation}
		\item[(e)] Taking $\mathcal{R}(x,y)=g(p,q)\frac{y^{\nu}}{x^{\mu}}\frac{xy -1}{(q-p^{-1})y}$, 
		we obtain the conformal super generators induced by the \textbf{Hounkonnou-Ngompe generalized 
			algebra}\cite{HNN}: 
		\begin{eqnarray*}
			\mathcal{L}^{(\Delta)}_n\phi_{\Delta}(z)=-[z\partial_z+\Delta(n+1)-n]^{\mu,\nu}_{p,q,g}\,z^n\,\phi_{\Delta}(z),\\
			\mathcal{G}^{(\Delta)}_n\phi_{\Delta}(z)=-\theta\, [z\partial_z+\Delta(n+1)-n]^{\mu,\nu}_{p,q,g}\,z^n\,\phi_{\Delta}(z),
		\end{eqnarray*}
		obeying  the commutation relations \eqref{p1a}  and \eqref{p1b} with 
		\begin{equation}
		{X_{\Delta}}=(pq)^{n}\frac{[n(\Delta-1)]^{\mu,\nu}_{p,q,g}[\Delta m]^{\mu,\nu}_{p,q,g}}{[n]^{\mu,\nu}_{p,q,g}[m]^{\mu,\nu}_{p,q,g}},\quad  {\tilde{X}_{\Delta}}=X_{\Delta}
		\end{equation}
		and 
		\begin{equation}
		{Y_{\Delta}}=\frac{[m(\Delta-1)]^{\mu,\nu}_{p,q,g}[\Delta n]^{\mu,\nu}_{p,q,g}}{(pq)^{-m}[n]^{\mu,\nu}_{p,q,g}[m]^{\mu,\nu}_{p,q,g}},\quad
		{\tilde{Y}_{\Delta}}=\frac{[(m+1)(\Delta -1)]^{\mu,\nu}_{p,q,g}[\Delta n]^{\mu,\nu}_{p,q,g}}{(pq)^{-m-1}[n]^{\mu,\nu}_{p,q,g}[m+1]^{\mu,\nu}_{p,q,g}}. 
		\end{equation}
	\end{enumerate}
\end{remark}
We redefine the parameters $p$ and $q$ as follows:
\begin{eqnarray*}
	\Theta=\sqrt{\frac{p}{q}}\quad\mbox{and}\quad \lambda=\sqrt{\frac{1}{pq}},
\end{eqnarray*}
with
\begin{equation*}
[n]=\lambda^{1-n}[n]_{\Theta}\quad\mbox{and}\quad [n]_{\Theta}=\big(\Theta-\Theta^{-1}\big)^{-1}\big(\Theta^n-\Theta^{-n}\big).
\end{equation*}
Then, the algebra \eqref{p1a}) and \eqref{p1b}, may be mapped to the $\mathcal{R}(\Theta)$-super Virasoro algebra which is the generalization of the $\Theta$-Virasoro algebra presented  by Chaichian  et {\it al} \cite{CILPP}.
\begin{proposition}
	The super generators given as follows:
	\begin{eqnarray*}
		\mathbb{L}_n\phi_{\Delta}(z)=\lambda^{N_{\Delta}-1}\,\mathcal{L}^{(\Delta)}_n\phi_{\Delta}(z)\quad\mbox{and}\quad \mathbb{G}_n\phi_{\Delta}(z)=\lambda^{N_{\Delta}-1}\,\mathcal{G}^{(\Delta)}_n\phi_{\Delta}(z) 
	\end{eqnarray*} 
	satisfies the following commutation relations:
	\begin{align*}
	[\mathbb{L}_n , \mathbb{L}_m]_{X,Y}
	&= K\left\lbrace \Theta^{N_{\Delta}}(X\Theta^{-n}-Y\Theta^{-m}) - \Theta^{N_{\Delta}}(X\Theta^{-n}-Y\Theta^{-m}) \right\rbrace \mathbb{L}_{n+m},\\
	[\mathbb{L}_n , \mathbb{G}_m]_{\tilde{X},\tilde{Y}}
	&= K\left\lbrace \Theta^{N_{\Delta}}(\tilde{X}\Theta^{-n}-\tilde{Y}\Theta^{-m}) - \Theta^{N_{\Delta}}(\tilde{X}\Theta^{-n}-\tilde{Y}\Theta^{-m}) \right\rbrace \mathbb{G}_{n+m},
	\end{align*}
	and other commutation relations are zeros, where 
	\begin{eqnarray*}
		\left \{
		\begin{array}{l}
			{X}=\frac{[n(\Delta-1)]_{\Theta}[\Delta m]_{\Theta}}{[n]_{\Theta}[m]_{\Theta}}\frac{f^{\Delta}_{n+m}(\Theta)}{f^{\Delta}_n(\Theta)f^{\Delta}_m(\Theta)}  \\{\tilde{X}}=X\\
			{Y}=\frac{[m(\Delta-1)]_{\Theta}[\Delta n]_{\Theta}}{[n]_{\Theta}[m]_{\Theta}}\frac{f^{\Delta}_{n+m}(\Theta)}{f^{\Delta}_n(\Theta)f^{\Delta}_m(\Theta)}\\
			{\tilde{Y}}=\frac{[(m+1)(\Delta -1)]_{\Theta}[\Delta n]_{\Theta}}{[n]_{\Theta}[m+1]_{\Theta}}\frac{f^{\Delta}_{n+m+1}(\Theta)}{f^{\Delta}_n(\Theta)f^{\Delta}_{m+1}(\Theta)}\\
			K=\big(\Theta-\Theta^{-1}\big)^{-1}.
		\end{array}
		\right. 
	\end{eqnarray*}
\end{proposition}
%\begin{proof}
%	From the $\mathcal{R}(p,q)-$ derivative \eqref{rpqder}, we have
%	\begin{eqnarray}
%	[n]_{\mathcal{R}(p,q)}=[n]_{p,q}\frac{p-q}{p^{n}-q^{n}}\mathcal{R}(p^{n},q^{n}).
%	\end{eqnarray}
%	Then, the generators takes the form:
%	\begin{eqnarray}
%	\mathcal{L}^{\alpha}_n\phi_{\alpha}(z) =L_n\,f^{\alpha}_n(p,q)\phi_{\alpha}(z),
%	\end{eqnarray}
%	where $$f^{\alpha}_n(p,q)\phi_{\alpha}(z)=\frac{p-q}{p^{\alpha(n+1)}-q^{\alpha(n+1)}}\mathcal{R}(p^{\alpha(n+1)}, q^{\alpha(n+1)}).$$ Then,
%	\begin{small}
%	\begin{eqnarray}
%	[\mathcal{L}^{(\alpha)}_n , \mathcal{L}^{(\alpha)}_m]_{X_{\alpha},Y_{\alpha}}
%	&=&X_{\alpha}\mathcal{L}^{(\alpha)}_n \mathcal{L}^{(\alpha)}_m-Y_{\alpha} \mathcal{L}^{(\alpha)}_m\mathcal{L}^{(\alpha)}_n \nonumber\\
%	&=&f^{\alpha}_n(p,q)f^{\alpha}_m(p,q)\big(X_{\alpha}L_n\,L_m-Y_{\alpha}L_m\,L_n\big)\nonumber\\
%	&=&
%	\end{eqnarray}
%\end{small}
%	\cqfd
%\end{proof}

%	\begin{remark}
%		\item (i) For the choice of the meromorphic function $\mathcal{R}(x,y)=(p-q)^{-1}(\frac{1}{x}-\frac{1}{y})$, we get the $(p,q)$-Virasoro generators given by Chakrabarti et {\it al}\cite{CJ}.
%		\item(ii) The $\mathcal{R}(p,q)$-Virasoro generators $\mathcal{L}^{(\Delta)}_n$ can be expressed in  terms of the $(p,q)$-Virasoro generators $e^{\Delta
%		}_n,$ i.e. \begin{equation}\label{r1}
%		\mathcal{L}^{(\Delta)}_n = w^{(\Delta)}_n(p,q)e^{\Delta
%		}_n
%		\end{equation}
%	\end{remark}
Now, we determine the $\mathcal{R}(p,q)$-conformal super Witt algebra. The result is presented in the lemma bellow:
\begin{lemma}
	The $\mathcal{R}(p,q)$-super Witt algebra with conformal dimension $(\Delta\neq 0,1 )$ is driven by  the following commutation  relations:
	\begin{align}
	\big[\mathcal{L}^{(\Delta)}_n,\mathcal{L}^{(\Delta)}_{m}\big]_{x, y}\phi_{\Delta}(z)&=[m-n]_{\mathcal{R}(p,q)}\,\mathcal{L}^{(\Delta)}_{n+m}\phi_{\Delta}(z),\label{crochet1q}\\
	\big[\mathcal{L}^{(\Delta)}_{n},\mathcal{G}^{(\Delta)}_{m}\big]_{\tilde{x}, \tilde{y}}\phi_{\Delta}(z)&= \big([m+1]_{\mathcal{R}(p,q)}-[n]_{\mathcal{R}(p,q)}\big)\,\mathcal{G}^{(\Delta)}_{n+m}\phi_{\Delta}(z),\label{crochet2}\\
	\big[\mathcal{G}^{(\Delta)}_{n},\mathcal{G}^{(\Delta)}_{m}\big]&=0,
	\end{align}
	where
	\begin{eqnarray}\label{coeff}
	\left \{
	\begin{array}{l}
	x=(p-q)\,[n-m]_{\mathcal{R}(p,q)}\,\chi_{nm}(p,q), \\
	y=(p-q)\,[n-m]_{\mathcal{R}(p,q)}\,\chi_{mn}(p,q),\\
	\tilde{x}=(p-q)\,\big([m+1]_{\mathcal{R}(p,q)}-[n]_{\mathcal{R}(p,q)}\big)\,\varLambda_{nm}(p,q), \\
	\tilde{y}=(p-q)\,\big([m+1]_{\mathcal{R}(p,q)}-[n]_{\mathcal{R}(p,q)}\big)\,\varLambda_{mn}(p,q),
	\end{array}
	\right .
	\end{eqnarray}
	\begin{align*}\chi_{mn}(p,q)&=\bigg\lbrace p^{N_{\Delta}}(p^{-n}-\frac{[m(\Delta-1)]_{p,q}[\Delta n]_{p,q}}{[n(\Delta-1)]_{p,q}[\Delta m]_{p,q}}\frac{q^{m-n}}{p^{n}})\nonumber\\&- q^{N_{\Delta}}(q^{-n}-\frac{[m(\Delta-1)]_{p,q}[\Delta n]_{p,q}}{[n(\Delta-1)]_{p,q}[\Delta m]_{p,q}}\frac{p^{m-n}}{q^{n}}) \bigg\rbrace^{-1}
	\end{align*}
	and 
	\begin{align*}
	\varLambda_{mn}(p,q)&=\bigg\lbrace p^{N_{\Delta}}\big(p^{-n}-\frac{[(m+1)(\Delta -1)]_{p,q}[\Delta n]_{p,q}[m]_{p,q}}{[m+1]_{p,q}[n(\Delta-1)]_{p,q}[\Delta m]_{p,q}}\frac{pq^{m+1}}{(pq)^{n
	}}\big)\nonumber\\&- q^{N_{\Delta}}(q^{-n}-\frac{[(m+1)(\Delta -1)]_{p,q}[\Delta n]_{p,q}[m]_{p,q}}{[m+1]_{p,q}[n(\Delta-1)]_{p,q}[\Delta m]_{p,q}}\frac{qp^{m+1}}{(pq)^{n
	}}) \bigg\rbrace^{-1}.
	\end{align*}
\end{lemma}
\begin{proof}
	We use the algebraic structure \eqref{p1a}, \eqref{p1b} and setting 
	\begin{equation*}
	\chi_{mn}(p,q)=\left\lbrace p^{N_{\Delta}}(p^{-n}-\frac{Y_{\Delta}}{X_{\Delta}}p^{-m}) - q^{N_{\Delta}}(q^{-n}-\frac{Y_{\Delta}}{X_{\Delta}}q^{-m}) \right\rbrace^{-1}
	\end{equation*}
	and 
	\begin{equation*}
	\varLambda_{mn}(p,q)=\left\lbrace p^{N_{\Delta}}(p^{-n}-\frac{\tilde{Y}_{\Delta}}{\tilde{X}_{\Delta}}p^{-m}) - q^{N_{\Delta}}(q^{-n}-\frac{\tilde{Y}_{\Delta}}{\tilde{X}_{\Delta}}q^{-m}) \right\rbrace^{-1}
	\end{equation*}
\end{proof}
Setting \begin{eqnarray}\label{rho}
\rho(\mathcal{L}^{(\Delta)}_{n})={[2\,n]_{\mathcal{R}(p,q)}\over [n]_{\mathcal{R}(p,q)}}\,\mathcal{L}^{(\Delta)}_{n}\quad\mbox{and}\quad \rho(\mathcal{G}^{(\Delta)}_{n})={[2(n+1)]_{\mathcal{R}(p,q)}\over [n+1]_{\mathcal{R}(p,q)}}\,\mathcal{G}^{(\Delta)}_{n},
\end{eqnarray}
then  the $\mathcal{R}(p,q)$-deformed super Jacobi identity is presented by the following relation:
\begin{eqnarray}
\sum_{(i,j,l)\in\mathcal{C}(n,m,k)}\,(-1)^{|B_i||B_l|}\Big[\rho(B_i),\Big[B_j,B_l\Big]_{{\mathcal R}(p,q)}\Big]_{{\mathcal R}(p,q)}=0,
\end{eqnarray} 
where  $\mathcal{C}(n,m,k)$ denotes the cyclic permutation of $(n,m,k)$ and the symbol $|B|$ appearing in the exponent of $(-1)$ is to be understood
as the parity of $B.$ 

\begin{proposition}
	The $\mathcal{R}(p,q)$-super Virasoro algebra with conformal dimension $(\Delta \neq 0,1)$ is generated by the following commutation relations:
	\begin{align}\label{J21}
	\,[L^{(\Delta)}_n, L^{(\Delta)}_m]_{x,y}&=[n-m]_{\mathcal{R}(p,q)}\,L^{(\Delta)}_{n+m} + \delta_{n+m,0}\,C^{\Delta}_n(p,q),\\
	\,[L^{(\Delta)}_n, G^{(\Delta)}_m]_{\hat{x},\hat{y}}&=\big([m+1]_{\mathcal{R}(p,q)}-[n]_{\mathcal{R}(p,q)}\big)G^{(\Delta)}_{n+m} + \delta_{n+m+1,0}\,C^{\Delta}_n(p,q),\\
	\,[G^{(\Delta)}_n, G^{(\Delta)}_m]&=0,
	\end{align}
	where $C^{\Delta}_n(p,q)$ is the central charge provided by \cite{HM}:
	\begin{equation}\label{ce1}
	C^{\Delta}_n(p,q) = C(p,q)\frac{[n]_{\mathcal{R}(p,q)}}{[2n]_{\mathcal{R}(p,q)}}(pq)^{\frac{N_{\Delta}}{2}+n}[n-1]_{\mathcal{R}(p,q)}[n]_{\mathcal{R}(p,q)}[n+1]_{\mathcal{R}(p,q)},
	\end{equation}
	$x,$ $y,$ $\tilde{x},$ $\tilde{y}$ are given by \eqref{coeff},
	and $C(p,q)$ is a function of $(p,q).$ 
\end{proposition} 
\begin{remark}
	The  particular cases are also attract our attention. We deduce the conformal super Virasoro induced by quantum algebras in the literature.
	\begin{enumerate}	
		\item[(a)] The conformal super Virasoro algebra corresponding to the {\bf Biedenharn-Macfarlane algebra} \cite{BC,M}:
		\begin{align*}
		\,[L^{(\Delta)}_n, L^{(\Delta)}_m]_{x,y}&=[n-m]_{q}\,L^{(\Delta)}_{n+m} + \delta_{n+m,0}\,C^{\Delta}_n(q),\\
		\,[L^{(\Delta)}_n, G^{(\Delta)}_m]_{\hat{x},\hat{y}}&=\big([m+1]_{q}-[n]_{q}\big)G^{(\Delta)}_{n+m} + \delta_{n+m+1,0}\,C^{\Delta}_n(q),
		\end{align*}
		and other commutator is zero, with 
		\begin{eqnarray*}
			\left \{
			\begin{array}{l}
				x=\big(q^{n-m}-q^{m-n}\big)\,\chi_{nm}(q), \\
				y=\big(q^{n-m}-q^{m-n}\big)\,\chi_{mn}(q),\\
				\tilde{x}=\big(q^{m+1}-q^{n}+q^{-n}-q^{-m-1}\big)\,\varLambda_{nm}(q), \\
				\tilde{y}=\big(q^{n-m}-q^{m-n}\big)\,\chi_{mn}(q),\\
				\tilde{x}=\big(q^{m+1}-q^{n}+q^{-n}-q^{-m-1}\big)\,\varLambda_{mn}(q),
			\end{array}
			\right .
		\end{eqnarray*}
		\begin{equation*}
		C^{\Delta}_n(q) = C(q)(q^n + q^{-n})^{-1}[n-1]_{q}[n]_{q}[n+1]_{q}.
		\end{equation*}
		\item[(b)] The conformal super Virasoro algebra from the \textbf{Jagannathan-Srinivasa algebra} \cite{JS} is driven by:
		\begin{align*}
		\,[L^{(\Delta)}_n, L^{(\Delta)}_m]_{x,y}&=[n-m]_{p,q}\,L^{(\Delta)}_{n+m} + \delta_{n+m,0}\,C^{\Delta}_n(p,q),\\
		\,[L^{(\Delta)}_n, G^{(\Delta)}_m]_{\hat{x},\hat{y}}&=\big([m+1]_{p,q}-[n]_{p,q}\big)G^{(\Delta)}_{n+m} + \delta_{n+m+1,0}\,C^{\Delta}_n(p,q),
		\end{align*}
		and other commutator is zero, with 
		\begin{eqnarray*}
			\left \{
			\begin{array}{l}
				x=\big(p^{n-m}-q^{n-m}\big)\,\chi_{nm}(p,q), \\
				y=\big(p^{n-m}-q^{n-m}\big)\,\chi_{mn}(p,q),\\
				\tilde{x}=\big(p^{m+1}-p^{n}+q^{n}-q^{m+1}\big)\,\varLambda_{nm}(p,q), \\
				\tilde{y}=\big(p^{m+1}-p^{n}+q^{n}-q^{m+1}\big)\,\varLambda_{mn}(p,q),
			\end{array}
			\right .
		\end{eqnarray*}
		\begin{equation}
		C^{\Delta}_n(p,q) = C(p,q)(pq)^{\frac{N_{\Delta}}{2}+n}(p^n + q^n)^{-1}[n-1]_{p,q}[n]_{p,q}[n+1]_{p,q}.
		\end{equation}
		\item The conformal super Virasoro algebra induced by the  \textbf{Chakrabarti- Jagannathan algebra}\cite{Chakrabarti&Jagan}: 
		\begin{align*}
		\,[L^{(\Delta)}_n, L^{(\Delta)}_m]_{x,y}&=[n-m]_{p^{-1},q}\,L^{(\Delta)}_{n+m} + \delta_{n+m,0}\,C^{\Delta}_n(p^{-1},q),\\
		\,[L^{(\Delta)}_n, G^{(\Delta)}_m]_{\hat{x},\hat{y}}&=\big([m+1]_{p^{-1},q}-[n]_{p^{-1},q}\big)G^{(\Delta)}_{n+m} + \delta_{n+m+1,0}\,C^{\Delta}_n(p^{-1},q),
		\end{align*}
		and other commutator is zero, with 
		\begin{eqnarray*}
			\left \{
			\begin{array}{l}
				x=(p-q)\,[n-m]_{p^{-1},q}\,\chi_{nm}(p^{-1},q), \\
				y=(p-q)\,[n-m]_{p^{-1},q}\,\chi_{mn}(p^{-1},q),\\
				\tilde{x}=(p-q)\,\big([m+1]_{p^{-1},q}-[n]_{p^{-1},q}\big)\,\varLambda_{nm}(p^{-1},q), \\
				\tilde{y}=(p-q)\,\big([m+1]_{p^{-1},q}-[n]_{p^{-1},q}\big)\,\varLambda_{mn}(p^{-1},q),
			\end{array}
			\right .
		\end{eqnarray*}
		\begin{equation*}
		C^{\Delta}_n(p^{-1},q) = C(p,q)(\frac{q}{p})^{\frac{N_{\Delta}}{2}+n}(p^{-n} + q^n)^{-1}[n-1]_{p^{-1},q}[n]_{p^{-1},q}[n+1]_{p^{-1},q}.
		\end{equation*}
		\item The conformal super Virasoro algebra induced by the  \textbf{Generalized $q$-Quesne algebra} \cite{HN}:
		\begin{align*}
		\,[L^{(\Delta)}_n, L^{(\Delta)}_m]_{x,y}&=[n-m]^Q_{p,q}\,L^{(\Delta)}_{n+m} + \delta_{n+m,0}\,C^{\Delta}_n(p,q),\\
		\,[L^{(\Delta)}_n, G^{(\Delta)}_m]_{\hat{x},\hat{y}}&=\big([m+1]^Q_{p,q}-[n]^Q_{p,q}\big)G^{(\Delta)}_{n+m} + \delta_{n+m+1,0}\,C^{\Delta}_n(p,q),
		\end{align*}
		and other commutator is zero, with 
		\begin{eqnarray*}
			\left \{
			\begin{array}{l}
				x=(p-q)\,[n-m]^Q_{p,q}\,\chi^Q_{nm}(p,q), \\
				y=(p-q)\,[n-m]^Q_{p,q}\,\chi^Q_{mn}(p,q),\\
				\tilde{x}=(p-q)\,\big([m+1]^Q_{p,q}-[n]^Q_{p,q}\big)\,\varLambda^Q_{nm}(p,q), \\
				\tilde{y}=(p-q)\,\big([m+1]^Q_{p,q}-[n]^Q_{p,q}\big)\,\varLambda^Q_{mn}(p,q),
			\end{array}
			\right .
		\end{eqnarray*} 
		\begin{equation}
		C^{\Delta}_n(p,q) = C(p,q)(\frac{p}{q})^{\frac{N_{\Delta}}{2}+n}(p^n + q^{-n})^{-1}[n-1]^Q_{p,q}[n]^Q_{p,q}[n+1]^Q_{p,q} 
		\end{equation}
		\item The conformal super Virasoro algebra generated by the \textbf{Hounkonnou-Ngompe generalized algebra} \cite{HNN}:
		\begin{align*}
		\,[L^{(\Delta)}_n, L^{(\Delta)}_m]_{x,y}&=[n-m]^{\mu,\nu}_{p,q,g}\,L^{(\Delta)}_{n+m} + \delta_{n+m,0}\,C^{\Delta}_n(p,q),\\
		\,[L^{(\Delta)}_n, G^{(\Delta)}_m]_{\hat{x},\hat{y}}&=\big([m+1]^{\mu,\nu}_{p,q,g}-[n]^{\mu,\nu}_{p,q,g}\big)G^{(\Delta)}_{n+m} + \delta_{n+m+1,0}\,C^{\Delta}_n(p,q),
		\end{align*}
		and other commutator is zero, with 
		\begin{eqnarray*}
			\left \{
			\begin{array}{l}
				x=(p-q)\,[n-m]^{\mu,\nu}_{p,q,g}\,\chi^{\mu,\nu}_{nm}(p,q,g), \\
				y=(p-q)\,[n-m]^{\mu,\nu}_{p,q,g}\,\chi^{\mu,\nu}_{mn}(p,q,g),\\
				\tilde{x}=(p-q)\,\big([m+1]^{\mu,\nu}_{p,q,g}-[n]^{\mu,\nu}_{p,q,g}\big)\,\varLambda^{\mu,\nu}_{nm}(p,q,g), \\
				\tilde{y}=(p-q)\,\big([m+1]^{\mu,\nu}_{p,q,g}-[n]^{\mu,\nu}_{p,q,g}\big)\,\varLambda^{\mu,\nu}_{mn}(p,q,g),
			\end{array}
			\right .
		\end{eqnarray*}
		\begin{equation}
		\tilde{C}^R_n(p,q) = C(p,q)(pq)^{\frac{N_{\Delta}}{2}+n}(p^n + q^n)^{-1}[n-1]^{\mu,\nu}_{p, q, g}[n]^{\mu,\nu}_{p, q, g}[n+1]^{\mu,\nu}_{p, q, g}. 
		\end{equation}
	\end{enumerate}
\end{remark}
\subsection{Generalized super Virasoro algebra with  $\Delta=1$}
In this section, we consider the algebra given by the relation \eqref{p1a} and \eqref{p1b} for $\Delta=1,$ and derive the  $\mathcal{R}(p,q)$-deformed super Virasoro algebra. 
Then,  the $\mathcal{R}(p,q)$-deformed super generators
$\mathcal{L}^1_n $ and $\mathcal{G}^1_n :$
\begin{eqnarray*}
	\mathcal{L}^{1}_n\phi(z)=-[z\partial_z+1]_{\mathcal{R}(p,q)}\,z^n\,\phi(z),\\
	\mathcal{G}^{1}_n\phi(z)=-\theta\, [z\partial_z+1]_{\mathcal{R}(p,q)}\,z^n\,\phi(z)
\end{eqnarray*}
defined the superalgebra
obeying the following relations:
\begin{align}\label{d4}
\,[\mathcal{L}^{1}_n , \mathcal{L}^{1}_m]_{\hat{u},\hat{v}}\phi(z)&=[n-m]_{\mathcal{R}(p,q)}q^{N_{1}-m}\mathcal{L}^{1}_{n+m}\phi(z),\\
\,[\mathcal{L}^{1}_n , \mathcal{G}^{1}_m]_{\tilde{u},\tilde{v}}\phi(z)&=\big([m+1]_{\mathcal{R}(p,q)}-[n]_{\mathcal{R}(p,q)}\big)q^{N_{1}-m}\mathcal{G}^{1}_{n+m}\phi(z),\\
\,[\mathcal{G}^{1}_n , \mathcal{G}^{1}_m]&=0\nonumber,
\end{align}
with 
\begin{eqnarray}\label {d5}
\left \{
\begin{array}{l}
\displaystyle
\hat{u}=\hat{\Lambda}_{nm}(p,q),\quad \tilde{u}=\tilde{\Lambda}_{nm}(p,q),\\
\\
\hat{v}=p^{m-n}\hat{\Lambda}_{nm}(p,q),\quad \tilde{v}=p^{m-n}\,\tilde{\Lambda}_{nm}(p,q)\\
,
\end{array}
\right. 
\end{eqnarray}
\begin{eqnarray*}
	\hat{\Lambda}_{nm}(p,q)=\frac{[n-m]_{\mathcal{R}(p,q)}}{[m-n]_{p,q}}\frac{f^{1}_{n+m}(p,q)}{f^{1}_n(p,q)f^{1}_m(p,q)},
\end{eqnarray*}
and 
\begin{eqnarray*}
	\tilde{\Lambda}_{nm}(p,q)=\frac{[m+1]_{\mathcal{R}(p,q)}-[n]_{\mathcal{R}(p,q)}}{[m-n]_{p,q}}\frac{f^{1}_{n+m}(p,q)}{f^{1}_n(p,q)f^{1}_m(p,q)}.
\end{eqnarray*}

Now, we rewrite the $\mathcal{R}(p,q)$-deformed super generators in the following way:
\begin{eqnarray*}
	L^{1}_n\phi(z) =q^{-N_1}\,\mathcal{L}^{1}_n\phi(z)\quad\mbox{and}\quad 
	G^{1}_n\phi(z)=q^{-N_1}\,\mathcal{G}^{1}_n\phi(z).
\end{eqnarray*}
Then, they obeys  the following commutation relations:
\begin{align}
\,[L^1_n , L^1_m]_{u,v}\phi(z)&=[n-m]_{\mathcal{R}(p,q)}L^1_{n+m}\phi(z),\label{d15a}\\
\,[L^{1}_n ,G^{1}_m]_{a,b}\phi(z)&=\big([m+1]_{\mathcal{R}(p,q)}-[n]_{\mathcal{R}(p,q)}\big)\,G^{1}_{n+m}\phi(z),\label{d15b}\\
\,[G^{1}_n ,G^{1}_m]&=0\label{d15c},
\end{align}
where
\begin{align}\label {d16}
u&=q^{m-n}\,\hat{\Lambda}_{nm}(p,q),\quad
v=p^{m-n}\,\hat{\Lambda}_{mn}(p,q),\\
a&=q^{m-n}\,\tilde{\Lambda}_{nm}(p,q),\quad
b=p^{m-n}\,\tilde{\Lambda}_{nm}(p,q),
\end{align}
equivalently,
\begin{align*}
\,[L^1_n , L^1_m]\phi(z)&=[n-m]_{\mathcal{R}(p,q)}\,p^{N_1-m}q^{-N_1+n}\,\hat{\varLambda}_{nm}(p,q)L^1_{n+m}\phi(z)\\
\,[L^1_n , G^1_m]\phi(z)&=\big([m+1]_{\mathcal{R}(p,q)}-[n]_{\mathcal{R}(p,q)}\big)p^{N_1-m}q^{-N_1+n}\,\tilde{\varLambda}_{nm}(p,q)G^1_{n+m}\phi(z)\\
\,[G^1_n , G^1_m]&=0.
\end{align*}
\begin{remark} 
	From the particular cases of the intergers $n$ and $m,$ we deduce the $\mathcal{R}(p,q)$-deformed super $su(1,1)$ subalgebra is generated by the  commutation relations:
	\begin{align*}
	\,[L^1_0 , L^1_1]_{u,v}\phi(z)&=[-1]_{\mathcal{R}(p,q)}L^1_{1}\phi(z)\\
	\,[L^{1}_0 ,G^{1}_1]_{a,b}\phi(z)&=\big([2]_{\mathcal{R}(p,q)}\big)\,G^{1}_{1}\phi(z)\\
	\,[G^{1}_0 ,G^{1}_1]&=0,
	\end{align*}
	where
	\begin{align*}
	u&=q\,\hat{\Lambda}_{01}(p,q),\quad
	v=p\,\hat{\Lambda}_{10}(p,q),\\
	a&=q\,\tilde{\Lambda}_{01}(p,q),\quad
	b=p\,\tilde{\Lambda}_{10}(p,q), 
	\end{align*}
	\begin{align*}
	\,[L^1_{-1} , L^1_0]_{u,v}\phi(z)&=[-1]_{\mathcal{R}(p,q)}L^1_{-1}\phi(z)\\
	\,[L^{1}_{-1} ,G^{1}_0]_{a,b}\phi(z)&=\big([1]_{\mathcal{R}(p,q)}-[-1]_{\mathcal{R}(p,q)}\big)\,G^{1}_{-1}\phi(z)\\
	\,[G^{1}_{-1} ,G^{1}_0]&=0,
	\end{align*}
	where
	\begin{align*}
	u&=q\,\hat{\Lambda}_{-10}(p,q),\quad
	v=p\,\hat{\Lambda}_{0(-1)}(p,q),\\
	a&=q\,\tilde{\Lambda}_{-10}(p,q),\quad
	b=p\,\tilde{\Lambda}_{0(-1)}(p,q), 
	\end{align*}
	and
	\begin{align*}
	\,[L^1_{-1} , L^1_1]\phi(z)&= [2]p^{N_1-1}q^{-N_1-1}\,\hat{\varLambda}_{-10}(p,q)L^1_{0}\phi(z)\\
	\,[L^1_{-1} , G^1_1]\phi(z)&=\big([2]_{\mathcal{R}(p,q)}-[-1]_{\mathcal{R}(p,q)}\big)p^{N_1-1}q^{-N_1-1}\,\tilde{\varLambda}_{-11}(p,q)G^1_{0}\phi(z)\\
	\,[G^1_{-1} , G^1_1]\phi(z)&=0.
	\end{align*}
\end{remark}

The central extension of the algebra (\ref{d15a}),  (\ref{d15b}), and (\ref{d15c}) is generated by the commutation relations:
\begin{align}
\,[L^1_n , L^1_m]_{u,v}\phi(z)&=[n-m]_{\mathcal{R}(p,q)}L^1_{n+m}\phi(z)+\delta_{n+m,0}\,C^{R}_n(p,q),\label{d17a}\\
\,[L^{1}_n ,G^{1}_m]_{a,b}\phi(z)&=\big([m+1]_{\mathcal{R}(p,q)}-[n]_{\mathcal{R}(p,q)}\big)\,G^{1}_{n+m}\phi(z)\nonumber\\&+\delta_{n+m+1,0}\,C^{R}_n(p,q),\label{d17b}\\
\,[G^{1}_n ,G^{1}_m]\phi(z)&=0,\label{d17c}
\end{align}
where $u,$ $v,$ $a,$ and $b$ are given by (\ref{d16}). 
\begin{remark}
	The $(p,q)$-deformed super generators
	$\mathcal{L}^1_n $ and $\mathcal{G}^1_n :$
	\begin{align*}
	\,\mathcal{L}^{1}_n\phi(z)&=-[z\partial_z+1]_{p,q}\,z^n\,\phi(z),\\
	\,\mathcal{G}^{1}_n\phi(z)&=-\theta\, [z\partial_z+1]_{p,q}\,z^n\,\phi(z)
	\end{align*}
	defined the superalgebra
	satisfying the following relations:
	\begin{align*}
	\,[\mathcal{L}^{1}_n , \mathcal{L}^{1}_m]_{\hat{u},\hat{v}}\phi(z)&=[n-m]_{p,q}q^{N_{1}-m}\mathcal{L}^{1}_{n+m}\phi(z),\\
	\,[\mathcal{L}^{1}_n , \mathcal{G}^{1}_m]_{\tilde{u},\tilde{v}}\phi(z)&=\big([m+1]_{p,q}-[n]_{p,q}\big)q^{N_{1}-m}\mathcal{G}^{1}_{n+m}\phi(z),\\
	\,[\mathcal{G}^{1}_n , \mathcal{G}^{1}_m]=0\nonumber,
	\end{align*}
	with 
	\begin{eqnarray*}
		\left \{
		\begin{array}{l}
			\displaystyle
			\hat{u}=\hat{\Lambda}_{nm}(p,q),\quad \tilde{u}=\tilde{\Lambda}_{nm}(p,q),\\
			\\
			\hat{v}=p^{m-n}\hat{\Lambda}_{nm}(p,q),\quad \tilde{v}=p^{m-n}\,\tilde{\Lambda}_{nm}(p,q)\\
			,
		\end{array}
		\right. 
	\end{eqnarray*}
	\begin{eqnarray*}
		\hat{\Lambda}_{nm}(p,q)=\frac{[n-m]_{p,q}}{[m-n]_{p,q}}\frac{f^{1}_{n+m}(p,q)}{f^{1}_n(p,q)f^{1}_m(p,q)},
	\end{eqnarray*}
	and 
	\begin{eqnarray*}
		\tilde{\Lambda}_{nm}(p,q)=\frac{[m+1]_{p,q}-[n]_{p,q}}{[m-n]_{p,q}}\frac{f^{1}_{n+m}(p,q)}{f^{1}_n(p,q)f^{1}_m(p,q)}.
	\end{eqnarray*}
	Now, we rewrite the $(p,q)$-deformed super generators in the following way:
	\begin{eqnarray*}
		L^{1}_n\phi(z) =q^{-N_1}\,\mathcal{L}^{1}_n\phi(z)\quad\mbox{and}\quad 
		G^{1}_n\phi(z)=q^{-N_1}\,\mathcal{G}^{1}_n\phi(z).
	\end{eqnarray*}
	Then, they obeys  the following commutation relations:
	\begin{align*}
	\,[L^1_n , L^1_m]_{u,v}\phi(z)&=[n-m]_{p,q}L^1_{n+m}\phi(z),\\
	\,[L^{1}_n ,G^{1}_m]_{a,b}\phi(z)&=\big([m+1]_{p,q}-[n]_{p,q}\big)\,G^{1}_{n+m}\phi(z),\\
	\,[G^{1}_n ,G^{1}_m]&=0,
	\end{align*}
	where
	\begin{align*}
	u&=q^{m-n}\,\hat{\Lambda}_{nm}(p,q),\quad
	v&=p^{m-n}\,\hat{\Lambda}_{mn}(p,q),\\
	a&=q^{m-n}\,\tilde{\Lambda}_{nm}(p,q),\quad
	b&=p^{m-n}\,\tilde{\Lambda}_{nm}(p,q), 
	\end{align*}
	equivalently,
	\begin{align*}
	\,[L^1_n , L^1_m]\phi(z)&=[n-m]_{p,q}\,p^{N_1-m}q^{-N_1+n}\,\hat{\varLambda}_{nm}(p,q)L^1_{n+m}\phi(z)\\
	\,[L^1_n , G^1_m]\phi(z)&=\big([m+1]_{p,q}-[n]_{p,q}\big)p^{N_1-m}q^{-N_1+n}\,\tilde{\varLambda}_{nm}(p,q)G^1_{n+m}\phi(z)\\
	\,[G^1_n , G^1_m]&=0.
	\end{align*}
	Furthermore, the $(p,q)$-super $su(1,1)$ subalgebra is generated by the following commutation relations:
	\begin{align*}
	\,[L^1_0 , L^1_1]_{u,v}\phi(z)&=[-1]_{p,q}L^1_{1}\phi(z)\\
	\,[L^{1}_0 ,G^{1}_1]_{a,b}\phi(z)&=\big([2]_{p,q}\big)\,G^{1}_{1}\phi(z)\\
	\,[G^{1}_0 ,G^{1}_1]&=0,
	\end{align*}
	where
	\begin{align*}
	u&=q\,\hat{\Lambda}_{01}(p,q),\quad
	v&=p\,\hat{\Lambda}_{10}(p,q),\\
	a&=q\,\tilde{\Lambda}_{01}(p,q),\quad
	b&=p\,\tilde{\Lambda}_{10}(p,q), 
	\end{align*}
	\begin{align*}
	\,[L^1_{-1} , L^1_0]_{u,v}\phi(z)&=[-1]_{p,q}L^1_{-1}\phi(z)\\
	\,[L^{1}_{-1} ,G^{1}_0]_{a,b}\phi(z)&=\big([1]_{p,q}-[-1]_{p,q}\big)\,G^{1}_{-1}\phi(z)\\
	\,[G^{1}_{-1} ,G^{1}_0]&=0,
	\end{align*}
	where
	\begin{align*}
	u&=q\,\hat{\Lambda}_{-10}(p,q),\quad
	v&=p\,\hat{\Lambda}_{0(-1)}(p,q),\\
	a&=q\,\tilde{\Lambda}_{-10}(p,q),\quad
	b&=p\,\tilde{\Lambda}_{0(-1)}(p,q), 
	\end{align*}
	and
	\begin{align*}
	\,[L^1_{-1} , L^1_1]\phi(z)&= [2]p^{N_1-1}q^{-N_1-1}\,\hat{\varLambda}_{-10}(p,q)L^1_{0}\phi(z)\\
	\,[L^1_{-1} , G^1_1]\phi(z)&=\big([2]_{p,q}-[-1]_{p,q}\big)p^{N_1-1}q^{-N_1-1}\,\tilde{\varLambda}_{-11}(p,q)G^1_{0}\phi(z)\\
	\,[G^1_{-1} , G^1_1]\phi(z)&=0.
	\end{align*}
\end{remark}
\section{ $\mathcal{R}(p,q)$-Virasoro $n$-algebra with conformal dimension $(\Delta\neq 0,1)$ }
In this section, we construct the quantum conformal  Virasoro $n$-algebra from the $\mathcal{R}(p,q)$-quantum algebra. Besides,  particular cases induced by quantum algebras in the literature are deduced.
\subsection{ $\mathcal{R}(p,q)$-conformal Witt $n$-algebra }
The  $\mathcal{R}(p,q)$-Witt $n$-algebra with conformal dimension $(\Delta\neq 0,1)$ is investigated. Moreover, several examples are presented. 
\begin{definition}\cite{WYLWZ}
	The L\'evi-Civit\'a symbol is defined by:
	\begin{eqnarray}\label{LCsymbol}
	\epsilon^{j_1 \cdots j_p}_{i_1 \cdots i_p}= det\left( \begin{array} {ccc}
	\delta^{j_1}_{i_1} &\cdots&  \delta^{j_1}_{i_p}   \\ 
	\vdots && \vdots \\
	\delta^{j_p}_{i_1} & \cdots& \delta^{j_p}_{i_p}
	\end {array} \right) .
	\end{eqnarray}
\end{definition}
Let us firstly introduced the $\mathcal{R}(p,q)$-conformal Witt $n$-algebra ($n$ even). The result is contained in the following lemma.
\begin{lemma}
	The $\mathcal{R}(p,q)$-Witt $n$-algebra with conformal dimension $(\Delta \neq 0,1)$ is given by:
	\begin{align}\label{nalg}
	\big[\mathcal{L}^{\Delta}_{m_1},\mathcal{L}^{\Delta}_{m_2},\ldots, \mathcal{L}^{\Delta}_{m_n}\big]&=\frac{\big(q-p\big)^{n-1\choose 2}}{ (\tau_1\tau_2)^{\lfloor \frac{n-1}{2} \rfloor \bar{m}}} \Bigg(\frac{[-2\bar{m}]_{\mathcal{R}(p,q)}}{2[-\bar{m}]_{{\mathcal R}(p,q)}}\Bigg)^{\beta}\nonumber\\&\times\prod_{j=1}^{n}\bigg(\tau_1^{\Delta(\bar{m}-1)}+ (-1)^n\tau_2^{\Delta(t\partial_t-\bar{m})}\bigg)\nonumber\\&\times\prod_{1\leq i < j\leq n}\Big([m_i]_{{\mathcal R}(p,q)}-[m_j]_{{\mathcal R}(p,q)}\Big)\mathcal{L}^{\Delta}_{\bar{m}},
	\end{align}
	where $\bar{m}=\sum_{l=1}^{n}m_l.$
\end{lemma}
\begin{proof} 
	The generalization of the generators of the conformal Witt algebra from the $\mathcal{R}(p,q)$-quantum algebras is first introduced by Hounkonnou and Melong   as follows \cite{HM}:
	\begin{eqnarray}\label{rpqop}
	\mathcal{L}^{\Delta}_{m}\phi(t)=[t\partial_{t}+ \Delta(m+1)-m]_{\mathcal{R}(p,q)}\,t^{m}\phi(t)
	\end{eqnarray}
	and
	recently,  the generators of the Witt algebra were constructed in the form \cite{HMM}:
	\begin{eqnarray}\label{oppg}
	l_{m}\phi(t)=[t\partial_{t}-m]_{\mathcal{R}(p,q)}\,t^{m}\phi(t).
	\end{eqnarray}
	By using the relation
	\begin{align}
	[u+v]_{\mathcal{R}(p,q)}&= \tau^{v}_1\,[u]_{\mathcal{R}(p,q)}+ \tau^{u}_2\,[v]_{\mathcal{R}(p,q)}\nonumber\\&=\tau^{v}_2\,[u]_{\mathcal{R}(p,q)}+ \tau^{u}_1\,[v]_{\mathcal{R}(p,q)}, 
	\end{align}
	where $\big(\tau_i\big)_{i\in\{1,2\}}$ are  parameters depending on $p$ and $q,$ 
	the operators (\ref{rpqop}) can be re-written in the simpler  form:
	\begin{eqnarray}
	\mathcal{L}^{\Delta}_{m}\phi(t)=l^{\Delta}_{m}\phi(t) + S^{\Delta}_m\,\phi(t),
	\end{eqnarray} 
	where 
	\begin{eqnarray}\label{opp1}
	l^{\Delta}_{m}\phi(t)=\tau^{\Delta(m+1)}_1\,[t\partial_t-m]_{\mathcal{R}(p,q)}\,t^m\phi(t)
	\end{eqnarray}
	and 
	\begin{eqnarray}\label{opp2}
	S^{\Delta}_m\phi(t)=\tau^{(t\partial_t-m)}_2\,[\Delta(m+1)]_{\mathcal{R}(p,q)}\,t^m\phi(t).
	\end{eqnarray}
	From our previous works \cite{HMM, melong2022} and Wang et {\it al} \cite{WYLWZ},  
	we can consider the  $n$-brackets $(n\geq 3)$  defined as follows:
	\begin{align}\label{rnb1}
	\big[l^{\Delta}_{m_1},l^{\Delta}_{m_2},\cdots, l^{\Delta}_{m_n}\big]&:=\bigg(\frac{[-2\bar{m}]_{\mathcal{R}(p,q)}}{2[-\bar{m}]_{\mathcal{R}(p,q)}}\bigg)^{\beta}\epsilon^{i_1i_2\cdots i_n}_{12\cdots n}\prod_{j=1}^{n}\tau_1^{\Delta(m_{i_j}-1)}\nonumber\\&\times(\tau_1\tau_2)^{\sum_{j=1}^{n}\big(\lfloor \frac{n}{2} \rfloor-j+1\big)m_{i_j}}l^{\Delta}_{m_{i_1}} \cdots
	l^{\Delta}_{m_{i_n}},
	\end{align}
	and 
	\begin{align}\label{rnb2}
	\big[S^{\Delta}_{m_1},S^{\Delta}_{m_2},\cdots, S^{\Delta}_{m_n}\big]&:=\bigg(\frac{[-2\bar{m}]_{\mathcal{R}(p,q)}}{2[-\bar{m}]_{\mathcal{R}(p,q)}}\bigg)^{\beta}\epsilon^{i_1i_2\cdots i_n}_{12\cdots n}\prod_{j=1}^{n}\tau_2^{\Delta(t\partial_t-m_{i_j})}\nonumber\\&\times(\tau_1\tau_2)^{\sum_{j=1}^{n}\big(\lfloor \frac{n}{2} \rfloor-j+1\big)m_{i_j}}S^{\Delta}_{m_{i_1}} \cdots
	S^{\Delta}_{m_{i_n}},
	\end{align}
	where $\beta=\frac{1+(-1)^n }{ 2},$  $\lfloor n \rfloor=\max\{m\in\mathbb{Z}\ m\leq n\}$ is the floor function. Putting the operators (\ref{opp1}) and \eqref{opp2} in the relation (\ref{rnb1}) and \eqref{rnb2}, we obtain the generators $\mathcal{L}^{\Delta}_{m}$ satisfying the $\mathcal{R}(p,q)$-conformal Witt $n$-algebra \eqref{nalg}.
\end{proof}
\begin{remark} Some examples of $\mathcal{R}(p,q)$-deformed conformal Witt $n$-algebra are deduced as:
	\begin{enumerate}
		\item[(i)]For $n=3,$ we obtain the $\mathcal{R}(p,q)$-conformal Witt $3$-algebra as follows:
		\begin{align*}
		\big[\mathcal{L}^{\Delta}_{m_1},\mathcal{L}^{\Delta}_{m_2}, \mathcal{L}^{\Delta}_{m_3}\big]&=\frac{q-p}{ (\tau_1\tau_2)^{m_1+m_2+m_3}} \nonumber\\&\times\prod_{j=1}^{3}\bigg(\tau_1^{\Delta(m_1+m_2+m_3-1)} -\tau_2^{\Delta(t\partial_t-m_1-m_2-m_3)}\bigg)\nonumber\\&\times\prod_{1\leq i < j\leq 3}\Big([m_i]_{{\mathcal R}(p,q)}-[m_j]_{{\mathcal R}(p,q)}\Big)\mathcal{L}^{\Delta}_{m_1+m_2+m_3}.
		\end{align*}
		\item[(ii)]The $\mathcal{R}(p,q)$-conformal Witt $4$-algebra is deduced by taking $n=4:$
		\begin{align*}
		\big[\mathcal{L}^{\Delta}_{m_1},\mathcal{L}^{\Delta}_{m_2},\mathcal{L}^{\Delta}_{m_3}, \mathcal{L}^{\Delta}_{m_4}\big]&=\frac{\big(q-p\big)^{3}}{ (\tau_1\tau_2)^{\bar{m}}} \Bigg(\frac{[-2\bar{m}]_{\mathcal{R}(p,q)}}{2[-\bar{m}]_{{\mathcal R}(p,q)}}\Bigg)\nonumber\\&\times\prod_{j=1}^{4}\bigg(\tau_1^{\Delta(\bar{m}-1)}+ \tau_2^{\Delta(t\partial_t-\bar{m})}\bigg)\nonumber\\&\times\prod_{1\leq i < j\leq 4}\Big([m_i]_{{\mathcal R}(p,q)}-[m_j]_{{\mathcal R}(p,q)}\Big)\mathcal{L}^{\Delta}_{\bar{m}},
		\end{align*}
		where $\bar{m}=\sum_{l=1}^{4}m_l.$
		\item[(iii)] In the limit $p\longrightarrow q,$ the relation (\ref{nalg}) is reduced to the null $n$-algebra in the case $n\geq 3.$
	\end{enumerate} 
\end{remark}
We can prove that the relation \eqref{nalg} satisfies the conformal
generalized Jacobi identity (GJI)
\begin{eqnarray*}
	\epsilon^{i_{1},\ldots,i_{2n-1}}_{n_{1},\ldots,n_{2n-1}}\big[\big[\mathcal{L}^{\Delta}_{i_1},\ldots,\mathcal{L}^{\Delta}_{i_{2n-1}}\big]_{\mathcal{R}(p,q)},\mathcal{L}^{\Delta}_{i_{n+1}},\ldots, \mathcal{L}^{\Delta}_{i_{2n-1}}\big]_{\mathcal{R}(p,q)}=0,
\end{eqnarray*}
and the skewsymmetry 
hold for \eqref{nalg}. Thus,  the  $\mathcal{R}(p,q)$-conformal Witt $n$-algebra \eqref{nalg} is a higher order Lie algebra.

Now, we investigate the $\mathcal{R}(p,q)$-Virasoro $2n$-algebra $(n\geq 2)$ with conformal dimension $(\Delta \neq 0,1).$ It's the centrally extended of the $\mathcal{R}(p,q)$-conformal Witt $n$-algebra \eqref{nalg}. Then, 
\begin{proposition}
	For $(\Delta \neq 0,1),$ the $\mathcal{R}(p,q)$-conformal Virasoro $2n$-algebra ($n$ even) is determined by the following commutation relation:
	\begin{eqnarray}\label{cnalg}
	\big[L^{\Delta}_{m_1},\ldots, L^{\Delta}_{m_{2n}}\big]=f^{\Delta}_{\mathcal{R}(p,q)}(m_1,\ldots,m_{2n})L^{\Delta}_{\tilde{m}}+ \tilde{C}^{\Delta}_{\mathcal{R}(p,q)}(m_1,\ldots,m_{2n}),
	\end{eqnarray}
	where
	\begin{align}\label{cnalg1}
	f^{\Delta}_{\mathcal{R}(p,q)}(m_1,\ldots,m_{2n})&=\frac{\big(q-p\big)^{2n-1\choose 2}}{ (\tau_1\tau_2)^{\lfloor \frac{2n-1}{2} \rfloor \tilde{m}}} \Bigg(\frac{[-2\tilde{m}]_{\mathcal{R}(p,q)}}{2[-\tilde{m}]_{{\mathcal R}(p,q)}}\Bigg)^{\beta}\prod_{j=1}^{2n}\bigg(\tau_1^{\Delta(\tilde{m}-1)}\nonumber\\&+ \tau_2^{\Delta(t\partial_t-\tilde{m})}\bigg)\prod_{1\leq i < j\leq 2n}\Big([m_i]_{{\mathcal R}(p,q)}-[m_j]_{{\mathcal R}(p,q)}\Big),
	\end{align}
	$\tilde{m}=\sum_{l=1}^{2n}\,m_l$ and
	\begin{align}\label{cv}
	\tilde{C}^{\Delta}_{\mathcal{R}(p,q)}(m_1,&\cdots,m_{2n})={c(p,q)\epsilon^{i_1\cdots i_{2n}}_{1\cdots 2n}\over 6\times  2^n\times n!}\prod_{l=1}^{n}\Big(\frac{q}{p}\Big)^{i\,N_{\Delta/2}}{[m_{i_{2l-1}}-1]_{{\mathcal R}(p,q)}\over \big(\tau_1\tau_2\big)^{m_{2l-1}}}\nonumber\\&\times{[m_{2l-1}]_{{\mathcal R}(p,q)}\over [2m_{2l-1}]_{{\mathcal R}(p,q)}} [m_{i_{2l-1}}]_{{\mathcal R}(p,q)}[m_{i_{2l-1}}+1]_{{\mathcal R}(p,q)}
	\delta_{m_{i_{2l-1}}+ m_{i_{2l}},0}.
	\end{align}
\end{proposition}
\begin{proof}  For the relation \eqref{nalg}, we deduce \eqref{cnalg1}. Moreover, the conformal central extention 
	$\tilde{C}^{\Delta}_{\mathcal{R}(p,q)}(m_1,\ldots,m_{2n})$ is provided by the factorization
	\begin{eqnarray*}
		\tilde{C}^{\Delta}_{\mathcal{R}(p,q)}(m_1,\ldots,m_{2n})=\prod_{i=1}^{n}\Gamma^{i}(N_{\Delta})\,C^{\Delta}_{\mathcal{R}(p,q)}(m_1,\ldots,m_{2n}),
	\end{eqnarray*}
	where $C^{\Delta}_{\mathcal{R}(p,q)}(m_1,\ldots,m_{2n})$ is given in \cite{melong2022}: 
	\begin{align*}
	C^{\Delta}_{\mathcal{R}(p,q)}(m_1,\cdots,m_{2n})&={c(p,q)\epsilon^{i_1\cdots i_{2n}}_{1\cdots 2n}\over 6\times  2^n\times n!}\prod_{l=1}^{n}{[m_{i_{2l-1}}-1]_{{\mathcal R}(p,q)}\over \big(\tau_1\tau_2\big)^{m_{2l-1}}}{[m_{2l-1}]_{{\mathcal R}(p,q)}\over [2m_{2l-1}]_{{\mathcal R}(p,q)}}\nonumber\\&\times [m_{i_{2l-1}}]_{{\mathcal R}(p,q)}[m_{i_{2l-1}}+1]_{{\mathcal R}(p,q)}
	\delta_{m_{i_{2l-1}}+ m_{i_{2l}},0}
	\end{align*}
	and $\Gamma^{i}(N_{\Delta})$ by \cite{CJ,HM}:
	\begin{eqnarray*}
		\Gamma(N_{\Delta})=\Big(\frac{q}{p}\Big)^{N_{\Delta/2}}.
	\end{eqnarray*}
	Thus, the relation \eqref{ce1} follows. 
	Finally, the $\mathcal{R}(p,q)$-conformal Virasoro $2n$-algebra ($n$-even) is generated by the relations \eqref{cnalg},\eqref{cnalg1}, and \eqref{cv}.	
\end{proof}
\begin{remark}
	The  particular cases of  the conformal  Virasoro $2n$-algebra ($n$ even) from the quantum algebras existing in the literature are deuced as:
	\begin{enumerate}	
		\item[(a)] The conformal Virasoro $2n$-algebra ($n$ even) associated to the {\bf Biedenharn-Macfarlane algebra} \cite{BC,M}:
		\begin{eqnarray}\label{BMcnalg}
		\big[L^{\Delta}_{m_1},\ldots, L^{\Delta}_{m_{2n}}\big]=f^{\Delta}_{q}(m_1,\ldots,m_{2n})L^{\Delta}_{\tilde{m}}+ \tilde{C}^{\Delta}_{q}(m_1,\ldots,m_{2n}),
		\end{eqnarray}
		where
		\begin{align}\label{BMcnalg1}
		f^{\Delta}_{q}(m_1,\ldots,m_{2n})&=\big(q^{-1}-q\big)^{2n-1\choose 2} \Bigg(\frac{[-2\tilde{m}]_{q}}{2[-\tilde{m}]_{q}}\Bigg)^{\beta}\prod_{j=1}^{2n}\bigg(q^{\Delta(\tilde{m}-1)}\nonumber\\&+ q^{-\Delta(t\partial_t-\tilde{m})}\bigg)\prod_{1\leq i < j\leq 2n}\Big([m_i]_{q}-[m_j]_{q}\Big),
		\end{align}
		$\tilde{m}=\sum_{l=1}^{2n}\,m_l$ and
		\begin{align}\label{BMcv}
		\tilde{C}^{\Delta}_{q}(m_1,\cdots,m_{2n}&)={c(q)\epsilon^{i_1\cdots i_{2n}}_{1\cdots 2n}\over 6\times  2^n\times n!}\prod_{l=1}^{n}q^{-2i\,N_{\Delta/2}}{[m_{2l-1}]_{q}\over [2m_{2l-1}]_{q}}\nonumber\\&\times [m_{i_{2l-1}}-1]_{q}[m_{i_{2l-1}}]_{q}[m_{i_{2l-1}}+1]_{q}
		\delta_{m_{i_{2l-1}}+ m_{i_{2l}},0}.
		\end{align}
		\item[(b)] The conformal Virasoro $2n$-algebra ($n$ even) from the \textbf{Jagannathan-Srinivasa algebra} \cite{JS}:
		\begin{eqnarray}\label{JScnalg}
		\big[L^{\Delta}_{m_1},\ldots, L^{\Delta}_{m_{2n}}\big]=f^{\Delta}_{p,q}(m_1,\ldots,m_{2n})L^{\Delta}_{\tilde{m}}+ \tilde{C}^{\Delta}_{p,q}(m_1,\ldots,m_{2n}),
		\end{eqnarray}
		where
		\begin{align}\label{JScnalg1}
		f^{\Delta}_{p,q}(m_1,\ldots,m_{2n})&=\frac{\big(q-p\big)^{2n-1\choose 2}}{ (p\,q)^{\lfloor \frac{2n-1}{2} \rfloor \tilde{m}}} \Bigg(\frac{[-2\tilde{m}]_{p,q}}{2[-\tilde{m}]_{p,q}}\Bigg)^{\beta}\prod_{j=1}^{2n}\bigg(p^{\Delta(\tilde{m}-1)}\nonumber\\&+ q^{\Delta(t\partial_t-\tilde{m})}\bigg)\prod_{1\leq i < j\leq 2n}\Big([m_i]_{p,q}-[m_j]_{p,q}\Big),
		\end{align}
		$\tilde{m}=\sum_{l=1}^{2n}\,m_l$ and
		\begin{align}\label{JScv}
		\tilde{C}^{\Delta}_{p,q}(m_1,\cdots,m_{2n}&)={c(p,q)\epsilon^{i_1\cdots i_{2n}}_{1\cdots 2n}\over 6\times  2^n\times n!}\prod_{l=1}^{n}\Big(\frac{q}{p}\Big)^{i\,N_{\Delta/2}}{[m_{i_{2l-1}}-1]_{p,q}\over \big(p\,q\big)^{m_{2l-1}}}\nonumber\\&\times{[m_{2l-1}]_{p,q}\over [2m_{2l-1}]_{p,q}} [m_{i_{2l-1}}]_{p,q}[m_{i_{2l-1}}+1]_{p,q}
		\delta_{m_{i_{2l-1}}+ m_{i_{2l}},0}.
		\end{align}
		\item[(c)] The conformal  Virasoro $2n$-algebra ($n$ even) induced by the  \textbf{Chakrabarti- Jagannathan algebra}\cite{Chakrabarti&Jagan}: 
		\begin{eqnarray}\label{CJcnalg}
		\big[L^{\Delta}_{m_1},\ldots, L^{\Delta}_{m_{2n}}\big]=f^{\Delta}_{p^{-1},q}(m_1,\ldots,m_{2n})L^{\Delta}_{\tilde{m}}+ \tilde{C}^{\Delta}_{p^{-1},q}(m_1,\ldots,m_{2n}),
		\end{eqnarray}
		where
		\begin{align}\label{CJcnalg1}
		f^{\Delta}_{p^{-1},q}(m_1,\ldots,m_{2n})&=\frac{\big(q-p\big)^{2n-1\choose 2}}{ (\frac{q}{p})^{\lfloor \frac{2n-1}{2} \rfloor \tilde{m}}} \Bigg(\frac{[-2\tilde{m}]_{p^{-1},q}}{2[-\tilde{m}]_{p^{-1},q}}\Bigg)^{\beta}\prod_{j=1}^{2n}\bigg(p^{-\Delta(\tilde{m}-1)}\nonumber\\&+ q^{\Delta(t\partial_t-\tilde{m})}\bigg)\prod_{1\leq i < j\leq 2n}\Big([m_i]_{p^{-1},q}-[m_j]_{p^{-1},q}\Big),
		\end{align}
		$\tilde{m}=\sum_{l=1}^{2n}\,m_l$ and
		\begin{align}\label{CJcv}
		\tilde{C}^{\Delta}_{p^{-1},q}(m_1&,\cdots,m_{2n})={c(p^{-1},q)\epsilon^{i_1\cdots i_{2n}}_{1\cdots 2n}\over 6\times  2^n\times n!}\prod_{l=1}^{n}\big(q\,p\big)^{i\,N_{\Delta/2}}{[m_{i_{2l-1}}-1]_{p^{-1},q}\over \big(\frac{q}{p}\big)^{m_{2l-1}}}\nonumber\\&\times{[m_{2l-1}]_{p^{-1},q}\over [2m_{2l-1}]_{p^{-1},q}} [m_{i_{2l-1}}]_{p^{-1},q}[m_{i_{2l-1}}+1]_{p^{-1},q}
		\delta_{m_{i_{2l-1}}+ m_{i_{2l}},0}.
		\end{align}
		\item The conformal super Virasoro $2n$-algebra ($n$ even) generated by the  \textbf{Generalized $q$-Quesne algebra} \cite{HN}:
		\begin{eqnarray}\label{HNcnalg}
		\big[L^{\Delta}_{m_1},\ldots, L^{\Delta}_{m_{2n}}\big]=f^{\Delta}_{p,q}(m_1,\ldots,m_{2n})L^{\Delta}_{\tilde{m}}+ \tilde{C}^{\Delta}_{p,q}(m_1,\ldots,m_{2n}),
		\end{eqnarray}
		where
		\begin{align}\label{HNcnalg1}
		f^{\Delta}_{p,q}(m_1,\ldots,m_{2n})&=\frac{\big(q^{-1}-p\big)^{2n-1\choose 2}}{ (\frac{p}{q})^{\lfloor \frac{2n-1}{2} \rfloor \tilde{m}}} \Bigg(\frac{[-2\tilde{m}]^Q_{p,q}}{2[-\tilde{m}]^Q_{p,q}}\Bigg)^{\beta}\prod_{j=1}^{2n}\bigg(p^{\Delta(\tilde{m}-1)}\nonumber\\&+ q^{-\Delta(t\partial_t-\tilde{m})}\bigg)\prod_{1\leq i < j\leq 2n}\Big([m_i]^Q_{p,q}-[m_j]^Q_{p,q}\Big),
		\end{align}
		$\tilde{m}=\sum_{l=1}^{2n}\,m_l$ and
		\begin{align}\label{HNcv}
		\tilde{C}^{\Delta}_{p,q}(m_1,\cdots,m_{2n}&)={c(p,q)\epsilon^{i_1\cdots i_{2n}}_{1\cdots 2n}\over 6\times  2^n\times n!}\prod_{l=1}^{n}(q\,p)^{-i\,N_{\Delta/2}}{[m_{i_{2l-1}}-1]^Q_{p,q}\over \big(\frac{p}{q}\big)^{m_{2l-1}}}\nonumber\\&\times{[m_{2l-1}]^Q_{p,q}\over [2m_{2l-1}]^Q_{p,q}} [m_{i_{2l-1}}]^Q_{p,q}[m_{i_{2l-1}}+1]^Q_{p,q}
		\delta_{m_{i_{2l-1}}+ m_{i_{2l}},0}.
		\end{align}
		\item The conformal super Virasoro algebra generated by the \textbf{Hounkonnou-Ngompe generalized algebra} \cite{HNN}:
		\begin{eqnarray}\label{HNNcnalg}
		\big[L^{\Delta}_{m_1},\ldots, L^{\Delta}_{m_{2n}}\big]=f^{\Delta}_{p,q}(m_1,\ldots,m_{2n})L^{\Delta}_{\tilde{m}}+ \tilde{C}^{\Delta}_{p,q}(m_1,\ldots,m_{2n}),
		\end{eqnarray}
		where
		\begin{align}\label{HNNcnalg1}
		f^{\Delta}_{p,q}(m_1,\ldots,m_{2n})&=\frac{\big(q-p\big)^{2n-1\choose 2}}{ (p\,q)^{\lfloor \frac{2n-1}{2} \rfloor \tilde{m}}} \Bigg(\frac{[-2\tilde{m}]^{\mu,\nu}_{p,q,g}}{2[-\tilde{m}]^{\mu,\nu}_{p,q,g}}\Bigg)^{\beta}\prod_{j=1}^{2n}\bigg(p^{\Delta(\tilde{m}-1)}\nonumber\\&+ q^{\Delta(t\partial_t-\tilde{m})}\bigg)\prod_{1\leq i < j\leq 2n}\Big([m_i]^{\mu,\nu}_{p,q,g}-[m_j]^{\mu,\nu}_{p,q,g}\Big),
		\end{align}
		$\tilde{m}=\sum_{l=1}^{2n}\,m_l$ and
		\begin{align}\label{HNNcv}
		\tilde{C}^{\Delta}_{p,q}(m_1,&\cdots,m_{2n})={c(p,q)\epsilon^{i_1\cdots i_{2n}}_{1\cdots 2n}\over 6\times  2^n\times n!}\prod_{l=1}^{n}\Big(\frac{q}{p}\Big)^{i\,N_{\Delta/2}}{[m_{i_{2l-1}}-1]^{\mu,\nu}_{p,q,g}\over \big(p\,q\big)^{m_{2l-1}}}\nonumber\\&\times{[m_{2l-1}]^{\mu,\nu}_{p,q,g}\over [2m_{2l-1}]^{\mu,\nu}_{p,q,g}} [m_{i_{2l-1}}]^{\mu,\nu}_{p,q,g}[m_{i_{2l-1}}+1]^{\mu,\nu}_{p,q,g}
		\delta_{m_{i_{2l-1}}+ m_{i_{2l}},0}.
		\end{align}
	\end{enumerate}
\end{remark}
\subsection{Another $\mathcal{R}(p,q)$-conformal Witt $n$-algebra}
In this section,  we determine another $\mathcal{R}(p,q)$-conformal Witt $n$-algebra. Some examples for a  value of $n$ are deduced and particular cases induced from quantum algebras are derived.
\begin{definition}
	For $(\Delta \neq 0,1),$ the $\mathcal{R}(p,q)$-conformal operator is defined by:
	\begin{eqnarray}\label{rpqop1}
	{\mathcal T}^{a\Delta}_m:=x^{(1-\Delta)(m+1)}{\mathcal D}_{{\mathcal R}(p^{a},q^{a})}\,x^{\Delta(m+1)},
	\end{eqnarray}
\end{definition}
where ${\mathcal D}_{{\mathcal R}(p^{a},q^{a})}$ is the ${\mathcal R}(p,q)$-derivative presented by:
\begin{eqnarray*}
	{\mathcal D}_{{\mathcal R}(p^{a},q^{a})}\big(\phi(z)\big)={p^{a}-q^{a}\over p^{a\,P}-q^{a\,Q}}{\mathcal R}(p^{a\,P},q^{a\,Q})\,\mathcal{D}_{p^a,q^a}\phi(z).
\end{eqnarray*}
Then, after computation, the relation \eqref{rpqop1} is reduced as:
\begin{eqnarray}\label{rpqopa}
{\mathcal T}^{a\Delta}_m=[\Delta(m+1)]_{{\mathcal R}(p^{a},q^{a})}\,z^{m}.
\end{eqnarray}
\begin{proposition}
	The $\mathcal{R}(p,q)$-conformal operators (\ref{rpqop1}) satisfy the product relation given by: \begin{align}\label{pre}
	\mathcal{T}^{a\Delta}_{m}.\mathcal{T}^{b\Delta}_n&={\big(\tau^{a+b}_1-\tau^{a+b}_2\big)\tau^{an(1-\Delta)}_1\over \big(\tau^{a}_1-\tau^{a}_2\big)\big(\tau^{b}_1-\tau^{b}_2\big)\tau^{b\Delta\,m}_1}\,\mathcal{T}^{(a+b)\Delta}_{m+n}+ f^{\Delta(a,b)}_{\mathcal{R}(p,q)}(m,n)\nonumber\\& - {\tau^{a(\Delta(m+1)+n)}_2\over \tau^{b\Delta\,m}_1\big(\tau^{a}_1-\tau^{a}_2\big)}\,\mathcal{T}^{b\Delta}_{m+n}- {\tau^{b\Delta(n+1)}_2\over \tau^{an(\Delta-1)}_1\big(\tau^{b}_1-\tau^{b}_2\big)}\, \mathcal{T}^{a\Delta}_{m+n},
	\end{align}
	and the commutation relation
	\begin{align}\label{crto}
	\Big[\mathcal{T}^{a\Delta}_{m}, \mathcal{T}^{b\Delta}_n\Big]&={\big(\tau^{a+b}_1-\tau^{a+b}_2\big)\over \big(\tau^{a}_1-\tau^{a}_2\big)\big(\tau^{b}_1-\tau^{b}_2\big)}\bigg(\frac{\tau^{an(1-\Delta)}_1}{\tau^{b\Delta\,m}_1}-\frac{\tau^{bm(1-\Delta)}_1}{\tau^{a\Delta\,n}_1}\bigg)\mathcal{T}^{(a+b)\Delta}_{m+n}\nonumber\\ &-\frac{\tau^{b\Delta(n+1)}_2}{\tau^{a\Delta\,n}_1}\frac{\big(\tau^{an}_1-\tau^{bm}_2\big)}{\big(\tau^{b}_1-\tau^{b}_2\big)}\mathcal{T}^{a\Delta}_{m+n}+h^{\Delta(a,b)}_{\mathcal{R}(p,q)}(m,n)\nonumber\\ &- \frac{\tau^{a\Delta(m+1)}_2}{\tau^{b\Delta\,m}_1}\frac{\big(\tau^{bm}_1-\tau^{an}_2\big)}{\big(\tau^{a}_1-\tau^{a}_2\big)}\mathcal{T}^{b\Delta}_{m+n},
	\end{align}
	where $f^{\Delta(a,b)}_{\mathcal{R}(p,q)}(m,n)$ and $h^{\Delta(a,b)}_{\mathcal{R}(p,q)}(m,n)$ are given by the relations \eqref{crtof} and \eqref{crtoh}.
\end{proposition}
\begin{proof} From the relation \eqref{rpqop1} and after computation,  we have
	\begin{eqnarray*}
		\mathcal{T}^{a\Delta}_{m}.\mathcal{T}^{b\Delta}_n=[\Delta(n+1)]_{{\mathcal R}(p^{b},q^{b})}[\Delta(m+1)+n]_{{\mathcal R}(p^{a},q^{a})}x^{m+n}.
	\end{eqnarray*}
	Using the $\mathcal{R}(p,q)$-numbers\cite{HMM}:
	\begin{eqnarray}\label{rpqn}
	[n]_{\mathcal{R}(p,q)}=\frac{\tau^n_1-\tau^n_2}{\tau_1-\tau_2},\quad \tau_1\neq \tau_2,
	\end{eqnarray}
	the relation \eqref{pre} and \eqref{crto} follows
	by  straightforward computation with
	\begin{eqnarray}\label{crtof}
	f^{\Delta(a,b)}_{\mathcal{R}(p,q)}(m,n)=\tau^{(a+b)\Delta(m+n+1)}_2\,[-\Delta\,m]_{\mathcal{R}(p^{b},q^{b})}\,[n(1-\Delta)]_{\mathcal{R}(p^{a},q^{a})},
	\end{eqnarray}
	\begin{eqnarray*}
		f^{\Delta(b,a)}_{\mathcal{R}(p,q)}(n,m)=\tau^{(a+b)\Delta(m+n+1)}_2\,[-\Delta\,n]_{\mathcal{R}(p^{a},q^{a})}\,[m(1-\Delta)]_{\mathcal{R}(p^{b},q^{b})}.
	\end{eqnarray*}	
	and
	\begin{align}\label{crtoh}
	h^{\Delta(a,b)}_{\mathcal{R}(p,q)}(m,n)&=f^{\Delta(a,b)}_{\mathcal{R}(p,q)}(m,n)-f^{\Delta(b,a)}_{\mathcal{R}(p,q)}(n,m)\nonumber\\&=\tau^{(a+b)\Delta(m+n+1)}_2\bigg([-\Delta\,m]_{\mathcal{R}(p^{b},q^{b})}\,[n(1-\Delta)]_{\mathcal{R}(p^{a},q^{a})}\nonumber\\&-[-\Delta\,n]_{\mathcal{R}(p^{a},q^{a})}\,[m(1-\Delta)]_{\mathcal{R}(p^{b},q^{b})}\bigg).
	\end{align}
\end{proof}
\begin{remark} 
	Putting $a=b=1,$ in the relation \eqref{crto}, we obtain the commutation relation:
	\begin{align*}
	\Big[\mathcal{T}^{\Delta}_{m}, \mathcal{T}^{\Delta}_n\Big]&={\tau^{-\Delta(n+m)}_1\big(\tau^{n}_1-\tau^{m}_1\big)\over \big(\tau_1-\tau_2\big)}[2]_{\mathcal{R}(p,q)}\mathcal{T}^{2\Delta}_{m+n}+h^{\Delta}_{\mathcal{R}(p,q)}(m,n)\nonumber\\
	&-\bigg(\frac{\tau^{\Delta(n+1)}_2\big(\tau^{n}_1-\tau^{m}_2\big)}{\tau^{\Delta\,n}_1\big(\tau_1-\tau_2\big)}+ \frac{\tau^{\Delta(m+1)}_2\big(\tau^{m}_1-\tau^{n}_2\big)}{\tau^{\Delta\,m}_1\big(\tau_1-\tau_2\big)}\bigg)\mathcal{T}^{\Delta}_{m+n},
	\end{align*}
	where
	\begin{align*}
	h^{\Delta(1,1)}_{\mathcal{R}(p,q)}(m,n)&=\tau^{2\Delta(m+n+1)}_2\bigg([-\Delta\,m]_{\mathcal{R}(p,q)}\,[n(1-\Delta)]_{\mathcal{R}(p,q)}\nonumber\\&-[-\Delta\,n]_{\mathcal{R}(p,q)}\,[m(1-\Delta)]_{\mathcal{R}(p,q)}\bigg).
	\end{align*}
\end{remark}
Note that for $\Delta=1,$ we recovered the $\mathcal{R}(p,q)$-deformed Witt algebra given in \cite{HMM}.
\begin{corollary}
	The $q$-conformal operators
	\begin{eqnarray*}\label{qopa}
		{\mathcal T}^{a\Delta}_m=[\Delta(m+1)]_{q^{a}}\,z^{m}.
	\end{eqnarray*}
	obeys  the commutation relation
	\begin{align*}
	\Big[\mathcal{T}^{a\Delta}_{m}, \mathcal{T}^{b\Delta}_n\Big]&={\big(q^{a+b}-1\big)\over \big(q^{a}-1\big)\big(q^{b}-1\big)}\bigg(\frac{q^{an(1-\Delta)}}{q^{b\Delta\,m}}-\frac{q^{bm(1-\Delta)}}{q^{a\Delta\,n}}\bigg)\mathcal{T}^{(a+b)\Delta}_{m+n}\nonumber\\ &-	\frac{1}{q^{a\Delta\,n}}\frac{\big(q^{an}-1\big)}{\big(q^{b}-1\big)}\mathcal{T}^{a\Delta}_{m+n}- \frac{1}{q^{b\Delta\,m}}\frac{\big(q^{bm}-1\big)}{\big(q^{a}-1\big)}\mathcal{T}^{b\Delta}_{m+n}\nonumber\\&+h^{\Delta(a,b)}_{q}(m,n),
	\end{align*}
	where
	\begin{eqnarray*}
		h^{\Delta(a,b)}_{q}(m,n)=\bigg([-\Delta\,m]_{q^{b}}[n(1-\Delta)]_{q^{a}}-[-\Delta\,n]_{q^{a}}[m(1-\Delta)]_{q^{b}}\bigg).
	\end{eqnarray*}
\end{corollary}

%\begin{remark} Some particular cases are derived as follows:
%	\begin{enumerate}
%		\item [(a)] 
%		\item[(b)]Taking $\mathcal{R}(x,1)=\frac{x-1}{q-1},$ we obtain the $q$-conformal operator and the commutation relation:
%		\begin{eqnarray*}
%			{\mathcal T}^{a\Delta}_m=\frac{q^{\Delta(m+1)}-1}{q-1},\quad q\neq 1
%		\end{eqnarray*}
%		and 
%		\begin{align*}
%		\Big[\mathcal{T}^{\Delta}_{m}, \mathcal{T}^{\Delta}_n\Big]&={\big(q^{n-\Delta(n+m)}-q^{m-\Delta(n+m)}\big)\over \big(q-1\big)}[2]_{\mathcal{R}(p,q)}\mathcal{T}^{2\Delta}_{m+n}\nonumber\\
%		&+{\big(q^{-\Delta\,n}-q^{-\Delta\,m}\big)\over q-1}\mathcal{T}^{\Delta}_{m+n}+h^{\Delta}_{q}(m,n),
%		\end{align*}
%		where
%		\begin{align*}
%		h^{\Delta}_{q}(m,n)=
%		\end{align*}
%		\item[(c)]The $(p,q)$-conformal operator and the commutation relation:
%		\begin{eqnarray*}
%			{\mathcal T}^{a\Delta}_m=\frac{p^{\Delta(m+1)}-q^{\Delta(m+1)}}{p-q},
%		\end{eqnarray*}
%		and \begin{align*}
%		\Big[\mathcal{T}^{\Delta}_{m}, \mathcal{T}^{\Delta}_n\Big]&={\big(p^{n-\Delta(n+m)}-p^{m-\Delta(n+m)}\big)\over \big(p-q\big)}[2]_{p,q}\mathcal{T}^{2\Delta}_{m+n}\nonumber\\
%		&+\bigg({\big(q^{(\Delta(n+1)+m)}p^{-\Delta\,n}-q^{(\Delta(m+1)+n)}p^{-\Delta\,m}\big)\over p-q}\\&+ \frac{\big(q^{\Delta(m+1)}-q^{\Delta(n+1)}\big)}{p-q}\bigg)\mathcal{T}^{\Delta}_{m+n}+h^{\Delta}_{p,q}(m,n),
%		\end{align*}
%		where
%		\begin{align*}
%		h^{\Delta}_{p,q}(m,n)=
%		\end{align*} 
%	\end{enumerate}
%\end{remark}
\begin{definition}
	The $n$-bracket is defined by:
	\begin{eqnarray*}
		\Big[{\mathcal T}^{a_1\Delta}_{m_1},\cdots,{\mathcal T}^{a_n\Delta}_{m_n}
		\Big]:=\epsilon^{i_1 \cdots i_n}_{1 \cdots n}\,{\mathcal T}^{a_{i_1}\Delta}_{m_{i_1}} \cdots {\mathcal T}^{a_{i_n}\Delta}_{m_{i_n}},
	\end{eqnarray*}
	where $\epsilon^{i_1 \cdots i_n}_{1 \cdots n}$ is the L\'evi-Civit\'a \eqref{LCsymbol}.
\end{definition}
Our interest is	focussed on the special  case with the same $a\Delta.$ Then, 
\begin{eqnarray*}
	\Big[{\mathcal T}^{a\Delta}_{m_1},\cdots,{\mathcal T}^{a\Delta}_{m_n}
	\Big]=\epsilon^{1\cdots n}_{1\cdots n}\,{\mathcal T}^{a\Delta}_{m_1} \cdots {\mathcal T}^{a\Delta}_{m_n}.
\end{eqnarray*}
Putting $a=b$ in the relation (\ref{crto}),  the following commutation relation holds:
\begin{eqnarray*}\label{crtoa}
	\Big[\mathcal{T}^{a\Delta}_{m}, \mathcal{T}^{a\Delta}_n\Big]&=&{\tau^{-a\Delta(n+m)}_1\big(\tau^{an}_1-\tau^{am}_1\big)\over \tau^{a}_1-\tau^{a}_2}[2]_{\mathcal{R}(p^a,q^a)}\mathcal{T}^{2a\Delta}_{m+n}\nonumber\\ &-&\frac{1}{\big(\tau_1-\tau_2\big)}\bigg(\frac{\tau^{a\Delta(n+1)}_2}{\tau^{a\Delta\,n}_1}\big(\tau^{an}_1-\tau^{am}_2\big)\nonumber\\&-& \frac{\tau^{a\Delta(m+1)}_2}{\tau^{a\Delta\,m}_1}\big(\tau^{am}_1-\tau^{an}_2\big)\bigg)\mathcal{T}^{a\Delta}_{m+n}+h^{\Delta(a)}_{\mathcal{R}(p,q)}(m,n),
\end{eqnarray*}
where \begin{align*}
h^{\Delta(a)}_{\mathcal{R}(p,q)}(m,n)&=\tau^{2a\Delta(m+n+1)}_2\bigg([-\Delta\,m]_{{\mathcal R}(p^{a},q^{a})}[n(1-\Delta)]_{{\mathcal R}(p^{a},q^{a})}\nonumber\\&-[-\Delta\,n]_{{\mathcal R}(p^{a},q^{a})}[m(1-\Delta)]_{{\mathcal R}(p^{a},q^{a})}\bigg).
\end{align*}
\begin{proposition}\label{propcWitt} The $\mathcal{R}(p,q)$-conformal Witt $n$-algebra is given by the following commutation relation:
	\begin{align}\label{crna}
	\Big[{\mathcal T}^{a\Delta}_{m_1},\cdots,{\mathcal T}^{a\Delta}_{m_n}
	\Big]&={(-1)^{n+1}\over \big(\tau^{a}_1-\tau^{a}_2\big)^{n-1}}\Big( V^n_{a\Delta}[n]_{{\mathcal R}(p^{a},q^{a})}{\mathcal T}^{n\,a\Delta}_{\bar{m}}\nonumber\\ &- [n-1]_{{\mathcal R}(p^{a},q^{a})}\big(D^n_{a\Delta}+ N^n_{a\Delta}\big){\mathcal T}^{a(n-1)\Delta}_{\bar{m}}\Big)\nonumber\\&+h^{\Delta}_{\mathcal{R}(p,q)}(m_1,\ldots,m_n),
	\end{align}
	where 
	\begin{align}\label{vn}
	V^n_{a\Delta}&= \tau^{a(n-1)(1-\Delta )\bar{m}}_1\Big(\big(\tau^{a}_1-\tau^{a}_2\big)^{n\choose 2}\prod_{1\leq j < k \leq n}\Big([m_j]_{{\mathcal R}(p^{a},q^{a})}\nonumber\\&-[m_k]_{{\mathcal R}(p^{a},q^{a})}\Big)+\prod_{1\leq j < k \leq n}\Big(\tau^{a\,m_j}_2-\tau^{a\,m_k}_2\Big)\Big),
	\end{align}
	\begin{align}\label{dn}
	D^n_{a\Delta}
	&=\Big(\big(\tau^{a}_1-\tau^{a}_2\big)^{n\choose 2}\prod_{1\leq j < k \leq n}\frac{\tau^{a\Delta(m_k+1)}_2}{\tau^{a\Delta\,m_k}_1}\Big([m_k]_{{\mathcal R}(p^{a},q^{a})}-[m_j]_{{\mathcal R}(p^{a},q^{a})}\nonumber\\&+\tau^{am_k}_2-\tau^{am_j}_1\Big),
	\end{align}
	\begin{align}\label{nn}
	N^n_{a\Delta}&=(-1)^{n+1}
	\Big(\big(\tau^{a}_1-\tau^{a}_2\big)^{n\choose 2}\prod_{1\leq j < k \leq n}\frac{\tau^{a\Delta(m_j+1)}_2}{\tau^{a\Delta\,m_j}_1}\Big([m_j]_{{\mathcal R}(p^{a},q^{a})}\nonumber\\&-[m_k]_{{\mathcal R}(p^{a},q^{a})}+\tau^{am_j}_2-\tau^{am_k}_1\Big),
	\end{align}
	and
	\begin{align}\label{hn}
	h^{\Delta(a)}_{\mathcal{R}(p,q)}(m_1,&\ldots,m_n)=\tau^{an\Delta(\bar{m}+1)}_2\prod_{1\leq j < k \leq n}\bigg([-\Delta\,m_j]_{{\mathcal R}(p^{a},q^{a})}\nonumber\\&\times[m_k(1-\Delta)]_{{\mathcal R}(p^{a},q^{a})}-[-\Delta\,m_k]_{{\mathcal R}(p^{a},q^{a})}[m_j(1-\Delta)]_{{\mathcal R}(p^{a},q^{a})}\bigg).
	\end{align}
\end{proposition}
\begin{proof}
	By straigthfoward computation.
\end{proof}
The results concerning the $q$-deformed conformal Witt $n$-algebra are contained in the following corollary.
\begin{corollary}
	The $q$-deformed conformal Witt $n$-algebra is given by the following commutation relation:
	\begin{align}\label{qcrna}
	\Big[{\mathcal T}^{a\Delta}_{m_1},\cdots,{\mathcal T}^{a\Delta}_{m_n}
	\Big]&={(-1)^{n+1}\over \big(q^{a}-1\big)^{n-1}}\Big( V^n_{a\Delta}[n]_{q^{a}}{\mathcal T}^{n\,a\Delta}_{\bar{m}}\nonumber\\ &- [n-1]_{q^{a}}\big(D^n_{a\Delta}+ N^n_{a\Delta}\big){\mathcal T}^{a(n-1)\Delta}_{\bar{m}}\Big)\nonumber\\&+h^{\Delta}_{q}(m_1,\ldots,m_n),
	\end{align}
	where 
	\begin{eqnarray*}
		V^n_{a\Delta}= q^{a(n-1)(1-\Delta )\bar{m}}\Big(\big(q^{a}-1\big)^{n\choose 2}\prod_{1\leq j < k \leq n}\Big([m_j]_{q^{a}}-[m_k]_{q^{a}}\Big)\Big),
	\end{eqnarray*}
	\begin{eqnarray*}
		D^n_{a\Delta}
		=\Big(\big(q^{a}-1\big)^{n\choose 2}\prod_{1\leq j < k \leq n}\frac{1}{q^{a\Delta\,m_k}}\Big([m_k]_{q^{a}}-[m_j]_{q^{a}}+1-q^{am_j}\Big),
	\end{eqnarray*}
	\begin{eqnarray*}
		N^n_{a\Delta}=(-1)^{n+1}
		\Big(\big(q^{a}-1\big)^{n\choose 2}\prod_{1\leq j < k \leq n}\frac{1}{q^{a\Delta\,m_j}}\Big([m_j]_{q^{a}}-[m_k]_{q^{a}}+1-q^{am_k}\Big),
	\end{eqnarray*}
	and
	\begin{align*}
	h^{\Delta(a)}_{q}(m_1,\ldots,m_n)&=\prod_{1\leq j < k \leq n}\bigg([-\Delta\,m_j]_{q^{a}}[m_k(1-\Delta)]_{q^{a}}\nonumber\\&-[-\Delta\,m_k]_{q^{a}}[m_j(1-\Delta)]_{q^{a}}\bigg).
	\end{align*}
\end{corollary}
\begin{proof}
	By taking $\tau_1=q$ and $\tau_2=1$ in the proposition \eqref{propcWitt}.
\end{proof}
\begin{remark}
	For $n=3$ in the relation (\ref{crna}), we obtain the ${\mathcal R}(p,q)$-conformal Witt $3$-algebra:
	\begin{align*}
	\Big[{\mathcal T}^{a\Delta}_{m_1},{\mathcal T}^{a\Delta}_{m_2},{\mathcal T}^{a\Delta}_{m_3}
	\Big]&={1\over \big(\tau^{a}_1-\tau^{a}_2\big)^{2}}\Big( V^3_{a\Delta}[3]_{{\mathcal R}(p^{a},q^{a})}{\mathcal T}^{3\,a\Delta}_{m_1+m_2+m_3}\nonumber\\ &- [2]_{{\mathcal R}(p^{a},q^{a})}\big(D^3_{a\Delta}+ N^3_{a\Delta}\big){\mathcal T}^{2a\Delta}_{m_1+m_2+m_3}\Big)\nonumber\\&+h^{\Delta}_{\mathcal{R}(p,q)}(m_1,m_2,m_3),
	\end{align*}
	where 
	\begin{align*}
	V^3_{a\Delta}&= \tau^{2a(1-\Delta )(m_1+m_2+m_3)}_1\Big(\big(\tau^{a}_1-\tau^{a}_2\big)^{3\choose 2}\prod_{1\leq j < k \leq 3}\Big([m_j]_{{\mathcal R}(p^{a},q^{a})}\nonumber\\&-[m_k]_{{\mathcal R}(p^{a},q^{a})}\Big)+\prod_{1\leq j < k \leq 3}\Big(\tau^{a\,m_j}_2-\tau^{a\,m_k}_2\Big)\Big),
	\end{align*}
	\begin{align*}
	D^3_{a\Delta}
	&=\Big(\big(\tau^{a}_1-\tau^{a}_2\big)^{3\choose 2}\prod_{1\leq j < k \leq 3}\frac{\tau^{a\Delta(m_k+1)}_2}{\tau^{a\Delta\,m_k}_1}\Big([m_k]_{{\mathcal R}(p^{a},q^{a})}-[m_j]_{{\mathcal R}(p^{a},q^{a})}\nonumber\\&+\tau^{am_k}_2-\tau^{am_j}_1\Big),
	\end{align*}
	\begin{align*}
	N^3_{a\Delta}&=
	\Big(\big(\tau^{a}_1-\tau^{a}_2\big)^{3\choose 2}\prod_{1\leq j < k \leq 3}\frac{\tau^{a\Delta(m_j+1)}_2}{\tau^{a\Delta\,m_j}_1}\Big([m_j]_{{\mathcal R}(p^{a},q^{a})}-[m_k]_{{\mathcal R}(p^{a},q^{a})}\nonumber\\&+\tau^{am_j}_2-\tau^{am_k}_1\Big),.
	\end{align*}
	and
	\begin{align*}
	h^{\Delta(a)}_{\mathcal{R}(p,q)}(m_1,&m_2,m_3)=\tau^{3a\Delta(m_1+m_2+m_3+1)}_2\prod_{1\leq j < k \leq 3}\bigg([-\Delta\,m_j]_{{\mathcal R}(p^{a},q^{a})}\nonumber\\&\times[m_k(1-\Delta)]_{{\mathcal R}(p^{a},q^{a})}-[-\Delta\,m_k]_{{\mathcal R}(p^{a},q^{a})}[m_j(1-\Delta)]_{{\mathcal R}(p^{a},q^{a})}\bigg).		
	\end{align*}
\end{remark}
\subsection{Conformal Witt $n$-algebra and quantum algebras}
In this section, the conformal Witt $n$-algebra corresponding to quantum algebra known in the literature are deduced.
\subsubsection{Conformal Witt $n$-algebra associated to the Biedenharn-Macfarlane algebra }
The conformal Witt $n$-algebra induced by the Biedenharn-Macfarlane algebra\cite{BC,M} is derived by:  
for $(\Delta \neq 0,1),$ the $q$-conformal operator defined by:
\begin{eqnarray}\label{qop1}
{\mathcal T}^{a\Delta}_m:=x^{(1-\Delta)(m+1)}{\mathcal D}_{q^{a}}\,x^{\Delta(m+1)},
\end{eqnarray}
where ${\mathcal D}_{q^{a}}$ is the $q$-derivative :
\begin{eqnarray}\label{qder}
\mathcal{D}_{q^a}\phi(x)= \frac{\phi(q^a\,x)-\phi(q^{-a}\,x)}{q^a-q^{-a}}.
\end{eqnarray}
From the  $q$-number
\begin{eqnarray}\label{qnumber}
[n]_{q^a}=\frac{q^{an}-q^{-an}}{q^a-q^{-a}},
\end{eqnarray} the relation \eqref{qop1} is reduced as:
\begin{eqnarray*}
	{\mathcal T}^{a\Delta}_m=[\Delta(m+1)]_{q^{a}}\,x^{m}.
\end{eqnarray*}
Thus, the $q$-conformal operators (\ref{qop1}) satisfy the product relation: \begin{align}
\mathcal{T}^{a\Delta}_{m}.\mathcal{T}^{b\Delta}_n&={\big(q^{a+b}-q^{-a-b}\big)q^{an(1-\Delta)}\over \big(q^{a}-q^{-a}\big)\big(q^{b}-q^{-b}\big)q^{b\Delta\,m}}\,\mathcal{T}^{(a+b)\Delta}_{m+n}+ f^{\Delta(a,b)}_{q}(m,n)\nonumber\\& - {q^{-a(\Delta(m+1)+n)}\over q^{b\Delta\,m}\big(q^{a}-q^{-a}\big)}\,\mathcal{T}^{b\Delta}_{m+n}- {q^{-b\Delta(n+1)}\over q^{an(\Delta-1)}\big(q^{b}-q^{-b}\big)}\, \mathcal{T}^{a\Delta}_{m+n},
\end{align}
with \begin{eqnarray*}
	f^{\Delta(a,b)}_{q}(m,n)=q^{-(a+b)\Delta(m+n+1)}\,[-\Delta\,m]_{q^{b}}\,[n(1-\Delta)]_{q^{a}},
\end{eqnarray*}	
and the commutation relation
\begin{align}\label{qcrtoa}
\Big[\mathcal{T}^{a\Delta}_{m}, \mathcal{T}^{b\Delta}_n\Big]&={\big(q^{a+b}-q^{-a-b}\big)\over \big(q^{a}-q^{-a}\big)\big(q^{b}-q^{-b}\big)}\bigg(\frac{q^{an(1-\Delta)}}{q^{b\Delta\,m}}-\frac{q^{bm(1-\Delta)}}{q^{a\Delta\,n}}\bigg)\mathcal{T}^{(a+b)\Delta}_{m+n}\nonumber\\ &-\frac{q^{-b\Delta(n+1)}}{q^{a\Delta\,n}}\frac{\big(q^{an}-q^{-bm}\big)}{\big(q^{b}-q^{-b}\big)}\mathcal{T}^{a\Delta}_{m+n}+h^{\Delta(a,b)}_{q}(m,n)\nonumber\\ &- \frac{q^{-a\Delta(m+1)}}{q^{b\Delta\,m}}\frac{\big(q^{bm}-q^{-an}\big)}{\big(q^{a}-q^{-a}\big)}\mathcal{T}^{b\Delta}_{m+n},
\end{align}
where 
\begin{align*}
h^{\Delta(a,b)}_{q}(m,n)&=q^{-(a+b)\Delta(m+n+1)}\bigg([-\Delta\,m]_{q^{b}}\,[n(1-\Delta)]_{q^{a}}\nonumber\\&-[-\Delta\,n]_{q^{a}}\,[m(1-\Delta)]_{q^{b}}\bigg).
\end{align*}
Putting $a=b=1,$ in the relation \eqref{qcrtoa},  the commutation relation:
\begin{align*}
\Big[\mathcal{T}^{\Delta}_{m}, \mathcal{T}^{\Delta}_n\Big]&={q^{-\Delta(n+m)}\big(q^{n}-q^{m}\big)\over \big(q-q^{-1}\big)}[2]_{q}\mathcal{T}^{2\Delta}_{m+n}+h^{\Delta}_{q}(m,n)\nonumber\\
&-\bigg(\frac{q^{-\Delta(n+1)}\big(q^{n}-q^{-m}\big)}{q^{\Delta\,n}\big(q-q^{-1}\big)}+ \frac{q^{-\Delta(m+1)}\big(q^{m}-q^{-n}\big)}{q^{\Delta\,m}\big(q-q^{-1}\big)}\bigg)\mathcal{T}^{\Delta}_{m+n},
\end{align*}
where
\begin{eqnarray*}
	h^{\Delta(1,1)}_{q}(m,n)=q^{-2\Delta(m+n+1)}\bigg([-\Delta\,m]_{q}\,[n(1-\Delta)]_{q}-[-\Delta\,n]_{q}\,[m(1-\Delta)]_{q}\bigg).
\end{eqnarray*}
Besides, the $n$-bracket is defined by:
\begin{eqnarray*}
	\Big[{\mathcal T}^{a_1\Delta}_{m_1},\cdots,{\mathcal T}^{a_n\Delta}_{m_n}
	\Big]:=\epsilon^{i_1 \cdots i_n}_{1 \cdots n}\,{\mathcal T}^{a_{i_1}\Delta}_{m_{i_1}} \cdots {\mathcal T}^{a_{i_n}\Delta}_{m_{i_n}}.
\end{eqnarray*}
Our interest is	focussed on the special  case with the same $a\Delta.$ Then, 
\begin{eqnarray*}
	\Big[{\mathcal T}^{a\Delta}_{m_1},\cdots,{\mathcal T}^{a\Delta}_{m_n}
	\Big]=\epsilon^{1\cdots n}_{1\cdots n}\,{\mathcal T}^{a\Delta}_{m_1} \cdots {\mathcal T}^{a\Delta}_{m_n}.
\end{eqnarray*}
Setting $a=b$ in the relation (\ref{qcrtoa}),  we obtain:
\begin{align*}
\Big[\mathcal{T}^{a\Delta}_{m}, \mathcal{T}^{a\Delta}_n\Big]&={q^{-a\Delta(n+m)}\big(q^{an}-q^{am}\big)\over q^{a}-q^{-a}}[2]_{q^a}\mathcal{T}^{2a\Delta}_{m+n}\nonumber\\ &-\frac{1}{\big(q-q^{-1}\big)}\bigg(\frac{q^{-a\Delta(n+1)}}{q^{a\Delta\,n}}\big(q^{an}-q^{-am}\big)\nonumber\\&- \frac{q^{-a\Delta(m+1)}}{q^{a\Delta\,m}}\big(q^{am}-q^{-an}\big)\bigg)\mathcal{T}^{a\Delta}_{m+n}+h^{\Delta(a)}_{q}(m,n),
\end{align*}
where \begin{align*}
h^{\Delta(a)}_{q}(m,n)&=q^{-2a\Delta(m+n+1)}\bigg([-\Delta\,m]_{q^{a}}[n(1-\Delta)]_{q^{a}}\nonumber\\&-[-\Delta\,n]_{q^{a}}[m(1-\Delta)]_{q^{a}}\bigg).
\end{align*}
The $q$-conformal Witt $n$-algebra is presented  by :
\begin{align*}
\Big[{\mathcal T}^{a\Delta}_{m_1},\cdots,{\mathcal T}^{a\Delta}_{m_n}
\Big]&={(-1)^{n+1}\over \big(q^{a}-q^{-a}\big)^{n-1}}\Big( V^n_{a\Delta}[n]_{q^{a}}{\mathcal T}^{n\,a\Delta}_{\bar{m}}\nonumber\\ &- [n-1]_{q^{a}}\big(D^n_{a\Delta}+ N^n_{a\Delta}\big){\mathcal T}^{a(n-1)\Delta}_{\bar{m}}\Big)\nonumber\\&+h^{\Delta}_{q}(m_1,\ldots,m_n),
\end{align*}
where 
\begin{align*}
V^n_{a\Delta}&= q^{a(n-1)(1-\Delta )\bar{m}}\Big(\big(q^{a}-q^{-a}\big)^{n\choose 2}\prod_{1\leq j < k \leq n}\Big([m_j]_{q^{a}}-[m_k]_{q^{a}}\Big)\nonumber\\&+\prod_{1\leq j < k \leq n}\Big(q^{-a\,m_j}-q^{-a\,m_k}\Big)\Big),
\end{align*}
\begin{eqnarray*}
	D^n_{a\Delta}
	=\Big(\big(q^{a}-q^{-a}\big)^{n\choose 2}\prod_{1\leq j < k \leq n}\frac{q^{-a\Delta(m_k+1)}}{q^{a\Delta\,m_k}}\Big([m_k]_{q^{a}}-[m_j]_{q^{a}}+q^{-am_k}-q^{am_j}\Big),
\end{eqnarray*}
\begin{align*}
N^n_{a\Delta}&=(-1)^{n+1}
\Big(\big(q^{a}-q^{-a}\big)^{n\choose 2}\prod_{1\leq j < k \leq n}\frac{q^{-a\Delta(m_j+1)}}{q^{a\Delta\,m_j}}\Big([m_j]_{q^{a}}-[m_k]_{q^{a}}\nonumber\\&+q^{-am_j}-q^{am_k}\Big),
\end{align*}
and
\begin{align*}
h^{\Delta(a)}_{q}(m_1,\ldots,m_n)&=q^{-an\Delta(\bar{m}+1)}\prod_{1\leq j < k \leq n}\bigg([-\Delta\,m_j]_{q^{a}}[m_k(1-\Delta)]_{q^{a}}\nonumber\\&-[-\Delta\,m_k]_{q^{a}}[m_j(1-\Delta)]_{q^{a}}\bigg).
\end{align*}
Putting $n=3$ in the relation (\ref{qcrna}), we obtain the $q$-conformal Witt $3$-algebra:
\begin{align*}
\Big[{\mathcal T}^{a\Delta}_{m_1},{\mathcal T}^{a\Delta}_{m_2},{\mathcal T}^{a\Delta}_{m_3}
\Big]&={1\over \big(q^{a}-q^{-a}\big)^{2}}\Big( V^3_{a\Delta}[3]_{q^{a}}{\mathcal T}^{3\,a\Delta}_{m_1+m_2+m_3}\nonumber\\ &- [2]_{q^{a}}\big(D^3_{a\Delta}+ N^3_{a\Delta}\big){\mathcal T}^{2a\Delta}_{m_1+m_2+m_3}\Big)\nonumber\\&+h^{\Delta}_{q}(m_1,m_2,m_3),
\end{align*}
where 
\begin{align*}
V^3_{a\Delta}&= q^{2a(1-\Delta )(m_1+m_2+m_3)}\Big(\big(q^{a}-q^{-a}\big)^{3\choose 2}\prod_{1\leq j < k \leq 3}\Big([m_j]_{q^{a}}-[m_k]_{q^{a}}\Big)\nonumber\\&+\prod_{1\leq j < k \leq 3}\Big(q^{-a\,m_j}-q^{-a\,m_k}\Big)\Big),
\end{align*}
\begin{eqnarray*}
	D^3_{a\Delta}
	=\Big(\big(q^{a}-q^{-a}\big)^{3\choose 2}\prod_{1\leq j < k \leq 3}\frac{q^{-a\Delta(m_k+1)}}{q^{a\Delta\,m_k}}\Big([m_k]_{q^{a}}-[m_j]_{q^{a}}+q^{-am_k}-q^{am_j}\Big),
\end{eqnarray*}
\begin{eqnarray*}
	N^3_{a\Delta}=
	\Big(\big(q^{a}-q^{-a}\big)^{3\choose 2}\prod_{1\leq j < k \leq 3}\frac{q^{-a\Delta(m_j+1)}}{q^{a\Delta\,m_j}}\Big([m_j]_{q^{a}}-[m_k]_{q^{a}}+q^{-am_j}-q^{am_k}\Big),
\end{eqnarray*}
and
\begin{align*}
h^{\Delta(a)}_{q}(m_1,m_2,m_3)&=q^{-3a\Delta(m_1+m_2+m_3+1)}\prod_{1\leq j < k \leq 3}\bigg([-\Delta\,m_j]_{q^{a}}[m_k(1-\Delta)]_{q^{a}}\nonumber\\&-[-\Delta\,m_k]_{q^{a}}[m_j(1-\Delta)]_{q^{a}}\bigg).		
\end{align*}
\subsubsection{Conformal Witt $n$-algebra corresponding to the Jagannathan-Srinivasa algebra } 
The conformal Witt $n$-algebra from the Jagannathan-Srinivasa algebra \cite{JS} is deduced as follows: We consider 
for $(\Delta \neq 0,1),$ the $(p,q)$-conformal operator defined by:
\begin{eqnarray}\label{pqop1}
{\mathcal T}^{a\Delta}_m:=x^{(1-\Delta)(m+1)}{\mathcal D}_{p^{a},q^{a}}\,x^{\Delta(m+1)},
\end{eqnarray}
where ${\mathcal D}_{p^{a},q^{a}}$ is the $(p,q)$-derivative :
\begin{eqnarray}\label{pqder}
\mathcal{D}_{p^a,q^a}\phi(x)= \frac{\phi(p^a\,x)-\phi(q^a\,x)}{p^a-q^a}.
\end{eqnarray}
Then, by using the  $(p,q)$-number
\begin{eqnarray}\label{pqnumber}
[n]_{p^a,q^a}=\frac{p^{an}-q^{an}}{p^a-q^a},
\end{eqnarray} the relation \eqref{pqop1} takes the form:
\begin{eqnarray}\label{pqopa}
{\mathcal T}^{a\Delta}_m=[\Delta(m+1)]_{p^{a},q^{a}}\,x^{m}.
\end{eqnarray}
Thus, the $(p,q)$-conformal operators (\ref{pqop1}) obeys the product relation: \begin{align*}
\mathcal{T}^{a\Delta}_{m}.\mathcal{T}^{b\Delta}_n&={\big(p^{a+b}-q^{a+b}\big)p^{an(1-\Delta)}\over \big(p^{a}-q^{a}\big)\big(p^{b}-q^{b}\big)p^{b\Delta\,m}}\,\mathcal{T}^{(a+b)\Delta}_{m+n}- {q^{b\Delta(n+1)}\over p^{an(\Delta-1)}\big(p^{b}-q^{b}\big)}\, \mathcal{T}^{a\Delta}_{m+n}\nonumber\\& - {q^{a(\Delta(m+1)+n)}\over p^{b\Delta\,m}\big(p^{a}-q^{a}\big)}\,\mathcal{T}^{b\Delta}_{m+n}+ f^{\Delta(a,b)}_{p,q}(m,n),
\end{align*}
with \begin{eqnarray*}
	f^{\Delta(a,b)}_{p,q}(m,n)=q^{(a+b)\Delta(m+n+1)}\,[-\Delta\,m]_{p^{b},q^{b}}\,[n(1-\Delta)]_{p^{a},q^{a}},
\end{eqnarray*}	
and the commutation relation
\begin{align}\label{pqcrto}
\Big[\mathcal{T}^{a\Delta}_{m}, \mathcal{T}^{b\Delta}_n\Big]&={\big(p^{a+b}-q^{a+b}\big)\over \big(p^{a}-q^{a}\big)\big(p^{b}-q^{b}\big)}\bigg(\frac{p^{an(1-\Delta)}}{p^{b\Delta\,m}}-\frac{p^{bm(1-\Delta)}}{p^{a\Delta\,n}}\bigg)\mathcal{T}^{(a+b)\Delta}_{m+n}\nonumber\\ &-\frac{q^{b\Delta(n+1)}}{p^{a\Delta\,n}}\frac{\big(p^{an}-q^{bm}\big)}{\big(p^{b}-q^{b}\big)}\mathcal{T}^{a\Delta}_{m+n}- \frac{q^{a\Delta(m+1)}}{p^{b\Delta\,m}}\frac{\big(p^{bm}-q^{an}\big)}{\big(p^{a}-q^{a}\big)}\mathcal{T}^{b\Delta}_{m+n}\nonumber\\ &+h^{\Delta(a,b)}_{p,q}(m,n),
\end{align}
where 
\begin{align*}
h^{\Delta(a,b)}_{p,q}(m,n)&=q^{(a+b)\Delta(m+n+1)}\bigg([-\Delta\,m]_{p^{b},q^{b}}\,[n(1-\Delta)]_{p^{a},q^{a}}\nonumber\\&-[-\Delta\,n]_{p^{a},q^{a}}\,[m(1-\Delta)]_{p^{b},q^{b}}\bigg).
\end{align*}
Putting $a=b=1,$ in the relation \eqref{pqcrto}, we obtain the commutation relation:
\begin{align*}
\Big[\mathcal{T}^{\Delta}_{m}, \mathcal{T}^{\Delta}_n\Big]&={p^{-\Delta(n+m)}\big(p^{n}-p^{m}\big)\over \big(p-q\big)}[2]_{p,q}\mathcal{T}^{2\Delta}_{m+n}+h^{\Delta}_{p,q}(m,n)\nonumber\\
&-\bigg(\frac{q^{\Delta(n+1)}\big(p^{n}-q^{m}\big)}{p^{\Delta\,n}\big(p-q\big)}+ \frac{q^{\Delta(m+1)}\big(p^{m}-q^{n}\big)}{p^{\Delta\,m}\big(p-q\big)}\bigg)\mathcal{T}^{\Delta}_{m+n},
\end{align*}
where
\begin{align*}
h^{\Delta(1,1)}_{p,q}(m,n)&=q^{2\Delta(m+n+1)}\bigg([-\Delta\,m]_{p,q}\,[n(1-\Delta)]_{p,q}\nonumber\\&-[-\Delta\,n]_{p,q}\,[m(1-\Delta)]_{p,q}\bigg).
\end{align*}
Moreover, the $n$-bracket is defined by:
\begin{eqnarray*}
	\Big[{\mathcal T}^{a_1\Delta}_{m_1},\cdots,{\mathcal T}^{a_n\Delta}_{m_n}
	\Big]:=\epsilon^{i_1 \cdots i_n}_{1 \cdots n}\,{\mathcal T}^{a_{i_1}\Delta}_{m_{i_1}} \cdots {\mathcal T}^{a_{i_n}\Delta}_{m_{i_n}}.
\end{eqnarray*}
Our interest is	focussed on the special  case with the same $a\Delta.$ Then, 
\begin{eqnarray*}
	\Big[{\mathcal T}^{a\Delta}_{m_1},\cdots,{\mathcal T}^{a\Delta}_{m_n}
	\Big]=\epsilon^{1\cdots n}_{1\cdots n}\,{\mathcal T}^{a\Delta}_{m_1} \cdots {\mathcal T}^{a\Delta}_{m_n}.
\end{eqnarray*}
Putting $a=b$ in the relation (\ref{pqcrto}),  the following commutation relation holds:
\begin{align*}
\Big[\mathcal{T}^{a\Delta}_{m}, \mathcal{T}^{a\Delta}_n\Big]&={p^{-a\Delta(n+m)}\big(p^{an}-p^{am}\big)\over p^{a}-q^{a}}[2]_{p^a,q^a}\mathcal{T}^{2a\Delta}_{m+n}\nonumber\\ &-\frac{1}{\big(\tau_1-\tau_2\big)}\bigg(\frac{q^{a\Delta(n+1)}}{p^{a\Delta\,n}}\big(p^{an}-q^{am}\big)\nonumber\\&- \frac{q^{a\Delta(m+1)}}{p^{a\Delta\,m}}\big(p^{am}-q^{an}\big)\bigg)\mathcal{T}^{a\Delta}_{m+n}+h^{\Delta(a)}_{p,q}(m,n),
\end{align*}
where \begin{align*}
h^{\Delta(a)}_{p,q}(m,n)&=q^{2a\Delta(m+n+1)}\bigg([-\Delta\,m]_{p^{a},q^{a}}[n(1-\Delta)]_{p^{a},q^{a}}\nonumber\\&-[-\Delta\,n]_{p^{a},q^{a}}[m(1-\Delta)]_{p^{a},q^{a}}\bigg)
\end{align*}
The $(p,q)$-conformal Witt $n$-algebra is presented  by :
\begin{align}\label{pqcrna}
\Big[{\mathcal T}^{a\Delta}_{m_1},\cdots,{\mathcal T}^{a\Delta}_{m_n}
\Big]&={(-1)^{n+1}\over \big(p^{a}-q^{a}\big)^{n-1}}\Big( V^n_{a\Delta}[n]_{p^{a},q^{a}}{\mathcal T}^{n\,a\Delta}_{\bar{m}}\nonumber\\ &- [n-1]_{p^{a},q^{a}}\big(D^n_{a\Delta}+ N^n_{a\Delta}\big){\mathcal T}^{a(n-1)\Delta}_{\bar{m}}\Big)\nonumber\\&+h^{\Delta}_{p,q}(m_1,\ldots,m_n),
\end{align}
where 
\begin{align*}
V^n_{a\Delta}&= p^{a(n-1)(1-\Delta )\bar{m}}\Big(\big(p^{a}-q^{a}\big)^{n\choose 2}\prod_{1\leq j < k \leq n}\Big([m_j]_{p^{a},q^{a}}-[m_k]_{p^{a},q^{a}}\Big)\nonumber\\&+\prod_{1\leq j < k \leq n}\Big(q^{a\,m_j}-q^{a\,m_k}\Big)\Big),
\end{align*}
\begin{eqnarray*}
	D^n_{a\Delta}
	=\Big(\big(p^{a}-q^{a}\big)^{n\choose 2}\prod_{1\leq j < k \leq n}\frac{q^{a\Delta(m_k+1)}}{p^{a\Delta\,m_k}}\Big([m_k]_{p^{a},q^{a}}-[m_j]_{p^{a},q^{a}}+q^{am_k}-p^{am_j}\Big),
\end{eqnarray*}
\begin{align*}
N^n_{a\Delta}&=(-1)^{n+1}
\Big(\big(p^{a}-q^{a}\big)^{n\choose 2}\prod_{1\leq j < k \leq n}\frac{q^{a\Delta(m_j+1)}}{p^{a\Delta\,m_j}}\Big([m_j]_{p^{a},q^{a}}-[m_k]_{p^{a},q^{a}}\nonumber\\&+q^{am_j}-p^{am_k}\Big),
\end{align*}
and
\begin{align*}
h^{\Delta(a)}_{p,q}(m_1,\ldots,m_n)&=q^{an\Delta(\bar{m}+1)}\prod_{1\leq j < k \leq n}\bigg([-\Delta\,m_j]_{p^{a},q^{a}}[m_k(1-\Delta)]_{p^{a},q^{a}}\nonumber\\&-[-\Delta\,m_k]_{p^{a},q^{a}}[m_j(1-\Delta)]_{p^{a},q^{a}}\bigg).
\end{align*}
Putting $n=3$ in the relation (\ref{pqcrna}), we obtain the $(p,q)$-conformal Witt $3$-algebra:
\begin{align*}
\Big[{\mathcal T}^{a\Delta}_{m_1},{\mathcal T}^{a\Delta}_{m_2},{\mathcal T}^{a\Delta}_{m_3}
\Big]&={1\over \big(p^{a}-q^{a}\big)^{2}}\Big( V^3_{a\Delta}[3]_{p^{a},q^{a}}{\mathcal T}^{3\,a\Delta}_{m_1+m_2+m_3}\nonumber\\ &- [2]_{p^{a},q^{a}}\big(D^3_{a\Delta}+ N^3_{a\Delta}\big){\mathcal T}^{2a\Delta}_{m_1+m_2+m_3}\Big)\nonumber\\&+h^{\Delta}_{p,q}(m_1,m_2,m_3),
\end{align*}
where 
\begin{align*}
V^3_{a\Delta}&= p^{2a(1-\Delta )(m_1+m_2+m_3)}\Big(\big(p^{a}-q^{a}\big)^{3\choose 2}\prod_{1\leq j < k \leq 3}\Big([m_j]_{p^{a},q^{a}}-[m_k]_{p^{a},q^{a}}\Big)\nonumber\\&+\prod_{1\leq j < k \leq 3}\Big(q^{a\,m_j}-q^{a\,m_k}\Big)\Big),
\end{align*}
\begin{eqnarray*}
	D^3_{a\Delta}
	=\Big(\big(p^{a}-q^{a}\big)^{3\choose 2}\prod_{1\leq j < k \leq 3}\frac{q^{a\Delta(m_k+1)}}{p^{a\Delta\,m_k}}\Big([m_k]_{p^{a},q^{a}}-[m_j]_{p^{a},q^{a}}+q^{am_k}-p^{am_j}\Big),
\end{eqnarray*}
\begin{eqnarray*}
	N^3_{a\Delta}=
	\Big(\big(p^{a}-q^{a}\big)^{3\choose 2}\prod_{1\leq j < k \leq 3}\frac{q^{a\Delta(m_j+1)}}{p^{a\Delta\,m_j}}\Big([m_j]_{p^{a},q^{a}}-[m_k]_{p^{a},q^{a}}+q^{am_j}-p^{am_k}\Big),
\end{eqnarray*}
and
\begin{align*}
h^{\Delta(a)}_{p,q}(m_1,m_2,m_3)&=q^{3a\Delta(m_1+m_2+m_3+1)}\prod_{1\leq j < k \leq 3}\bigg([-\Delta\,m_j]_{p^{a},q^{a}}[m_k(1-\Delta)]_{p^{a},q^{a}}\nonumber\\&-[-\Delta\,m_k]_{p^{a},q^{a}}[m_j(1-\Delta)]_{p^{a},q^{a}}\bigg).		
\end{align*}
\subsubsection{Conformal  Witt $n$-algebra associated to the Chakrabarty and Jagannathan algebra }
The conformal Witt $n$-algebra from the Chakrabarty and Jagannathan algebra \cite{CJ}  is deduced as follows: We consider 
for $(\Delta \neq 0,1),$ the $(p^{-1},q)$-conformal operator defined by:
\begin{eqnarray}\label{cjop1}
{\mathcal T}^{a\Delta}_m:=x^{(1-\Delta)(m+1)}{\mathcal D}_{p^{-a},q^{a}}\,x^{\Delta(m+1)},
\end{eqnarray}
where ${\mathcal D}_{p^{-a},q^{a}}$ is the $(p^{-1},q)$-derivative :
\begin{eqnarray}\label{cjder}
\mathcal{D}_{p^{-a},q^a}\phi(x)= \frac{\phi(p^{-a}\,x)-\phi(q^a\,x)}{(p^{-a}-q^a)x}.
\end{eqnarray}
Then, by using the  $(p^{-1},q)$-number
\begin{eqnarray}\label{cjnumber}
[n]_{p^{-a},q^a}=\frac{p^{-an}-q^{an}}{p^{-a}-q^a},
\end{eqnarray} the relation \eqref{cjop1} takes the form:
\begin{eqnarray}\label{cjopa}
{\mathcal T}^{a\Delta}_m=[\Delta(m+1)]_{p^{-a},q^{a}}\,x^{m}.
\end{eqnarray}
Thus, the $(p^{-1},q)$-conformal operators (\ref{cjop1}) obeys the product relation: \begin{align}
\mathcal{T}^{a\Delta}_{m}.\mathcal{T}^{b\Delta}_n&={\big(p^{-a-b}-q^{a+b}\big)p^{-an(1-\Delta)}\over \big(p^{-a}-q^{a}\big)\big(p^{-b}-q^{b}\big)p^{-b\Delta\,m}}\,\mathcal{T}^{(a+b)\Delta}_{m+n}+ f^{\Delta(a,b)}_{p^{-1},q}(m,n)\nonumber\\&-{q^{a(\Delta(m+1)+n)}\over p^{-b\Delta\,m}\big(p^{-a}-q^{a}\big)}\,\mathcal{T}^{b\Delta}_{m+n}  - {q^{b\Delta(n+1)}\over p^{-an(\Delta-1)}\big(p^{-b}-q^{b}\big)}\, \mathcal{T}^{a\Delta}_{m+n},
\end{align}
with \begin{eqnarray*}
	f^{\Delta(a,b)}_{p^{-1},q}(m,n)=q^{(a+b)\Delta(m+n+1)}\,[-\Delta\,m]_{p^{-b},q^{b}}\,[n(1-\Delta)]_{p^{-a},q^{a}},
\end{eqnarray*}	
and the commutation relation
\begin{align}\label{cjcrto}
\Big[\mathcal{T}^{a\Delta}_{m}, \mathcal{T}^{b\Delta}_n\Big]&={\big(p^{-a-b}-q^{a+b}\big)\over \big(p^{-a}-q^{a}\big)\big(p^{-b}-q^{b}\big)}\bigg(\frac{p^{-an(1-\Delta)}}{p^{-b\Delta\,m}}-\frac{p^{-bm(1-\Delta)}}{p^{-a\Delta\,n}}\bigg)\mathcal{T}^{(a+b)\Delta}_{m+n}\nonumber\\ &-\frac{q^{b\Delta(n+1)}}{p^{-a\Delta\,n}}\frac{\big(p^{-an}-q^{bm}\big)}{\big(p^{-b}-q^{b}\big)}\mathcal{T}^{a\Delta}_{m+n}+h^{\Delta(a,b)}_{p^{-1},q}(m,n)\nonumber\\ &- \frac{q^{a\Delta(m+1)}}{p^{-b\Delta\,m}}\frac{\big(p^{-bm}-q^{an}\big)}{\big(p^{-a}-q^{a}\big)}\mathcal{T}^{b\Delta}_{m+n},
\end{align}
where 
\begin{align*}
h^{\Delta(a,b)}_{p^{-1},q}(m,n)&=q^{(a+b)\Delta(m+n+1)}\bigg([-\Delta\,m]_{p^{-b},q^{b}}\,[n(1-\Delta)]_{p^{-a},q^{a}}\nonumber\\&-[-\Delta\,n]_{p^{-a},q^{a}}\,[m(1-\Delta)]_{p^{-b},q^{b}}\bigg).
\end{align*}
Putting $a=b=1,$ in the relation \eqref{cjcrto}, we obtain the commutation relation:
\begin{align*}
\Big[\mathcal{T}^{\Delta}_{m}, \mathcal{T}^{\Delta}_n\Big]&={p^{\Delta(n+m)}\big(p^{-n}-p^{m}\big)\over \big(p^{-1}-q\big)}[2]_{p^{-1},q}\mathcal{T}^{2\Delta}_{m+n}+h^{\Delta}_{p^{-1},q}(m,n)\nonumber\\
&-\bigg(\frac{q^{\Delta(n+1)}\big(p^{-n}-q^{m}\big)}{p^{-\Delta\,n}\big(p^{-1}-q\big)}+ \frac{q^{\Delta(m+1)}\big(p^{-m}-q^{n}\big)}{p^{-\Delta\,m}\big(p^{-1}-q\big)}\bigg)\mathcal{T}^{\Delta}_{m+n},
\end{align*}
where
\begin{align*}
h^{\Delta(1,1)}_{p^{-1},q}(m,n)&=q^{2\Delta(m+n+1)}\bigg([-\Delta\,m]_{p^{-1},q}\,[n(1-\Delta)]_{p^{-1},q}\nonumber\\&-[-\Delta\,n]_{p^{-1},q}\,[m(1-\Delta)]_{p^{-1},q}\bigg).
\end{align*}
Moreover, the $n$-bracket is defined by:
\begin{eqnarray*}
	\Big[{\mathcal T}^{a_1\Delta}_{m_1},\cdots,{\mathcal T}^{a_n\Delta}_{m_n}
	\Big]:=\epsilon^{i_1 \cdots i_n}_{1 \cdots n}\,{\mathcal T}^{a_{i_1}\Delta}_{m_{i_1}} \cdots {\mathcal T}^{a_{i_n}\Delta}_{m_{i_n}}.
\end{eqnarray*}
Our interest is	focussed on the special  case with the same $a\Delta.$ Then, 
\begin{eqnarray*}
	\Big[{\mathcal T}^{a\Delta}_{m_1},\cdots,{\mathcal T}^{a\Delta}_{m_n}
	\Big]=\epsilon^{1\cdots n}_{1\cdots n}\,{\mathcal T}^{a\Delta}_{m_1} \cdots {\mathcal T}^{a\Delta}_{m_n}.
\end{eqnarray*}
Putting $a=b$ in the relation (\ref{cjcrto}),  the following commutation relation holds:
\begin{align*}
\Big[\mathcal{T}^{a\Delta}_{m}, \mathcal{T}^{a\Delta}_n\Big]&={p^{a\Delta(n+m)}\big(p^{-an}-p^{-am}\big)\over p^{-a}-q^{a}}[2]_{p^{-a},q^a}\mathcal{T}^{2a\Delta}_{m+n}\nonumber\\ &-\frac{1}{\big(p^{-1}-q\big)}\bigg(\frac{q^{a\Delta(n+1)}}{p^{-a\Delta\,n}}\big(p^{-an}-q^{am}\big)\nonumber\\&- \frac{q^{a\Delta(m+1)}}{p^{-a\Delta\,m}}\big(p^{-am}-q^{an}\big)\bigg)\mathcal{T}^{a\Delta}_{m+n}+h^{\Delta(a)}_{p^{-1},q}(m,n),
\end{align*}
where \begin{align*}
h^{\Delta(a)}_{p,q}(m,n)&=q^{2a\Delta(m+n+1)}\bigg([-\Delta\,m]_{p^{-a},q^{a}}[n(1-\Delta)]_{p^{-a},q^{a}}\nonumber\\&-[-\Delta\,n]_{p^{-a},q^{a}}[m(1-\Delta)]_{p^{-a},q^{a}}\bigg)
\end{align*}
The $(p^{-1},q)$-conformal Witt $n$-algebra is presented  by :
\begin{align}\label{cjcrna}
\Big[{\mathcal T}^{a\Delta}_{m_1},\cdots,{\mathcal T}^{a\Delta}_{m_n}
\Big]&={(-1)^{n+1}\over \big(p^{-a}-q^{a}\big)^{n-1}}\Big( V^n_{a\Delta}[n]_{p^{-a},q^{a}}{\mathcal T}^{n\,a\Delta}_{\bar{m}}\nonumber\\ &- [n-1]_{p^{-a},q^{a}}\big(D^n_{a\Delta}+ N^n_{a\Delta}\big){\mathcal T}^{a(n-1)\Delta}_{\bar{m}}\Big)\nonumber\\&+h^{\Delta}_{p^{-1},q}(m_1,\ldots,m_n),
\end{align}
where 
\begin{align*}
V^n_{a\Delta}&= p^{-a(n-1)(1-\Delta )\bar{m}}\Big(\big(p^{-a}-q^{a}\big)^{n\choose 2}\prod_{1\leq j < k \leq n}\Big([m_j]_{p^{-a},q^{a}}-[m_k]_{p^{-a},q^{a}}\Big)\nonumber\\&+\prod_{1\leq j < k \leq n}\Big(q^{a\,m_j}-q^{a\,m_k}\Big)\Big),
\end{align*}
\begin{align*}
D^n_{a\Delta}
&=\Big(\big(p^{-a}-q^{a}\big)^{n\choose 2}\prod_{1\leq j < k \leq n}\frac{q^{a\Delta(m_k+1)}}{p^{-a\Delta\,m_k}}\Big([m_k]_{p^{-a},q^{a}}-[m_j]_{p^{-a},q^{a}}\\&+q^{am_k}-p^{-am_j}\Big),
\end{align*}
\begin{align*}
N^n_{a\Delta}&=(-1)^{n+1}
\Big(\big(p^{-a}-q^{a}\big)^{n\choose 2}\prod_{1\leq j < k \leq n}\frac{q^{a\Delta(m_j+1)}}{p^{-a\Delta\,m_j}}\Big([m_j]_{p^{-a},q^{a}}-[m_k]_{p^{-a},q^{a}}\nonumber\\&+q^{am_j}-p^{-am_k}\Big),
\end{align*}
and
\begin{align*}
h^{\Delta(a)}_{p^{-1},q}(m_1,\ldots,m_n)&=q^{an\Delta(\bar{m}+1)}\prod_{1\leq j < k \leq n}\bigg([-\Delta\,m_j]_{p^{-a},q^{a}}[m_k(1-\Delta)]_{p^{-a},q^{a}}\nonumber\\&-[-\Delta\,m_k]_{p^{-a},q^{a}}[m_j(1-\Delta)]_{p^{-a},q^{a}}\bigg).
\end{align*}
Putting $n=3$ in the relation (\ref{cjcrna}), we obtain the $(p^{-1},q)$-conformal Witt $3$-algebra:
\begin{align*}
\Big[{\mathcal T}^{a\Delta}_{m_1},{\mathcal T}^{a\Delta}_{m_2},{\mathcal T}^{a\Delta}_{m_3}
\Big]&={1\over \big(p^{-a}-q^{a}\big)^{2}}\Big( V^3_{a\Delta}[3]_{p^{-a},q^{a}}{\mathcal T}^{3\,a\Delta}_{m_1+m_2+m_3}\nonumber\\ &- [2]_{p^{-a},q^{a}}\big(D^3_{a\Delta}+ N^3_{a\Delta}\big){\mathcal T}^{2a\Delta}_{m_1+m_2+m_3}\Big)\nonumber\\&+h^{\Delta}_{p^{-1},q}(m_1,m_2,m_3),
\end{align*}
where 
\begin{align*}
V^3_{a\Delta}&= p^{-2a(1-\Delta )(m_1+m_2+m_3)}\Big(\big(p^{-a}-q^{a}\big)^{3\choose 2}\prod_{1\leq j < k \leq 3}\Big([m_j]_{p^{-a},q^{a}}-[m_k]_{p^{-a},q^{a}}\Big)\nonumber\\&+\prod_{1\leq j < k \leq 3}\Big(q^{a\,m_j}-q^{a\,m_k}\Big)\Big),
\end{align*}
\begin{eqnarray*}
	D^3_{a\Delta}
	=\Big(\big(p^{-a}-q^{a}\big)^{3\choose 2}\prod_{1\leq j < k \leq 3}\frac{q^{a\Delta(m_k+1)}}{p^{-a\Delta\,m_k}}\Big([m_k]_{p^{-a},q^{a}}-[m_j]_{p^{-a},q^{a}}+q^{am_k}-p^{-am_j}\Big),
\end{eqnarray*}
\begin{eqnarray*}
	N^3_{a\Delta}=
	\Big(\big(p^{-a}-q^{a}\big)^{3\choose 2}\prod_{1\leq j < k \leq 3}\frac{q^{a\Delta(m_j+1)}}{p^{-a\Delta\,m_j}}\Big([m_j]_{p^{-a},q^{a}}-[m_k]_{p^{-a},q^{a}}+q^{am_j}-p^{-am_k}\Big),
\end{eqnarray*}
and
\begin{align*}
h^{\Delta(a)}_{p^{-1},q}(m_1,m_2,m_3)&=q^{3a\Delta(m_1+m_2+m_3+1)}\prod_{1\leq j < k \leq 3}\bigg([-\Delta\,m_j]_{p^{-a},q^{a}}[m_k(1-\Delta)]_{p^{-a},q^{a}}\nonumber\\&-[-\Delta\,m_k]_{p^{-a},q^{a}}[m_j(1-\Delta)]_{p^{-a},q^{a}}\bigg).		
\end{align*}
\subsubsection{Conformal  Witt $n$-algebra induced by the Hounkonnou-Ngompe generalization of $q$-Quesne algebra  }
The conformal Witt $n$-algebra associated to the Hounkonnou-Ngompe generalization of $q$-Quesne algebra \cite{HN} is deduced as follows: We consider 
for $(\Delta \neq 0,1),$ the generalized $q$-Quesne conformal operator defined by:
\begin{eqnarray}\label{hnop1}
{\mathcal T}^{a\Delta}_m:=x^{(1-\Delta)(m+1)}{\mathcal D}^Q_{p^{a},q^{a}}\,x^{\Delta(m+1)},
\end{eqnarray}
where ${\mathcal D}^Q_{p^{a},q^{a}}$ is the generalized $q$-Quesne derivative :
\begin{eqnarray}\label{hnder}
\mathcal{D}^Q_{p^a,q^a}\phi(x)= \frac{\phi(p^a\,x)-\phi(q^{-a}\,x)}{(q^a-p^{-a})x}.
\end{eqnarray}
Then, by using the  generalized $q$-Quesne number
\begin{eqnarray}\label{hnnumber}
[n]^Q_{p^a,q^a}=\frac{p^{an}-q^{-an}}{-p^{-a}+q^a},
\end{eqnarray} the relation \eqref{hnop1} takes the form:
\begin{eqnarray}\label{hnopa}
{\mathcal T}^{a\Delta}_m=[\Delta(m+1)]^Q_{p^{a},q^{a}}\,x^{m}.
\end{eqnarray}
Thus, the generalized $q$-Quesne conformal operators (\ref{hnop1}) obeys the product relation: \begin{align*}
\mathcal{T}^{a\Delta}_{m}.\mathcal{T}^{b\Delta}_n&={\big(p^{a+b}-q^{-a-b}\big)p^{an(1-\Delta)}\over \big(-p^{-a}+q^{a}\big)\big(-p^{-b}+q^{b}\big)p^{b\Delta\,m}}\,\mathcal{T}^{(a+b)\Delta}_{m+n}+ f^{\Delta(a,b)}_{p,q}(m,n)\nonumber\\& - {q^{-a(\Delta(m+1)+n)}\over p^{b\Delta\,m}\big(-p^{-a}+q^{a}\big)}\,\mathcal{T}^{b\Delta}_{m+n}- {q^{-b\Delta(n+1)}\over p^{an(\Delta-1)}\big(-p^{-b}+q^{b}\big)}\, \mathcal{T}^{a\Delta}_{m+n},
\end{align*}
with \begin{eqnarray*}
	f^{\Delta(a,b)}_{p,q}(m,n)=q^{-(a+b)\Delta(m+n+1)}\,[-\Delta\,m]^Q_{p^{b},q^{b}}\,[n(1-\Delta)]^Q_{p^{a},q^{a}},
\end{eqnarray*}	
and the commutation relation
\begin{align}\label{hncrto}
\Big[\mathcal{T}^{a\Delta}_{m}, \mathcal{T}^{b\Delta}_n\Big]&={\big(p^{a+b}-q^{-a-b}\big)\over \big(-p^{-a}+q^{a}\big)\big(-p^{-b}+q^{b}\big)}\bigg(\frac{p^{an(1-\Delta)}}{p^{b\Delta\,m}}-\frac{p^{bm(1-\Delta)}}{p^{a\Delta\,n}}\bigg)\mathcal{T}^{(a+b)\Delta}_{m+n}\nonumber\\ &-\frac{q^{-b\Delta(n+1)}}{p^{a\Delta\,n}}\frac{\big(p^{an}-q^{-bm}\big)}{\big(-p^{-b}+q^{b}\big)}\mathcal{T}^{a\Delta}_{m+n}+h^{\Delta(a,b)}_{p,q}(m,n)\nonumber\\ &- \frac{q^{-a\Delta(m+1)}}{p^{b\Delta\,m}}\frac{\big(p^{bm}-q^{-an}\big)}{\big(-p^{-a}+q^{a}\big)}\mathcal{T}^{b\Delta}_{m+n},
\end{align}
where 
\begin{align*}
h^{\Delta(a,b)}_{p,q}(m,n)&=q^{-(a+b)\Delta(m+n+1)}\bigg([-\Delta\,m]^Q_{p^{b},q^{b}}\,[n(1-\Delta)]^Q_{p^{a},q^{a}}\nonumber\\&-[-\Delta\,n]^Q_{p^{a},q^{a}}\,[m(1-\Delta)]^Q_{p^{b},q^{b}}\bigg).
\end{align*}
Putting $a=b=1,$ in the relation \eqref{hncrto}, we obtain the commutation relation:
\begin{align*}
\Big[\mathcal{T}^{\Delta}_{m}, \mathcal{T}^{\Delta}_n\Big]&={p^{-\Delta(n+m)}\big(p^{n}-p^{m}\big)\over \big(-p^{-1}+q\big)}[2]^Q_{p,q}\mathcal{T}^{2\Delta}_{m+n}+h^{\Delta}_{p,q}(m,n)\nonumber\\
&-\bigg(\frac{q^{-\Delta(n+1)}\big(p^{n}-q^{-m}\big)}{p^{\Delta\,n}\big(-p^{-1}+q\big)}+ \frac{q^{-\Delta(m+1)}\big(p^{m}-q^{-n}\big)}{p^{\Delta\,m}\big(-p^{-1}+q\big)}\bigg)\mathcal{T}^{\Delta}_{m+n},
\end{align*}
where
\begin{eqnarray*}
	h^{\Delta(1,1)}_{p,q}(m,n)=q^{-2\Delta(m+n+1)}\bigg([-\Delta\,m]^Q_{p,q}\,[n(1-\Delta)]^Q_{p,q}-[-\Delta\,n]^Q_{p,q}\,[m(1-\Delta)]^Q_{p,q}\bigg).
\end{eqnarray*}
Besides, the $n$-bracket is defined by:
\begin{eqnarray*}
	\Big[{\mathcal T}^{a_1\Delta}_{m_1},\cdots,{\mathcal T}^{a_n\Delta}_{m_n}
	\Big]:=\epsilon^{i_1 \cdots i_n}_{1 \cdots n}\,{\mathcal T}^{a_{i_1}\Delta}_{m_{i_1}} \cdots {\mathcal T}^{a_{i_n}\Delta}_{m_{i_n}}.
\end{eqnarray*}
Our interest is	focussed on the special  case with the same $a\Delta.$ Then, 
\begin{eqnarray*}
	\Big[{\mathcal T}^{a\Delta}_{m_1},\cdots,{\mathcal T}^{a\Delta}_{m_n}
	\Big]=\epsilon^{1\cdots n}_{1\cdots n}\,{\mathcal T}^{a\Delta}_{m_1} \cdots {\mathcal T}^{a\Delta}_{m_n}.
\end{eqnarray*}
Putting $a=b$ in the relation (\ref{hncrto}),  the following commutation relation holds:
\begin{align*}
\Big[\mathcal{T}^{a\Delta}_{m}, \mathcal{T}^{a\Delta}_n\Big]&={p^{-a\Delta(n+m)}\big(p^{an}-p^{am}\big)\over -p^{-a}+q^{a}}[2]^Q_{p^a,q^a}\mathcal{T}^{2a\Delta}_{m+n}\nonumber\\ &-\frac{1}{\big(q-p^{-1}\big)}\bigg(\frac{q^{-a\Delta(n+1)}}{p^{a\Delta\,n}}\big(p^{an}-q^{-am}\big)\nonumber\\&- \frac{q^{-a\Delta(m+1)}}{p^{a\Delta\,m}}\big(p^{am}-q^{-an}\big)\bigg)\mathcal{T}^{a\Delta}_{m+n}+h^{\Delta(a)}_{p,q}(m,n),
\end{align*}
where \begin{align*}
h^{\Delta(a)}_{p,q}(m,n)&=q^{-2a\Delta(m+n+1)}\bigg([-\Delta\,m]^Q_{p^{a},q^{a}}[n(1-\Delta)]^Q_{p^{a},q^{a}}\nonumber\\&-[-\Delta\,n]^Q_{p^{a},q^{a}}[m(1-\Delta)]^Q_{p^{a},q^{a}}\bigg)
\end{align*}
The generalized $q$-Quesne conformal Witt $n$-algebra is given  by :
\begin{align}\label{hncrna}
\Big[{\mathcal T}^{a\Delta}_{m_1},\cdots,{\mathcal T}^{a\Delta}_{m_n}
\Big]&={(-1)^{n+1}\over \big(-p^{-a}+q^{a}\big)^{n-1}}\Big( V^n_{a\Delta}[n]^Q_{p^{a},q^{a}}{\mathcal T}^{n\,a\Delta}_{\bar{m}}\nonumber\\ &- [n-1]^Q_{p^{a},q^{a}}\big(D^n_{a\Delta}+ N^n_{a\Delta}\big){\mathcal T}^{a(n-1)\Delta}_{\bar{m}}\Big)\nonumber\\&+h^{\Delta}_{p,q}(m_1,\ldots,m_n),
\end{align}
where 
\begin{align*}
V^n_{a\Delta}&= p^{a(n-1)(1-\Delta )\bar{m}}\Big(\big(-p^{-a}+q^{a}\big)^{n\choose 2}\prod_{1\leq j < k \leq n}\Big([m_j]^Q_{p^{a},q^{a}}-[m_k]^Q_{p^{a},q^{a}}\Big)\nonumber\\&+\prod_{1\leq j < k \leq n}\Big(q^{-a\,m_j}-q^{-a\,m_k}\Big)\Big),
\end{align*}
\begin{eqnarray*}
	D^n_{a\Delta}
	&=&\Big(\big(-p^{-a}+q^{a}\big)^{n\choose 2}\prod_{1\leq j < k \leq n}\frac{q^{-a\Delta(m_k+1)}}{p^{a\Delta\,m_k}}\Big([m_k]^Q_{p^{a},q^{a}}-[m_j]^Q_{p^{a},q^{a}}\nonumber\\&+&q^{-am_k}-p^{am_j}\Big),
\end{eqnarray*}
\begin{align*}
N^n_{a\Delta}&=(-1)^{n+1}
\Big(\big(-p^{-a}+q^{a}\big)^{n\choose 2}\prod_{1\leq j < k \leq n}\frac{q^{-a\Delta(m_j+1)}}{p^{a\Delta\,m_j}}\Big([m_j]^Q_{p^{a},q^{a}}-[m_k]^Q_{p^{a},q^{a}}\nonumber\\&+q^{-am_j}-p^{am_k}\Big),
\end{align*}
and
\begin{align*}
h^{\Delta(a)}_{p,q}(m_1,\ldots,m_n)&=q^{-an\Delta(\bar{m}+1)}\prod_{1\leq j < k \leq n}\bigg([-\Delta\,m_j]^Q_{p^{a},q^{a}}[m_k(1-\Delta)]^Q_{p^{a},q^{a}}\nonumber\\&-[-\Delta\,m_k]^Q_{p^{a},q^{a}}[m_j(1-\Delta)]^Q_{p^{a},q^{a}}\bigg).
\end{align*}
Putting $n=3$ in the relation (\ref{hncrna}), we obtain the generalized  $q$-Quesne conformal Witt $3$-algebra:
\begin{align*}
\Big[{\mathcal T}^{a\Delta}_{m_1},{\mathcal T}^{a\Delta}_{m_2},{\mathcal T}^{a\Delta}_{m_3}
\Big]&={1\over \big(-p^{-a}+q^{a}\big)^{2}}\Big( V^3_{a\Delta}[3]^Q_{p^{a},q^{a}}{\mathcal T}^{3\,a\Delta}_{m_1+m_2+m_3}\nonumber\\ &- [2]^Q_{p^{a},q^{a}}\big(D^3_{a\Delta}+ N^3_{a\Delta}\big){\mathcal T}^{2a\Delta}_{m_1+m_2+m_3}\Big)\nonumber\\&+h^{\Delta}_{p,q}(m_1,m_2,m_3),
\end{align*}
where 
\begin{align*}
V^3_{a\Delta}&= p^{2a(1-\Delta )(m_1+m_2+m_3)}\Big(\big(-p^{-a}+q^{a}\big)^{3\choose 2}\prod_{1\leq j < k \leq 3}\Big([m_j]^Q_{p^{a},q^{a}}-[m_k]^Q_{p^{a},q^{a}}\Big)\nonumber\\&+\prod_{1\leq j < k \leq 3}\Big(q^{-a\,m_j}-q^{-a\,m_k}\Big)\Big),
\end{align*}
\begin{eqnarray*}
	D^3_{a\Delta}
	&=&\Big(\big(-p^{-a}+q^{a}\big)^{3\choose 2}\prod_{1\leq j < k \leq 3}\frac{q^{-a\Delta(m_k+1)}}{p^{a\Delta\,m_k}}\Big([m_k]^Q_{p^{a},q^{a}}-[m_j]^Q_{p^{a},q^{a}}\nonumber\\&+&q^{-am_k}-p^{am_j}\Big),
\end{eqnarray*}
\begin{eqnarray*}
	N^3_{a\Delta}&=&
	\Big(\big(-p^{a}+q^{a}\big)^{3\choose 2}\prod_{1\leq j < k \leq 3}\frac{q^{-a\Delta(m_j+1)}}{p^{a\Delta\,m_j}}\Big([m_j]^Q_{p^{a},q^{a}}-[m_k]^Q_{p^{a},q^{a}}\nonumber\\&+&q^{-am_j}-p^{am_k}\Big),
\end{eqnarray*}
and
\begin{align*}
h^{\Delta(a)}_{p,q}(m_1,m_2,m_3)&=q^{-3a\Delta(m_1+m_2+m_3+1)}\prod_{1\leq j < k \leq 3}\bigg([-\Delta\,m_j]^Q_{p^{a},q^{a}}[m_k(1-\Delta)]^Q_{p^{a},q^{a}}\nonumber\\&-[-\Delta\,m_k]^Q_{p^{a},q^{a}}[m_j(1-\Delta)]^Q_{p^{a},q^{a}}\bigg).		
\end{align*}
\subsection{A toy model for $\mathcal{R}(p,q)$-conformal Virasoro constraints}\label{sub4.2}
In this section, we study a toy model for the ${\mathcal R}(p,q)$-Virasoro constraints with conformal dimension $\Delta.$ We consider the generating function with infinitely many	parameters
introduced as follows \cite{NZ}: $$Z^{toy}(t)=\int \, \,x^{\gamma}\,\exp\left(\displaystyle\sum_{s=0}^{\infty}{t_s\over s!}x^s\right)\,dx,$$
which encodes many different integrals.
We assume that the following property holds for the ${\mathcal R}(p,q)$-derivative
\begin{eqnarray*}
	\int_{{\mathbb R}} x^{(m+1)(1-\Delta)}{\mathcal D}_{{\mathcal R}(p^{a},q^{a})}f(x)d\,x
	%&={h(p^{a},q^{a})\over \tau^{a}_1-\tau^{a}_2}\Big(\int_{-\infty}^{+\infty}x^{(m+1)(1-\Delta)}{\mathcal D}_{{\mathcal R}(p^{a},q^{a})}{f(\tau^{a}_1\,x)\over x}dx\nonumber\\& -\int_{-\infty}^{+\infty}x^{(m+1)(1-\Delta)}{\mathcal D}_{{\mathcal R}(p^{a},q^{a})}{f(\tau^{a}_2\,x)\over x}dx\Big)\\
	=0,
\end{eqnarray*} 
where 
$$h(p^{a},q^{a})={p^{a}-q^{a}\over p^{P^{a}}-q^{Q^{a}}}{\mathcal R}\big(p^{P^{a}},q^{Q^{a}}\big).$$
For $f(x)=x^{\Delta(m+1)+\gamma}\,\exp\left(\displaystyle\sum_{s=0}^{\infty}{t_s\over s!}x^s\right),$ we have
\begin{eqnarray*}
	\int_{-\infty}^{+\infty}x^{(m+1)(1-\Delta)}{\mathcal D}_{{\mathcal R}(p^{a},q^{a})}\left(x^{\Delta(m+1)+\gamma}\,\exp\left(\sum_{s=0}^{\infty}{t_s\over s!}x^s\right)\right)d\,x=0.
\end{eqnarray*}
We consider  the following relation:
\begin{eqnarray*}
	\exp\left(\displaystyle\sum_{s=0}^{\infty}{t_s\over s!}x^s\right)=\sum_{n=0}^{\infty}B_n(t_1,\cdots,t_n){x^n\over n!},
\end{eqnarray*}
where $B_n$ is the Bell polynomials. Then, from the $\mathcal{R}(p,q)$-Leibniz rule, we get:
\begin{align*}
{\mathcal D}_{{\mathcal R}(p^{a},q^{a})}\bigg(x^{\Delta(m+1)+\gamma}\exp\bigg(&\displaystyle\sum_{s=0}^{\infty}{t_s\over s!}x^s\bigg)\bigg)={\mathcal D}_{{\mathcal R}(p^{a},q^{a})}\big(x^{\Delta(m+1)+\gamma}\big)\nonumber\\&\times\exp\left(\displaystyle\sum_{s=0}^{\infty}{t_s\over s!}(x\tau_1)^{as}\right)+\frac{h(p^{a},q^{a})}{(\tau^{a}_1 - \tau^{a}_2)}\nonumber\\ &\times {\tau^{a(\Delta(m+1)+\gamma)}_2}\sum_{k=1}^{\infty}{B_k(t^{a}_1,\cdots,t^{a}_k)\over k!}\nonumber\\&\times x^{\Delta(k+m)+\gamma}\exp\bigg(\displaystyle\sum_{s=0}^{\infty}{t_s\over s!}x^s\bigg),
\end{align*}
where $t^{a}_k=(\tau^{a\,k}_1-\tau^{a\,k}_2)t_k.$ 
After computation, we have: 
\begin{align*}
x^{(m+1)(1-\Delta)}{\mathcal D}_{{\mathcal R}(p^{a},q^{a})}&\left(x^{\Delta(m+1)+\gamma}\,\exp\left(\displaystyle\sum_{s=0}^{\infty}{t_s\over s!}x^s\right)\right)\\&={x^{m+\gamma}\,[\Delta(m+1)+\gamma]_{{\mathcal R}(p^{a},q^{a})}\over \tau^{-a\,m}_1}\exp\left(\displaystyle\sum_{s=0}^{\infty}{t_s\over s!}x^s\right)\nonumber\\ &+ {h(p^{a},q^{a})\tau^{a(\Delta(m+1)+\gamma)}_2\over (\tau^{a}_1 - \tau^{a}_2)}\sum_{k=1}^{\infty}{B_k(t^{a}_1,\cdots,t^{a}_k)\over k!}\nonumber\\&\times x^{m+\gamma+k}\exp\left(\displaystyle\sum_{s=0}^{\infty}{t_s\over s!}x^s\right),
\end{align*}
Then,  from the constraints on the partition function,
$${\mathcal T}^{a\Delta}_m\,Z^{(toy)}(t)=0,\quad m\geq 0,$$
we obtain:
\begin{lemma} The $\mathcal{R}(p,q)$-conformal differential  operator is determined by:
	\begin{align}\label{Rpqop}
	{\mathcal T}^{a\Delta}_m&=[\Delta(m+1)+\gamma]_{{\mathcal R}(p^{a},q^{a})}\,m!\, \tau^{a\,m}_1\,{\partial\over \partial t_m}\nonumber\\ &+ h(p^{a},q^{a}){\tau^{a(\Delta(m+1)+\gamma)}_2\over \tau^{a}_1 - \tau^{a}_2}\sum_{k=1}^{\infty}{(k+m)!\over k!}B_k(t^{a}_1,\ldots,t^{a}_k){\partial\over \partial t_{k+m}}.
	\end{align}
\end{lemma}
Note that taking $\Delta=1,$ we obtained the result given in \cite{melong2022}.

Furthemore, there exist another way to determine the $\mathcal{R}(p,q)$-conformal differential operator. The result is given in the following lemma.
\begin{lemma}\label{lemdiffop}
	The $\mathcal{R}(p,q)$-conformal differential operator is given by:
	\begin{align}\label{rpqop2}
	{\mathcal T}^{a\Delta}_m&=[x\partial_x+\Delta(m+1)-m]_{{\mathcal R}(p^{a},q^{a})}\,m!\, \tau^{a\,m}_1\,{\partial\over \partial t_m}\nonumber\\ &+ h(p^{a},q^{a}){\tau^{a(\Delta(m+1)+\gamma)}_2\over \tau^{a}_1 - \tau^{a}_2}\sum_{k=1}^{\infty}{(m+1+\Delta(k-1))!\over k!}\nonumber\\ &\times B_k(t^{a}_1,\cdots,t^{a}_k){\partial\over \partial t_{m+1+\Delta(k-1)}}.
	\end{align}
\end{lemma}
\begin{proof}
	By using the definition of the $\mathcal{R}(p,q)$-derivative 
	\begin{eqnarray*}
		{\mathcal D}_{{\mathcal R}(p^{a},q^{a})}= \frac{1}{x}[x\partial_x]_{{\mathcal R}(p^{a},q^{a})},
	\end{eqnarray*}
	the conformal operators \eqref{rpqop1} may be rewritten in the form:
	\begin{eqnarray}\label{rpqopb}
	{\mathcal T}^{a\Delta}_m=[x\partial_x+\Delta(m+1)-m]_{{\mathcal R}(p^{a},q^{a})}\,x^{m}.
	\end{eqnarray} 
	Then, using step by step the same technique to construct the operator \eqref{Rpqop}, the result follows. 
\end{proof}	
\begin{corollary}
	The $q$-conformal differential operator is given by:
	\begin{eqnarray*}\label{qop2}
		{\mathcal T}^{a\Delta}_m&=&[x\partial_x+\Delta(m+1)-m]_{q^{a}}\,m!\, q^{a\,m}\,{\partial\over \partial t_m}\nonumber\\ &+& {h(q^{a})\over q^{a} - 1}\sum_{k=1}^{\infty}{(m+1+\Delta(k-1))!\over k!} B_k(t^{a}_1,\cdots,t^{a}_k){\partial\over \partial t_{m+1+\Delta(k-1)}}.
	\end{eqnarray*}
\end{corollary}
\begin{proof}
	The result is obtained by taking $\tau_1=q$ and $\tau_2=1.$
\end{proof}
\begin{remark}
	For $\Delta=0$ and $\Delta=1,$ we obtain, respectively, the following $\mathcal{R}(p,q)$-differential operators:
	\begin{align}\label{Rpqop0}
	{\mathcal T}_m&=[x\partial_x-m]_{{\mathcal R}(p^{a},q^{a})}\,m!\, \tau^{a\,m}_1\,{\partial\over \partial t_m}\nonumber\\ &+ h(p^{a},q^{a}){\tau^{a\,\gamma}_2\over \tau^{a}_1 - \tau^{a}_2}\sum_{k=1}^{\infty}{(m+1)!\over k!} B_k(t^{a}_1,\cdots,t^{a}_k){\partial\over \partial t_{m+1}}
	\end{align}
	and
	\begin{align}\label{Rpqop1}
	{\mathcal T}^{a}_m&=[x\partial_x+1]_{{\mathcal R}(p^{a},q^{a})}\,m!\, \tau^{a\,m}_1\,{\partial\over \partial t_m}\nonumber\\&+ h(p^{a},q^{a}){\tau^{a(m+1+\gamma)}_2\over \tau^{a}_1 - \tau^{a}_2}\sum_{k=1}^{\infty}{(m+k)!\over k!} B_k(t^{a}_1,\cdots,t^{a}_k){\partial\over \partial t_{m+k}}.
	\end{align}
\end{remark}
\begin{remark}
	The conformal differential operators induced by quantum algebras in the literature are deduced as follows:
	\begin{enumerate}
		\item[(a)] The $q$-conformal differential operators are given by:
		\begin{eqnarray}\label{qop}
		{\mathcal T}^{a\Delta}_m&=&[\Delta(m+1)+\gamma]_{q^{a}}\,m!\, q^{-a\,m}\,{\partial\over \partial t_m}\nonumber\\ &+& {q^{-a(\Delta(m+1)+\gamma)}\over q^{a} - q^{-a}}\sum_{k=1}^{\infty}{(m+k)!\over k!} B_k(t^{a}_1,\cdots,t^{a}_k){\partial\over \partial t_{m+k}}
		\end{eqnarray}
		and
		\begin{align}\label{qopp}
		{\mathcal T}^{a\Delta}_m&=[x\partial_x+\Delta(m+1)-m]_{q^{a}}\,m!\, \tau^{a\,m}_1\,{\partial\over \partial t_m}\nonumber\\ &+ h(q^{a}){q^{-a(\Delta(m+1)+\gamma)}\over q^{a} - q^{-a}}\sum_{k=1}^{\infty}{(m+1+\Delta(k-1))!\over k!}\nonumber\\ &\times B_k(t^{a}_1,\cdots,t^{a}_k){\partial\over \partial t_{m+1+\Delta(k-1)}}.
		\end{align}
		\item[(b)]The $(p,q)$-conformal differential operators are determined as follows: \begin{eqnarray}\label{pqop}
		{\mathcal T}^{a\Delta}_m&=&[\Delta(m+1)+\gamma]_{p^{a},q^{a}}\,m!\, p^{-a\,m}\,{\partial\over \partial t_m}\nonumber\\ &+& {q^{a(\Delta(m+1)+\gamma)}\over p^{a} - q^{a}}\sum_{k=1}^{\infty}{(m+k)!\over k!} B_k(t^{a}_1,\cdots,t^{a}_k){\partial\over \partial t_{m+k}}
		\end{eqnarray}
		and \begin{align}\label{pqop2}
		{\mathcal T}^{a\Delta}_m&=[x\partial_x+\Delta(m+1)-m]_{p^{a},q^{a}}\,m!\, p^{a\,m}\,{\partial\over \partial t_m}\nonumber\\ &+ h(p^{a},q^{a}){q^{a(\Delta(m+1)+\gamma)}\over p^{a} - q^{a}}\sum_{k=1}^{\infty}{(m+1+\Delta(k-1))!\over k!}\nonumber\\ &\times B_k(t^{a}_1,\cdots,t^{a}_k){\partial\over \partial t_{m+1+\Delta(k-1)}}.
		\end{align}
		\item[(b)]The $(p^{-1},q)$-conformal differential operators are presented by: \begin{eqnarray}\label{cjop}
		{\mathcal T}^{a\Delta}_m&=&[\Delta(m+1)+\gamma]_{p^{-a},q^{a}}\,m!\, p^{a\,m}\,{\partial\over \partial t_m}\nonumber\\ &+& {q^{a(\Delta(m+1)+\gamma)}\over p^{-a} - q^{a}}\sum_{k=1}^{\infty}{(m+k)!\over k!} B_k(t^{a}_1,\cdots,t^{a}_k){\partial\over \partial t_{m+k}}
		\end{eqnarray}
		and \begin{align}\label{cjop2}
		{\mathcal T}^{a\Delta}_m&=[x\partial_x+\Delta(m+1)-m]_{p^{-a},q^{a}}\,m!\, p^{-a\,m}\,{\partial\over \partial t_m}\nonumber\\ &+ h(p^{-a},q^{a}){q^{a(\Delta(m+1)+\gamma)}\over p^{-a} - q^{a}}\sum_{k=1}^{\infty}{(m+1+\Delta(k-1))!\over k!}\nonumber\\ &\times B_k(t^{a}_1,\cdots,t^{a}_k){\partial\over \partial t_{m+1+\Delta(k-1)}}.
		\end{align}
		\item[(d)]The generalized $q$-Quesne conformal differential operators are derived as:\begin{eqnarray}\label{qQop}
		{\mathcal T}^{a\Delta}_m&=&[\Delta(m+1)+\gamma]^Q_{p^{a},q^{a}}\,m!\, p^{-a\,m}\,{\partial\over \partial t_m}\nonumber\\ &+& {q^{-a(\Delta(m+1)+\gamma)}\over p^{a} - q^{-a}}\sum_{k=1}^{\infty}{(m+k)!\over k!} B_k(t^{a}_1,\cdots,t^{a}_k){\partial\over \partial t_{m+k}}
		\end{eqnarray}
		and
		\begin{align}\label{qQop2}
		{\mathcal T}^{a\Delta}_m&=[x\partial_x+\Delta(m+1)-m]^Q_{p^{a},q^{a}}\,m!\, p^{a\,m}\,{\partial\over \partial t_m}\nonumber\\ &+ h(p^{a},q^{a}){q^{-a(\Delta(m+1)+\gamma)}\over p^{a} - q^{-a}}\sum_{k=1}^{\infty}{(m+1+\Delta(k-1))!\over k!}\nonumber\\ &\times B_k(t^{a}_1,\cdots,t^{a}_k){\partial\over \partial t_{m+1+\Delta(k-1)}}.
		\end{align}
	\end{enumerate}
\end{remark}
\section{ $\mathcal{R}(p,q)$-super Virasoro $n$-algebra with conformal dimension $(\Delta\neq 0,1)$} This section is reserved to the characterization of the super Virasoro algebra with conformal dimension $(\Delta \neq 0,1).$ Related particular cases are deduced.
Using step by step  the procedure presented to construct the ${\mathcal R}(p,q)$-conformal Virasoro $2n$-algebra (\ref{nalg}), we deduce  the  $\mathcal{R}(p,q)$-conformal super Virasoro $2n$-algebra. It's generated by the  bosonic and fermionic operators $\mathcal{L}^{\Delta}_m$ of parity $0$ and $\mathcal{G}^{\Delta}_m$ of parity $1$ obeying the following commutation relations: 
\begin{eqnarray}
\big[\mathcal{L}^{\Delta}_{m_1},\ldots, \mathcal{L}^{\Delta}_{m_{2n}}\big]_{\mathcal{R}(p,q)}=f^{\Delta}_{\mathcal{R}(p,q)}(m_1,\ldots,m_{2n})\mathcal{L}^{\Delta}_{\tilde{m}}+ \tilde{C}^{\Delta}_{\mathcal{R}(p,q)}(m_1,\ldots,m_{2n}),
\end{eqnarray}
\begin{eqnarray}\label{sV2na}
\big[\mathcal{L}^{\Delta}_{m_1},\cdots, \mathcal{G}^{\mathcal{R}(p,q)}_{m_{2n}}\big]_{{\mathcal R}(p,q)}= g^{\Delta}_{\mathcal{R}(p,q)}(m_1,\cdots m_{2n})+ {\mathcal CS}^{\Delta}_{{\mathcal R}(p,q)}(m_1,\cdots m_{2n}),
\end{eqnarray}
where $f^{\Delta}_{\mathcal{R}(p,q)}(m_1,\cdots,m_{2n})$ and $C^{\Delta}_{\mathcal{R}(p,q)}(m_1,\cdots,m_{2n})$ are given by the relations (\ref{cnalg1}), (\ref{cv}), 
\begin{align*}
g^{\Delta}_{\mathcal{R}(p,q)}(m_1,m_2,\cdots m_{2n})&={\big(q-p\big)^{{2n-1\choose 2}}\over \big(\tau_1\tau_2\big)^{-(n-1) \tilde{m}+1}}\bigg({[-2\tilde{m}-1]_{{\mathcal R}(p,q)}\over 2[\tilde{m}-1]_{{\mathcal R}(p,q)}}\bigg)\nonumber\\&\times\prod_{1\leq i<j\leq 2n-1} \big([m_i]_{{\mathcal R}(p,q)}- [m_j]_{{\mathcal R}(p,q)}\big)\nonumber\\&\times
\prod_{i=1}^{2n-1} \big([m_i]_{{\mathcal R}(p,q)}- [m_{2n}+1]_{{\mathcal R}(p,q)}\big)\mathcal{G}^{\Delta}_{\tilde{m}},
\end{align*}
\begin{align*}
{\mathcal CS}^{\Delta}_{\mathcal{R}(p,q)}(m_1,m_2,\cdots m_{2n})&=\sum_{k=1}^{2n-1}{(-1)^{k+1}c(p,q)(\tau_1\tau_2)^{-m_k}\over 6\times 2^{n-1}(n-1)!}{[m_{k}]_{\mathcal{R}(p,q)}\over [2m_{k}]_{{\mathcal R}(p,q)}}\nonumber\\&\times[m_k+1]_{{\mathcal R}(p,q)}[m_k]_{{\mathcal R}(p,q)}[m_k-1]_{{\mathcal R}(p,q)}\delta_{m_k+m_{2n}+1,0}\nonumber\\&\times\epsilon^{i_1\cdots i_{2n-2}}_{j_1\cdots j_{2n-2}}\prod_{s=1}^{n-1}{(\tau_1\tau_2)^{-i_{2s-1}}[i_{2s-1}]_{{\mathcal R}(p,q)}\over [2\,i_{2s-1}]_{{\mathcal R}(p,q)}}\nonumber\\&\times[i_{2s-1}+1]_{{\mathcal R}(p,q)}\,[i_{2s-1}]_{{\mathcal R}(p,q)}\,[i_{2s-1}-1]_{{\mathcal R}(p,q)}\delta_{i_{2s-1}+i_{2s},0},
\end{align*}
with $\{j_1,\cdots, j_{2n-2}\}=\{1,\cdots,\hat{k},\cdots,2n-1\}$ and other anti-commutators are zeros.
\subsection{Another $\mathcal{R}(p,q)$-super Witt $n$-algebra with conformal dimension $(\Delta\neq 0,1)$}In this section, we construct another   super Witt $n$-algebra with conformal dimension $(\Delta\neq 0,1)$ induced by the $\mathcal{R}(p,q)$-deformed quantum algebra \cite{HB}.
\begin{definition}
	For $(\Delta \neq 0,1),$ the $\mathcal{R}(p,q)$-super conformal operators are defined by:
	\begin{align}
	{\mathcal T}^{a\Delta}_m&:=x^{(1-\Delta)(m+1)}\bar{\Delta}\,x^{\Delta(m+1)}\label{rpqsop1},\\
	{\mathbb T}^{a\Delta}_m&:=\theta\,x^{(1-\Delta)(m+1)}\bar{\Delta}\,x^{\Delta(m+1)}\label{rpqsop2},
	\end{align}
	where $\bar{\Delta}$ is given in the relation \eqref{deltaxyrpq}.
\end{definition}
The conformal super operators \eqref{rpqsop1} and \eqref{rpqsop2} can be rewritten as:
\begin{align}
{\mathcal T}^{a\Delta}_m&=[\Delta(m+1)]_{{\mathcal R}(p^{a},q^{a})}\,x^{m}\label{rpqsopa}\\
\mathbb{T}^{a\Delta}_m&=\theta\,[\Delta(m+1)]_{{\mathcal R}(p^{a},q^{a})}\,x^{m}\label{rpqsopb}.
\end{align}
\begin{proposition}
	The conformal super operators \eqref{rpqsop1} and \eqref{rpqsop2} satisfy the following
	product relation: \begin{align}\label{spre}
	\mathcal{T}^{a\Delta}_{m}.\mathbb{T}^{b\Delta}_n&={\big(\tau^{a+b}_1-\tau^{a+b}_2\big)\tau^{a(n(1-\Delta)+1)}_1\over \big(\tau^{a}_1-\tau^{a}_2\big)\big(\tau^{b}_1-\tau^{b}_2\big)\tau^{b\Delta\,m}_1}\,\mathbb{T}^{(a+b)\Delta}_{m+n}+ \mathbb{F}^{\Delta(a,b)}_{\mathcal{R}(p,q)}(m,n)\nonumber\\&- {\tau^{b\Delta(n+1)}_2\over \tau^{a(n(\Delta-1)-1)}_1\big(\tau^{b}_1-\tau^{b}_2\big)}\, \mathbb{T}^{a\Delta}_{m+n}- {\tau^{a(\Delta(m+1)+n+1)}_2\over \tau^{b\Delta\,m}_1\big(\tau^{a}_1-\tau^{a}_2\big)}\,\mathbb{T}^{b\Delta}_{m+n},
	\end{align}
	with 
	\begin{eqnarray*}
		\mathbb{F}^{\Delta(a,b)}_{\mathcal{R}(p,q)}(m,n)=\tau^{(a+b)\Delta(m+n+1)}_2\,[n(1-\Delta)+1]_{\mathcal{R}(p^{a},q^{a})}\,[-\Delta\,m]_{\mathcal{R}(p^{b},q^{b})}
	\end{eqnarray*}
	and the commutation relation
	\begin{align}\label{scrto}
	\Big[\mathcal{T}^{a\Delta}_{m}, \mathbb{T}^{b\Delta}_n\Big]&={\big(\tau^{a+b}_1-\tau^{a+b}_2\big)\over \big(\tau^{a}_1-\tau^{a}_2\big)\big(\tau^{b}_1-\tau^{b}_2\big)}\bigg(\frac{\tau^{a(n(1-\Delta)+1)}_1}{\tau^{b\Delta\,m}_1}-\frac{\tau^{b(m(1-\Delta)+1)}_1}{\tau^{a\Delta\,n}_1}\bigg)\mathbb{T}^{(a+b)\Delta}_{m+n}\nonumber\\ &-\frac{\tau^{b\Delta(n+1)}_2}{\tau^{a\Delta\,n}_1}\frac{\big(\tau^{a(n+1)}_1-\tau^{b(m+1)}_2\big)}{\big(\tau^{b}_1-\tau^{b}_2\big)}\mathbb{T}^{a\Delta}_{m+n}+\mathbb{H}^{\Delta(a,b)}_{\mathcal{R}(p,q)}(m,n)\nonumber\\&- \frac{\tau^{a\Delta(m+1)}_2}{\tau^{b\Delta\,m}_1}\frac{\big(\tau^{b(m+1)}_1-\tau^{a(n+1)}_2\big)}{\big(\tau^{a}_1-\tau^{a}_2\big)}\mathbb{T}^{b\Delta}_{m+n},
	\end{align}
	where 
	\begin{align*}
	\mathbb{H}^{\Delta(a,b)}_{\mathcal{R}(p,q)}(m,n)&=\tau^{(a+b)\Delta(m+n+1)}_2\bigg([-\Delta\,m]_{\mathcal{R}(p^{b},q^{b})}\,[n(1-\Delta)+1]_{\mathcal{R}(p^{a},q^{a})}\nonumber\\&-[-\Delta\,n]_{\mathcal{R}(p^{a},q^{a})}\,[m(1-\Delta)+1]_{\mathcal{R}(p^{b},q^{b})}\bigg)
	\end{align*}
\end{proposition}
\begin{proof}
	By using  the relations \eqref{rpqsopa}, \eqref{rpqsopb} and after computation,  we get
	\begin{eqnarray*}
		\mathcal{T}^{a\Delta}_{m}.\mathbb{T}^{b\Delta}_n=[\Delta(n+1)]_{{\mathcal R}(p^{b},q^{b})}[\Delta(m+1)+n+1]_{{\mathbb R}(p^{a},q^{a})}\,x^{m+n}.
	\end{eqnarray*}
	From  the relation \eqref{rpqn}, the results follows. 
\end{proof}
\begin{remark} 
	Setting $a=b=1,$ in the relation \eqref{scrto}, we obtain the following commutation relation:
	\begin{align*}
	\Big[\mathcal{T}^{\Delta}_{m}, \mathbb{T}^{\Delta}_n\Big]&={\tau^{-\Delta(n+m)}_1\big(\tau^{n+1}_1-\tau^{m+1}_1\big)\over \big(\tau_1-\tau_2\big)}[2]_{\mathcal{R}(p,q)}\mathbb{T}^{2\Delta}_{m+n}+\mathbb{H}^{\Delta(1,1)}_{\mathcal{R}(p,q)}(m,n)\nonumber\\
	&-\bigg(\frac{\tau^{\Delta(n+1)}_2\big(\tau^{n+1}_1-\tau^{m+1}_2\big)}{\tau^{\Delta\,n}_1\big(\tau_1-\tau_2\big)}+ \frac{\tau^{\Delta(m+1)}_2\big(\tau^{m+1}_1-\tau^{n+1}_2\big)}{\tau^{\Delta\,m}_1\big(\tau_1-\tau_2\big)}\bigg)\mathbb{T}^{\Delta}_{m+n},
	\end{align*}
	where
	\begin{align*}
	\mathbb{H}^{\Delta(1,1)}_{\mathcal{R}(p,q)}(m,n)&=\tau^{2\Delta(m+n+1)}_2\bigg([-\Delta\,m]_{\mathcal{R}(p,q)}\,[n(1-\Delta)+1]_{\mathcal{R}(p,q)}\nonumber\\&-[-\Delta\,n]_{\mathcal{R}(p,q)}\,[m(1-\Delta)+1]_{\mathcal{R}(p,q)}\bigg).
	\end{align*}
\end{remark}
\begin{corollary}The $q$-conformal super operators:
	\begin{align*}
	{\mathcal T}^{a\Delta}_m&=[\Delta(m+1)]_{q^{a}}\,z^{m}\\
	\mathbb{T}^{a\Delta}_m&=\theta\,[\Delta(m+1)]_{q^{a}}\,z^{m}
	\end{align*}
	satisfy the commutation relation
	\begin{align*}
	\Big[\mathcal{T}^{a\Delta}_{m}, \mathbb{T}^{b\Delta}_n\Big]&={\big(q^{a+b}-1\big)\over \big(q^{a}-1\big)\big(q^{b}-1\big)}\bigg(\frac{q^{a(n(1-\Delta)+1)}}{q^{b\Delta\,m}}-\frac{q^{b(m(1-\Delta)+1)}}{q^{a\Delta\,n}}\bigg)\mathbb{T}^{(a+b)\Delta}_{m+n}\nonumber\\ &-\frac{1}{q^{a\Delta\,n}}\frac{\big(q^{a(n+1)}-1\big)}{\big(q^{b}-1\big)}\mathbb{T}^{a\Delta}_{m+n}+\mathbb{H}^{\Delta(a,b)}_{q}(m,n)\nonumber\\&- \frac{1}{q^{b\Delta\,m}}\frac{\big(q^{b(m+1)}-1\big)}{\big(q^{a}-1\big)}\mathbb{T}^{b\Delta}_{m+n},
	\end{align*}
	where 
	\begin{eqnarray*}
		\mathbb{H}^{\Delta(a,b)}_{q}(m,n)=\bigg([-\Delta\,m]_{q^{b}}\,[n(1-\Delta)+1]_{q^{a}}-[-\Delta\,n]_{q^{a}}\,[m(1-\Delta)+1]_{q^{b}}\bigg).
	\end{eqnarray*}
\end{corollary}
From the super multibracket of order $n$, we have:
\begin{definition}
	The $\mathcal{R}(p,q)$-deformed super $n$-bracket is defined as follows:
	\begin{align*}
	\big[\mathcal{T}^{a\Delta}_{m_1},\mathcal{T}^{a\Delta}_{m_2},\cdots, \mathbb{T}^{a\Delta}_{m_n}\big]&:=\sum_{j=0}^{n-1}(-1)^{n-1+j}\epsilon^{i_1\ldots i_{n-1}}_{12\cdots n-1}\mathcal{T}^{a\Delta}_{m_{i_1}}\cdots \mathcal{T}^{a\Delta}_{m_{i_j}}\nonumber\\&\times
	\mathbb{T}^{a\Delta}_{m_{n}}\mathcal{T}^{a\Delta}_{m_{i_{j+1}}}\cdots \mathcal{T}^{a\Delta}_{m_{i_{n-1}}}.
	\end{align*}
\end{definition}

By using the relation (\ref{scrto}) with $a=b,$ we get:
\begin{align*}
\Big[\mathcal{T}^{a\Delta}_{m}, \mathbb{T}^{a\Delta}_n\Big]&={\tau^{-a\Delta(n+m)}_1\big(\tau^{a(n+1)}_1-\tau^{a(m+1)}_1\big)\over \big(\tau^{a}_1-\tau^{a}_2\big)}[2]_{\mathcal{R}(p^{a},q^{a})}\mathbb{T}^{2\Delta}_{m+n}\nonumber\\
&+\mathbb{H}^{\Delta(a)}_{\mathcal{R}(p,q)}(m,n)-\bigg(\frac{\tau^{a\Delta(n+1)}_2\big(\tau^{a(n+1)}_1-\tau^{a(m+1)}_2\big)}{\tau^{a\Delta\,n}_1\big(\tau^{a}_1-\tau^{a}_2\big)}\nonumber\\&+ \frac{\tau^{a\Delta(m+1)}_2\big(\tau^{a(m+1)}_1-\tau^{a(n+1)}_2\big)}{\tau^{a\Delta\,m}_1\big(\tau^{a}_1-\tau^{a}_2\big)}\bigg)\mathbb{T}^{a\Delta}_{m+n},
\end{align*}
where
\begin{align*}
\mathbb{H}^{\Delta(a)}_{\mathcal{R}(p,q)}(m,n)&=\tau^{2a\Delta(m+n+1)}_2\bigg([-\Delta\,m]_{\mathcal{R}(p^{a},q^{a})}\,[n(1-\Delta)+1]_{\mathcal{R}(p^{a},q^{a})}\nonumber\\&-[-\Delta\,n]_{\mathcal{R}(p^{a},q^{a})}\,[m(1-\Delta)+1]_{\mathcal{R}(p^{a},q^{a})}\bigg).
\end{align*}
\begin{proposition}
	The $\mathcal{R}(p,q)$-conformal super Witt $n$ algebra is generated by relation \eqref{crna} and the following commutation relation:
	\begin{align*}
	\Big[\mathcal{T}^{a\Delta}_{m_1},\cdots, \mathbb{T}^{a\Delta}_{m_n}\Big]&={(-1)^{n+1}\over \big(\tau^{a}_1-\tau^{a}_2\big)^{n-1}}\Big( \mathbb{V}^n_{a\Delta}[n]_{{\mathcal R}(p^{a},q^{a})}\mathbb{T}^{an\Delta}_{\bar{m}}\nonumber\\ &- [n-1]_{{\mathcal R}(p^{a},q^{a})}\big(\mathbb{D}^n_{a\Delta}+ \mathbb{N}^n_{a\Delta}\big){\mathbb T}^{a(n-1)\Delta}_{\bar{m}}\Big)\nonumber\\&+ \mathbb{H}^{\Delta(a)}_{\mathcal{R}(p,q)}\big(m_{1},\cdots, m_{n}\big),
	\end{align*}
	where 
	\begin{align*}
	\mathbb{V}^n_{a\Delta}&= \tau^{a(n-1)(1-\Delta )\bar{m}}_1\Big(\big(\tau^{a}_1-\tau^{a}_2\big)^{n\choose 2}\prod_{1\leq j < k \leq n}\Big([m_j+1]_{{\mathcal R}(p^{a},q^{a})}\nonumber\\&-[m_k+1]_{{\mathcal R}(p^{a},q^{a})}\Big)+\prod_{1\leq j < k \leq n}\Big(\tau^{a(m_j+1)}_2-\tau^{a(m_k+1)}_2\Big)\Big),
	\end{align*}
	\begin{align*}
	\mathbb{D}^n_{a\Delta}
	&=\Big(\big(\tau^{a}_1-\tau^{a}_2\big)^{n\choose 2}\prod_{1\leq j < k \leq n}\frac{\tau^{a\Delta(m_k+1)}_2}{\tau^{a\Delta\,m_k}_1}\Big([m_k+1]_{{\mathcal R}(p^{a},q^{a})}\nonumber\\&-[m_j+1]_{{\mathcal R}(p^{a},q^{a})}+\tau^{a(m_k+1)}_2-\tau^{a(m_j+1)}_1\Big),
	\end{align*}
	\begin{align*}
	\mathbb{N}^n_{a\Delta}&=(-1)^{n+1}
	\Big(\big(\tau^{a}_1-\tau^{a}_2\big)^{n\choose 2}\prod_{1\leq j < k \leq n}\frac{\tau^{a\Delta(m_j+1)}_2}{\tau^{a\Delta\,m_j}_1}\Big([m_j+1]_{{\mathcal R}(p^{a},q^{a})}\nonumber\\&-[m_k+1]_{{\mathcal R}(p^{a},q^{a})}+\tau^{a(m_j+1)}_2-\tau^{a(m_k+1)}_1\Big),
	\end{align*}
	and
	\begin{align*}
	\mathbb{H}^{\Delta(a)}_{\mathcal{R}(p,q)}(m_1,\ldots,m_n)&=\prod_{1\leq j < k \leq n}\bigg([-\Delta\,m_j]_{{\mathcal R}(p^{a},q^{a})}[m_k(1-\Delta)+1]_{{\mathcal R}(p^{a},q^{a})}\nonumber\\&-[-\Delta\,m_k]_{{\mathcal R}(p^{a},q^{a})}[m_j(1-\Delta)+1]_{{\mathcal R}(p^{a},q^{a})}\bigg)\tau^{an\Delta(\bar{m}+1)}_2.
	\end{align*}
\end{proposition}
\subsection{Conformal super Witt $n$-algebra and quantum algebras}
In this section, the conformal super Witt $n$-algebra corresponding to quantum algebra existing in the literature are derived.
\subsubsection{Conformal super Witt $n$-algebra associated to the Biedenharn-Macfarlane algebra }
The conformal Witt $n$-algebra induced by the Biedenharn-Macfarlane algebra \cite{BC,M} is derived by:  
for $(\Delta \neq 0,1),$ the $q$-conformal super operators are defined by:
\begin{eqnarray}
{\mathcal T}^{a\Delta}_m&:=&x^{(1-\Delta)(m+1)}\bar{\Delta}\,x^{\Delta(m+1)}\label{qsop1}\\{\mathbb T}^{a\Delta}_m&:=&\theta x^{(1-\Delta)(m+1)}\bar{\Delta}\,x^{\Delta(m+1)}\label{qsop2},
\end{eqnarray}
where ${\mathcal D}_{q^{a}}$ is the $q$-derivative \eqref{qder}.
From the  $q$-number\eqref{qnumber},
the relations \eqref{qsop1} and \eqref{qsop2} are reduced as:
\begin{eqnarray*}
	{\mathcal T}^{a\Delta}_m&=&[\Delta(m+1)]_{q^{a}}\,x^{m}\\{\mathbb T}^{a\Delta}_m&=&\theta[\Delta(m+1)]_{q^{a}}\,x^{m}.
\end{eqnarray*}
Thus, the $q$-conformal super operators (\ref{qsop1}) and \eqref{qsop2} satisfies the product relation: \begin{align*}
\mathcal{T}^{a\Delta}_{m}.\mathbb{T}^{b\Delta}_n&={\big(q^{a+b}-q^{-a-b}\big)q^{a(n(1-\Delta)+1)}\over \big(q^{a}-q^{-a}\big)\big(q^{b}-q^{-b}\big)q^{b\Delta\,m}}\,\mathbb{T}^{(a+b)\Delta}_{m+n}- {q^{-b\Delta(n+1)}\over q^{a(n(\Delta-1)-1)}\big(q^{b}-q^{-b}\big)}\, \mathbb{T}^{a\Delta}_{m+n}\nonumber\\& - {q^{-a(\Delta(m+1)+n+1)}\over q^{b\Delta\,m}\big(q^{a}-q^{-a}\big)}\,\mathbb{T}^{b\Delta}_{m+n}+ \mathbb{F}^{\Delta(a,b)}_{q}(m,n),
\end{align*}
with \begin{eqnarray*}
	\mathbb{F}^{\Delta(a,b)}_{q}(m,n)=q^{-(a+b)\Delta(m+n+1)}\,[-\Delta\,m]_{q^{b}}\,[n(1-\Delta)+1]_{q^{a}},
\end{eqnarray*}	
and the commutation relation
\begin{align}\label{qscrto}
\Big[\mathcal{T}^{a\Delta}_{m}, \mathbb{T}^{b\Delta}_n\Big]&={\big(q^{a+b}-q^{-a-b}\big)\over \big(q^{a}-q^{-a}\big)\big(q^{b}-q^{-b}\big)}\bigg(\frac{q^{a(n(1-\Delta)+1)}}{q^{b\Delta\,m}}-\frac{q^{bm(1-\Delta)}}{q^{a\Delta\,n}}\bigg)\mathbb{T}^{(a+b)\Delta}_{m+n}\nonumber\\ &-\frac{q^{-b\Delta(n+1)}}{q^{a\Delta\,n}}\frac{\big(q^{a(n+1)}-q^{-b(m+1)}\big)}{\big(q^{b}-q^{-b}\big)}\mathbb{T}^{a\Delta}_{m+n}+\mathbb{H}^{\Delta(a,b)}_{q}(m,n) \nonumber\\ &-\frac{q^{-a\Delta(m+1)}}{q^{b\Delta\,m}}\frac{\big(q^{b(m+1)}-q^{-a(n+1)}\big)}{\big(q^{a}-q^{-a}\big)}\mathbb{T}^{b\Delta}_{m+n},
\end{align}
where 
\begin{align*}
\mathbb{H}^{\Delta(a,b)}_{q}(m,n)&=q^{-(a+b)\Delta(m+n+1)}\bigg([-\Delta\,m]_{q^{b}}\,[n(1-\Delta)+1]_{q^{a}}\nonumber\\&-[-\Delta\,n]_{q^{a}}\,[m(1-\Delta)+1]_{q^{b}}\bigg).
\end{align*}
Putting $a=b=1,$ in the relation \eqref{qscrto},  the commutation relation:
\begin{align*}
\Big[\mathcal{T}^{\Delta}_{m}, \mathbb{T}^{\Delta}_n\Big]&={q^{-\Delta(n+m)}\big(q^{n+1}-q^{m+1}\big)\over \big(q-q^{-1}\big)}[2]_{q}\mathbb{T}^{2\Delta}_{m+n}+\mathbb{H}^{\Delta}_{q}(m,n)\nonumber\\
&-\bigg(\frac{q^{-\Delta(n+1)}\big(q^{n+1}-q^{-m-1}\big)}{q^{\Delta\,n}\big(q-q^{-1}\big)}+ \frac{q^{-\Delta(m+1)}\big(q^{m+1}-q^{-n-1}\big)}{q^{\Delta\,m}\big(q-q^{-1}\big)}\bigg)\mathbb{T}^{\Delta}_{m+n},
\end{align*}
where
\begin{align*}
\mathbb{H}^{\Delta(1,1)}_{q}(m,n)&=q^{-2\Delta(m+n+1)}\bigg([-\Delta\,m]_{q}[n(1-\Delta)+1]_{q}\\&-[-\Delta\,n]_{q}[m(1-\Delta)+1]_{q}\bigg).
\end{align*}
Besides,
the super $n$-bracket is defined by:
\begin{align*}
\big[\mathcal{T}^{a\Delta}_{m_1},\mathcal{T}^{a\Delta}_{m_2},\cdots, \mathbb{T}^{a\Delta}_{m_n}\big]&:=\sum_{j=0}^{n-1}(-1)^{n-1+j}\epsilon^{i_1\ldots i_{n-1}}_{12\cdots n-1}\mathcal{T}^{a\Delta}_{m_{i_1}}\cdots \mathcal{T}^{a\Delta}_{m_{i_j}}\nonumber\\&\times
\mathbb{T}^{a\Delta}_{m_{n}}\mathcal{T}^{a\Delta}_{m_{i_{j+1}}}\cdots \mathcal{T}^{a\Delta}_{m_{i_{n-1}}}.
\end{align*}
Setting $a=b$ in the relation (\ref{qscrto}),  we obtain:
\begin{align*}
\Big[\mathcal{T}^{a\Delta}_{m}, \mathbb{T}^{a\Delta}_n\Big]&={q^{-a\Delta(n+m)}\big(q^{a(n+1)}-q^{a(m+1)}\big)\over q^{a}-q^{-a}}[2]_{q^a}\mathbb{T}^{2a\Delta}_{m+n}\nonumber\\ &-\frac{1}{\big(q-q^{-1}\big)}\bigg(\frac{q^{-a\Delta(n+1)}}{q^{a\Delta\,n}}\big(q^{a(n+1)}-q^{-a(m+1)}\big)\nonumber\\&- \frac{q^{-a\Delta(m+1)}}{q^{a\Delta\,m}}\big(q^{a(m+1)}-q^{-a(n+1)}\big)\bigg)\mathbb{T}^{a\Delta}_{m+n}+\mathbb{H}^{\Delta(a)}_{q}(m,n),
\end{align*}
where \begin{align*}
\mathbb{H}^{\Delta(a)}_{q}(m,n)&=q^{-2a\Delta(m+n+1)}\bigg([-\Delta\,m]_{q^{a}}[n(1-\Delta)+1]_{q^{a}}\\&-[-\Delta\,n]_{q^{a}}[m(1-\Delta)+1]_{q^{a}}\bigg).
\end{align*}
The $q$-conformal super Witt $n$-algebra is determined  by :
\begin{align*}
\Big[{\mathcal T}^{a\Delta}_{m_1},\cdots,{\mathbb T}^{a\Delta}_{m_n}
\Big]&={(-1)^{n+1}\over \big(q^{a}-q^{-a}\big)^{n-1}}\Big( \mathbb{V}^n_{a\Delta}[n]_{q^{a}}{\mathbb T}^{n\,a\Delta}_{\bar{m}}\nonumber\\ &- [n-1]_{q^{a}}\big(\mathbb{D}^n_{a\Delta}+ \mathbb{N}^n_{a\Delta}\big){\mathbb T}^{a(n-1)\Delta}_{\bar{m}}\Big)\nonumber\\&+\mathbb{H}^{\Delta}_{q}(m_1,\ldots,m_n),
\end{align*}
where 
\begin{align*}
\mathbb{V}^n_{a\Delta}&= q^{a(n-1)(1-\Delta )\bar{m}}\Big(\big(q^{a}-q^{-a}\big)^{n\choose 2}\prod_{1\leq j < k \leq n}\Big([m_j+1]_{q^{a}}-[m_k+1]_{q^{a}}\Big)\nonumber\\&+\prod_{1\leq j < k \leq n}\Big(q^{-a(m_j+1)}-q^{-a(m_k+1)}\Big)\Big),
\end{align*}
\begin{align*}
\mathbb{D}^n_{a\Delta}
&=\Big(\big(q^{a}-q^{-a}\big)^{n\choose 2}\prod_{1\leq j < k \leq n}\frac{q^{-a\Delta(m_k+1)}}{q^{a\Delta\,m_k}}\Big([m_k+1]_{q^{a}}-[m_j+1]_{q^{a}}\\&+q^{-a(m_k+1)}-q^{a(m_j+1)}\Big),
\end{align*}
\begin{align*}
\mathbb{D}^n_{a\Delta}&=(-1)^{n+1}
\Big(\big(q^{a}-q^{-a}\big)^{n\choose 2}\prod_{1\leq j < k \leq n}\frac{q^{-a\Delta(m_j+1)}}{q^{a\Delta\,m_j}}\Big([m_j+1]_{q^{a}}-[m_k+1]_{q^{a}}\nonumber\\&+q^{-a(m_j+1)}-q^{a(m_k+1)}\Big),
\end{align*}
and
\begin{align*}
\mathbb{H}^{\Delta(a)}_{q}(m_1,\ldots,m_n)&=q^{-an\Delta(\bar{m}+1)}\prod_{1\leq j < k \leq n}\bigg([-\Delta\,m_j]_{q^{a}}[m_k(1-\Delta)+1]_{q^{a}}\nonumber\\&-[-\Delta\,m_k]_{q^{a}}[m_j(1-\Delta)+1]_{q^{a}}\bigg).
\end{align*}
\subsubsection{Conformal super Witt $n$-algebra corresponding to the Jagannathan-Srinivasa algebra } 
The conformal super Witt $n$-algebra from the Jagannathan-Srinivasa algebra\cite{JS} is derived as: We consider 
for $(\Delta \neq 0,1),$ the $(p,q)$-conformal super operators defined by:
\begin{eqnarray}
\,{\mathcal T}^{a\Delta}_m&:=&x^{(1-\Delta)(m+1)}\bar{\Delta}\,x^{\Delta(m+1)}\label{pqsop1}\\
\,{\mathbb T}^{a\Delta}_m&:=&\theta\,x^{(1-\Delta)(m+1)}\bar{\Delta}\,x^{\Delta(m+1)}\label{pqsop2},
\end{eqnarray}
where ${\mathcal D}_{p^{a},q^{a}}$ is the $(p,q)$-derivative \eqref{pqder}.
By using the  $(p,q)$-number \eqref{pqnumber},  the relation \eqref{pqsop1} and \eqref{pqsop2} take the form:
\begin{eqnarray*}
	{\mathcal T}^{a\Delta}_m&=&[\Delta(m+1)]_{p^{a},q^{a}}\,x^{m}\\
	{\mathbb T}^{a\Delta}_m&=&\theta[\Delta(m+1)]_{p^{a},q^{a}}\,x^{m}.
\end{eqnarray*}
Then, the $(p,q)$-conformal super operators (\ref{pqsop1}) and \eqref{pqsop2} satisfy the product relation: \begin{align*}
\mathcal{T}^{a\Delta}_{m}.\mathbb{T}^{b\Delta}_n&={\big(p^{a+b}-q^{a+b}\big)p^{a(n(1-\Delta)+1)}\over \big(p^{a}-q^{a}\big)\big(p^{b}-q^{b}\big)p^{b\Delta\,m}}\,\mathbb{T}^{(a+b)\Delta}_{m+n}- {q^{b\Delta(n+1)}\over p^{a(n(\Delta-1)-1)}\big(p^{b}-q^{b}\big)}\, \mathbb{T}^{a\Delta}_{m+n}\nonumber\\& - {q^{a(\Delta(m+1)+n+1)}\over p^{b\Delta\,m}\big(p^{a}-q^{a}\big)}\,\mathbb{T}^{b\Delta}_{m+n}+ \mathbb{F}^{\Delta(a,b)}_{p,q}(m,n),
\end{align*}
with \begin{eqnarray*}
	\mathbb{F}^{\Delta(a,b)}_{p,q}(m,n)=q^{(a+b)\Delta(m+n+1)}\,[-\Delta\,m]_{p^{b},q^{b}}\,[n(1-\Delta)+1]_{p^{a},q^{a}},
\end{eqnarray*}	
and the commutation relation
\begin{align}\label{pqscrto}
\Big[\mathcal{T}^{a\Delta}_{m}, \mathbb{T}^{b\Delta}_n\Big]&={\big(p^{a+b}-q^{a+b}\big)\over \big(p^{a}-q^{a}\big)\big(p^{b}-q^{b}\big)}\bigg(\frac{p^{a(n(1-\Delta)+1)}}{p^{b\Delta\,m}}-\frac{p^{b(m(1-\Delta)+1)}}{p^{a\Delta\,n}}\bigg)\mathbb{T}^{(a+b)\Delta}_{m+n}\nonumber\\ &-\frac{q^{b\Delta(n+1)}}{p^{a\Delta\,n}}\frac{\big(p^{a(n+1)}-q^{b(m+1)}\big)}{\big(p^{b}-q^{b}\big)}\mathbb{T}^{a\Delta}_{m+n}+\mathbb{H}^{\Delta(a,b)}_{p,q}(m,n) \nonumber\\ &-\frac{q^{a\Delta(m+1)}}{p^{b\Delta\,m}}\frac{\big(p^{b(m+1)}-q^{a(n+1)}\big)}{\big(p^{a}-q^{a}\big)}\mathbb{T}^{b\Delta}_{m+n},
\end{align}
with 
\begin{eqnarray*}
	\mathbb{H}^{\Delta(a,b)}_{p,q}(m,n)&=&q^{(a+b)\Delta(m+n+1)}\bigg([-\Delta\,m]_{p^{b},q^{b}}\,[n(1-\Delta)+1]_{p^{a},q^{a}}\nonumber\\&-&[-\Delta\,n]_{p^{a},q^{a}}\,[m(1-\Delta)+1]_{p^{b},q^{b}}\bigg).
\end{eqnarray*}
Putting $a=b=1,$ in the relation \eqref{pqscrto}, we obtain the commutation relation:
\begin{align*}
\Big[\mathcal{T}^{\Delta}_{m}, \mathbb{T}^{\Delta}_n\Big]&={p^{-\Delta(n+m)}\big(p^{n+1}-p^{m+1}\big)\over \big(p-q\big)}[2]_{p,q}\mathbb{T}^{2\Delta}_{m+n}+\mathbb{H}^{\Delta}_{p,q}(m,n)\nonumber\\
&-\bigg(\frac{q^{\Delta(n+1)}\big(p^{n+1}-q^{m+1}\big)}{p^{\Delta\,n}\big(p-q\big)}+ \frac{q^{\Delta(m+1)}\big(p^{m+1}-q^{n+1}\big)}{p^{\Delta\,m}\big(p-q\big)}\bigg)\mathbb{T}^{\Delta}_{m+n},
\end{align*}
where
\begin{align*}
\mathbb{H}^{\Delta(1,1)}_{p,q}(m,n)&=q^{2\Delta(m+n+1)}\bigg([-\Delta\,m]_{p,q}\,[n(1-\Delta)+1]_{p,q}\nonumber\\&-[-\Delta\,n]_{p,q}\,[m(1-\Delta)+1]_{p,q}\bigg).
\end{align*}
Furthermore, the super $n$-bracket is defined by:
\begin{align*}
\big[\mathcal{T}^{a\Delta}_{m_1},\mathcal{T}^{a\Delta}_{m_2},\cdots, \mathbb{T}^{a\Delta}_{m_n}\big]&:=\sum_{j=0}^{n-1}(-1)^{n-1+j}\epsilon^{i_1\ldots i_{n-1}}_{12\cdots n-1}\mathcal{T}^{a\Delta}_{m_{i_1}}\cdots \mathcal{T}^{a\Delta}_{m_{i_j}}\nonumber\\&\times
\mathbb{T}^{a\Delta}_{m_{n}}\mathcal{T}^{a\Delta}_{m_{i_{j+1}}}\cdots \mathcal{T}^{a\Delta}_{m_{i_{n-1}}}.
\end{align*}
Putting $a=b$ in the relation (\ref{pqscrto}),  we obtain:
\begin{align*}
\Big[\mathcal{T}^{a\Delta}_{m}, \mathbb{T}^{a\Delta}_n\Big]&={p^{-a\Delta(n+m)}\big(p^{a(n+1)}-p^{a(m+1)}\big)\over p^{a}-q^{a}}[2]_{p^a,q^a}\mathbb{T}^{2a\Delta}_{m+n}\nonumber\\ &-\frac{1}{\big(p-q\big)}\bigg(\frac{q^{a\Delta(n+1)}}{p^{a\Delta\,n}}\big(p^{a(n+1)}-q^{a(m+1)}\big)\nonumber\\&- \frac{q^{a\Delta(m+1)}}{p^{a\Delta\,m}}\big(p^{a(m+1)}-q^{a(n+1)}\big)\bigg)\mathbb{T}^{a\Delta}_{m+n}+\mathbb{H}^{\Delta(a)}_{p,q}(m,n),
\end{align*}
where \begin{align*}
\mathbb{H}^{\Delta(a)}_{p,q}(m,n)&=q^{2a\Delta(m+n+1)}\bigg([-\Delta\,m]_{p^{a},q^{a}}[n(1-\Delta)+1]_{p^{a},q^{a}}\nonumber\\&-[-\Delta\,n]_{p^{a},q^{a}}[m(1-\Delta)+1]_{p^{a},q^{a}}\bigg)
\end{align*}
The $(p,q)$-conformal super Witt $n$-algebra is given  by :
\begin{align*}
\Big[{\mathcal T}^{a\Delta}_{m_1},\cdots,{\mathbb T}^{a\Delta}_{m_n}
\Big]&={(-1)^{n+1}\over \big(p^{a}-q^{a}\big)^{n-1}}\Big( \mathbb {V}^n_{a\Delta}[n]_{p^{a},q^{a}}{\mathbb T}^{n\,a\Delta}_{\bar{m}}\nonumber\\ &- [n-1]_{p^{a},q^{a}}\big(\mathbb{D}^n_{a\Delta}+ \mathbb{N}^n_{a\Delta}\big){\mathbb T}^{a(n-1)\Delta}_{\bar{m}}\Big)\nonumber\\&+\mathbb{H}^{\Delta}_{p,q}(m_1,\ldots,m_n),
\end{align*}
where 
\begin{align*}
\mathbb{V}^n_{a\Delta}&= p^{a(n-1)(1-\Delta )\bar{m}}\Big(\big(p^{a}-q^{a}\big)^{n\choose 2}\prod_{1\leq j < k \leq n}\Big([m_j+1]_{p^{a},q^{a}}-[m_k+1]_{p^{a},q^{a}}\Big)\nonumber\\&+\prod_{1\leq j < k \leq n}\Big(q^{a(m_j+1)}-q^{a(m_k+1)}\Big)\Big),
\end{align*}
\begin{eqnarray*}
	\mathbb{D}^n_{a\Delta}
	&=&\Big(\big(p^{a}-q^{a}\big)^{n\choose 2}\prod_{1\leq j < k \leq n}\frac{q^{a\Delta(m_k+1)}}{p^{a\Delta\,m_k}}\Big([m_k+1]_{p^{a},q^{a}}-[m_j+1]_{p^{a},q^{a}}\\&+&q^{a(m_k+1)}-p^{a(m_j+1)}\Big),
\end{eqnarray*}
\begin{eqnarray*}
	\mathbb{N}^n_{a\Delta}&=&(-1)^{n+1}
	\Big(\big(p^{a}-q^{a}\big)^{n\choose 2}\prod_{1\leq j < k \leq n}\frac{q^{a\Delta(m_j+1)}}{p^{a\Delta\,m_j}}\Big([m_j+1]_{p^{a},q^{a}}-[m_k+1]_{p^{a},q^{a}}\nonumber\\&+&q^{a(m_j+1)}-p^{a(m_k+1)}\Big),
\end{eqnarray*}
and
\begin{align*}
\mathbb{H}^{\Delta(a)}_{p,q}(m_1,\ldots,m_n)&=q^{an\Delta(\bar{m}+1)}\prod_{1\leq j < k \leq n}\bigg([-\Delta\,m_j]_{p^{a},q^{a}}[m_k(1-\Delta)+1]_{p^{a},q^{a}}\nonumber\\&-[-\Delta\,m_k]_{p^{a},q^{a}}[m_j(1-\Delta)+1]_{p^{a},q^{a}}\bigg).
\end{align*}
\subsubsection{Conformal super Witt $n$-algebra associated to the Chakrabarty and Jagannathan algebra }
The conformal super Witt $n$-algebra from the Chakrabarty and Jagannathan algebra \cite{CJ}  is derived as: We consider 
for $(\Delta \neq 0,1),$ the $(p^{-1},q)$-conformal super operators defined by:
\begin{align}
\,{\mathcal T}^{a\Delta}_m&:=x^{(1-\Delta)(m+1)}\bar{\Delta}\,x^{\Delta(m+1)}\label{cjsop1}\\
\,{\mathbb T}^{a\Delta}_m&:=\theta\,x^{(1-\Delta)(m+1)}\bar{\Delta}\,x^{\Delta(m+1)}\label{cjsop2},
\end{align}
where ${\mathcal D}_{p^{-a},q^{a}}$ is the $(p^{-1},q)$-derivative \eqref{cjder}.
By using the  $(p^{-1},q)$-number \eqref{cjnumber},  the relation \eqref{cjsop1} and \eqref{cjsop2} take the form:
\begin{align*}
{\mathcal T}^{a\Delta}_m&=[\Delta(m+1)]_{p^{-a},q^{a}}\,x^{m}\\
{\mathbb T}^{a\Delta}_m&=\theta[\Delta(m+1)]_{p^{-a},q^{a}}\,x^{m}.
\end{align*}
Then, the $(p^{-1},q)$-conformal super operators (\ref{cjsop1}) and \eqref{cjsop2} satisfy the product relation: \begin{align*}
\mathcal{T}^{a\Delta}_{m}.\mathbb{T}^{b\Delta}_n&={\big(p^{-a-b}-q^{a+b}\big)p^{-a(n(1-\Delta)+1)}\over \big(p^{-a}-q^{a}\big)\big(p^{-b}-q^{b}\big)p^{-b\Delta\,m}}\,\mathbb{T}^{(a+b)\Delta}_{m+n}+ \mathbb{F}^{\Delta(a,b)}_{p^{-1},q}(m,n)\nonumber\\& - {q^{a(\Delta(m+1)+n+1)}\over p^{-b\Delta\,m}\big(p^{-a}-q^{a}\big)}\,\mathbb{T}^{b\Delta}_{m+n}- {q^{b\Delta(n+1)}\over p^{-a(n(\Delta-1)-1)}\big(p^{-b}-q^{b}\big)}\, \mathbb{T}^{a\Delta}_{m+n},
\end{align*}
with \begin{eqnarray*}
	\mathbb{F}^{\Delta(a,b)}_{p^{-1},q}(m,n)=q^{(a+b)\Delta(m+n+1)}\,[-\Delta\,m]_{p^{-b},q^{b}}\,[n(1-\Delta)+1]_{p^{-a},q^{a}},
\end{eqnarray*}	
and the commutation relation
\begin{align}\label{cjscrto}
\Big[\mathcal{T}^{a\Delta}_{m}, \mathbb{T}^{b\Delta}_n\Big]&={\big(p^{-a-b}-q^{a+b}\big)\over \big(p^{-a}-q^{a}\big)\big(p^{-b}-q^{b}\big)}\bigg(\frac{p^{-a(n(1-\Delta)+1)}}{p^{-b\Delta\,m}}-\frac{p^{-b(m(1-\Delta)+1)}}{p^{-a\Delta\,n}}\bigg)\mathbb{T}^{(a+b)\Delta}_{m+n}\nonumber\\ &-\frac{q^{b\Delta(n+1)}}{p^{-a\Delta\,n}}\frac{\big(p^{-a(n+1)}-q^{b(m+1)}\big)}{\big(p^{-b}-q^{b}\big)}\mathbb{T}^{a\Delta}_{m+n}+\mathbb{H}^{\Delta(a,b)}_{p^{-1},q}(m,n) \nonumber\\ &-\frac{q^{a\Delta(m+1)}}{p^{-b\Delta\,m}}\frac{\big(p^{-b(m+1)}-q^{a(n+1)}\big)}{\big(p^{-a}-q^{a}\big)}\mathbb{T}^{b\Delta}_{m+n},
\end{align}
with 
\begin{align*}
\mathbb{H}^{\Delta(a,b)}_{p^{-1},q}(m,n)&=q^{(a+b)\Delta(m+n+1)}\bigg([-\Delta\,m]_{p^{-b},q^{b}}\,[n(1-\Delta)+1]_{p^{-a},q^{a}}\nonumber\\&-[-\Delta\,n]_{p^{-a},q^{a}}\,[m(1-\Delta)+1]_{p^{-b},q^{b}}\bigg).
\end{align*}
Putting $a=b=1,$ in the relation \eqref{cjscrto}, we obtain the commutation relation:
\begin{align*}
\Big[\mathcal{T}^{\Delta}_{m}, \mathbb{T}^{\Delta}_n\Big]&={p^{\Delta(n+m)}\big(p^{-n-1}-p^{-m-1}\big)\over \big(p^{-1}-q\big)}[2]_{p^{-1},q}\mathbb{T}^{2\Delta}_{m+n}+\mathbb{H}^{\Delta}_{p^{-1},q}(m,n)\nonumber\\
&-\bigg(\frac{q^{\Delta(n+1)}\big(p^{-n-1}-q^{m+1}\big)}{p^{-\Delta\,n}\big(p^{-1}-q\big)}+ \frac{q^{\Delta(m+1)}\big(p^{-m-1}-q^{n+1}\big)}{p^{-\Delta\,m}\big(p^{-1}-q\big)}\bigg)\mathbb{T}^{\Delta}_{m+n},
\end{align*}
where
\begin{align*}
\mathbb{H}^{\Delta(1,1)}_{p^{-1},q}(m,n)&=q^{2\Delta(m+n+1)}\bigg([-\Delta\,m]_{p^{-1},q}\,[n(1-\Delta)+1]_{p^{-1},q}\nonumber\\&-[-\Delta\,n]_{p^{-1},q}\,[m(1-\Delta)+1]_{p^{-1},q}\bigg).
\end{align*}
Furthermore, the super $n$-bracket is defined by:
\begin{align*}
\big[\mathcal{T}^{a\Delta}_{m_1},\mathcal{T}^{a\Delta}_{m_2},\cdots, \mathbb{T}^{a\Delta}_{m_n}\big]&:=\sum_{j=0}^{n-1}(-1)^{n-1+j}\epsilon^{i_1\ldots i_{n-1}}_{12\cdots n-1}\mathcal{T}^{a\Delta}_{m_{i_1}}\cdots \mathcal{T}^{a\Delta}_{m_{i_j}}\nonumber\\&\times
\mathbb{T}^{a\Delta}_{m_{n}}\mathcal{T}^{a\Delta}_{m_{i_{j+1}}}\cdots \mathcal{T}^{a\Delta}_{m_{i_{n-1}}}.
\end{align*}
Putting $a=b$ in the relation (\ref{cjscrto}),  we have:
\begin{align*}
\Big[\mathcal{T}^{a\Delta}_{m}, \mathbb{T}^{a\Delta}_n\Big]&={p^{a\Delta(n+m)}\big(p^{-a(n+1)}-p^{-a(m+1)}\big)\over p^{-a}-q^{a}}[2]_{p^{-a},q^a}\mathbb{T}^{2a\Delta}_{m+n}\nonumber\\ &-\frac{1}{\big(p^{-1}-q\big)}\bigg(\frac{q^{a\Delta(n+1)}}{p^{-a\Delta\,n}}\big(p^{-a(n+1)}-q^{a(m+1)}\big)\nonumber\\&- \frac{q^{a\Delta(m+1)}}{p^{-a\Delta\,m}}\big(p^{-a(m+1)}-q^{a(n+1)}\big)\bigg)\mathbb{T}^{a\Delta}_{m+n}+\mathbb{H}^{\Delta(a)}_{p^{-1},q}(m,n),
\end{align*}
where \begin{align*}
\mathbb{H}^{\Delta(a)}_{p^{-1},q}(m,n)&=q^{2a\Delta(m+n+1)}\bigg([-\Delta\,m]_{p^{-a},q^{a}}[n(1-\Delta)+1]_{p^{-a},q^{a}}\nonumber\\&-[-\Delta\,n]_{p^{-a},q^{a}}[m(1-\Delta)+1]_{p^{-a},q^{a}}\bigg)
\end{align*}
The $(p^{-1},q)$-conformal super Witt $n$-algebra is given  by :
\begin{align*}
\Big[{\mathcal T}^{a\Delta}_{m_1},\cdots,{\mathbb T}^{a\Delta}_{m_n}
\Big]&={(-1)^{n+1}\over \big(p^{-a}-q^{a}\big)^{n-1}}\Big( \mathbb {V}^n_{a\Delta}[n]_{p^{-a},q^{a}}{\mathbb T}^{n\,a\Delta}_{\bar{m}}\nonumber\\ &- [n-1]_{p^{-a},q^{a}}\big(\mathbb{D}^n_{a\Delta}+ \mathbb{N}^n_{a\Delta}\big){\mathbb T}^{a(n-1)\Delta}_{\bar{m}}\Big)\nonumber\\&+\mathbb{H}^{\Delta}_{p^{-1},q}(m_1,\ldots,m_n),
\end{align*}
where 
\begin{align*}
\mathbb{V}^n_{a\Delta}&= p^{-a(n-1)(1-\Delta )\bar{m}}\Big(\big(p^{-a}-q^{a}\big)^{n\choose 2}\prod_{1\leq j < k \leq n}\Big([m_j+1]_{p^{-a},q^{a}}\nonumber\\&-[m_k+1]_{p^{-a},q^{a}}\Big)+\prod_{1\leq j < k \leq n}\Big(q^{a(m_j+1)}-q^{a(m_k+1)}\Big)\Big),
\end{align*}
\begin{align*}
\mathbb{D}^n_{a\Delta}
&=\Big(\big(p^{-a}-q^{a}\big)^{n\choose 2}\prod_{1\leq j < k \leq n}\frac{q^{a\Delta(m_k+1)}}{p^{-a\Delta\,m_k}}\Big([m_k+1]_{p^{-a},q^{a}}-[m_j+1]_{p^{-a},q^{a}}\\&+q^{a(m_k+1)}-p^{-a(m_j+1)}\Big),
\end{align*}
\begin{align*}
\mathbb{N}^n_{a\Delta}&=(-1)^{n+1}
\Big(\big(p^{-a}-q^{a}\big)^{n\choose 2}\prod_{1\leq j < k \leq n}\frac{q^{a\Delta(m_j+1)}}{p^{-a\Delta\,m_j}}\Big([m_j+1]_{p^{-a},q^{a}}\nonumber\\&-[m_k+1]_{p^{-a},q^{a}}+q^{a(m_j+1)}-p^{-a(m_k+1)}\Big),
\end{align*}
and
\begin{align*}
\mathbb{H}^{\Delta(a)}_{p^{-1},q}(m_1,\ldots,m_n)&=\prod_{1\leq j < k \leq n}\bigg([-\Delta\,m_j]_{p^{-a},q^{a}}[m_k(1-\Delta)+1]_{p^{-a},q^{a}}\nonumber\\&-[-\Delta\,m_k]_{p^{-a},q^{a}}[m_j(1-\Delta)+1]_{p^{-a},q^{a}}\bigg)q^{an\Delta(\bar{m}+1)}.
\end{align*}
\subsubsection{Conformal super Witt $n$-algebra induced by the Hounkonnou-Ngompe generalization of $q$-Quesne algebra  }
The conformal super Witt $n$-algebra from the Hounkonnou-Ngompe generalization of $q$-Quesne algebra \cite{HNN} is derived as: We consider 
for $(\Delta \neq 0,1),$ the generalized $q$-Quesne conformal super operators defined by:
\begin{align}
\,{\mathcal T}^{a\Delta}_m&:=x^{(1-\Delta)(m+1)}\bar{\Delta}\,x^{\Delta(m+1)}\label{hnsop1}\\
\,{\mathbb T}^{a\Delta}_m&:=\theta\,x^{(1-\Delta)(m+1)}\bar{\Delta}\,x^{\Delta(m+1)}\label{hnsop2},
\end{align}
where ${\mathcal D}^Q_{p^{a},q^{a}}$ is the generalized $q$-Quesne derivative \eqref{hnder}.
From the  generalized $q$-Quesne number \eqref{hnnumber},  the relations \eqref{hnsop1} and \eqref{hnsop2} take the form:
\begin{align*}
{\mathcal T}^{a\Delta}_m&=[\Delta(m+1)]^Q_{p^{a},q^{a}}\,x^{m}\\
{\mathbb T}^{a\Delta}_m&=\theta[\Delta(m+1)]^Q_{p^{a},q^{a}}\,x^{m}.
\end{align*}
Then, the generalized $q$-Quesne conformal super operators (\ref{hnsop1}) and \eqref{hnsop2} satisfy the product relation: \begin{align*}
\mathcal{T}^{a\Delta}_{m}.\mathbb{T}^{b\Delta}_n&={\big(p^{a+b}-q^{-a-b}\big)p^{a(n(1-\Delta)+1)}\over \big(-p^{-a}+q^{a}\big)\big(-p^{-b}+q^{b}\big)p^{b\Delta\,m}}\,\mathbb{T}^{(a+b)\Delta}_{m+n}\\&- {q^{-b\Delta(n+1)}\over p^{a(n(\Delta-1)-1)}\big(-p^{-b}+q^{b}\big)}\, \mathbb{T}^{a\Delta}_{m+n}\nonumber\\& - {q^{-a(\Delta(m+1)+n+1)}\over p^{b\Delta\,m}\big(-p^{-a}+q^{a}\big)}\,\mathbb{T}^{b\Delta}_{m+n}+ \mathbb{F}^{\Delta(a,b)}_{p,q}(m,n),
\end{align*}
with \begin{eqnarray*}
	\mathbb{F}^{\Delta(a,b)}_{p,q}(m,n)=q^{-(a+b)\Delta(m+n+1)}\,[-\Delta\,m]^Q_{p^{b},q^{b}}\,[n(1-\Delta)+1]^Q_{p^{a},q^{a}},
\end{eqnarray*}	
and the commutation relation
\begin{align}\label{hnscrto}
\Big[\mathcal{T}^{a\Delta}_{m}, \mathbb{T}^{b\Delta}_n\Big]&={\big(p^{a+b}-q^{-a-b}\big)\over \big(-p^{-a}+q^{a}\big)\big(-p^{-b}+q^{b}\big)}\bigg(\frac{p^{a(n(1-\Delta)+1)}}{p^{b\Delta\,m}}-\frac{p^{b(m(1-\Delta)+1)}}{p^{a\Delta\,n}}\bigg)\mathbb{T}^{(a+b)\Delta}_{m+n}\nonumber\\ &-\frac{q^{-b\Delta(n+1)}}{p^{a\Delta\,n}}\frac{\big(p^{a(n+1)}-q^{-b(m+1)}\big)}{\big(-p^{-b}+q^{b}\big)}\mathbb{T}^{a\Delta}_{m+n}+\mathbb{H}^{\Delta(a,b)}_{p,q}(m,n) \nonumber\\ &-\frac{q^{-a\Delta(m+1)}}{p^{b\Delta\,m}}\frac{\big(p^{b(m+1)}-q^{-a(n+1)}\big)}{\big(-p^{-a}+q^{a}\big)}\mathbb{T}^{b\Delta}_{m+n},
\end{align}
with 
\begin{align*}
\mathbb{H}^{\Delta(a,b)}_{p,q}(m,n)&=q^{-(a+b)\Delta(m+n+1)}\bigg([-\Delta\,m]^Q_{p^{b},q^{b}}\,[n(1-\Delta)+1]^Q_{p^{a},q^{a}}\nonumber\\&-[-\Delta\,n]^Q_{p^{a},q^{a}}\,[m(1-\Delta)+1]^Q_{p^{b},q^{b}}\bigg).
\end{align*}
Putting $a=b=1,$ in the relation \eqref{hnscrto}, we obtain the commutation relation:
\begin{align*}
\Big[\mathcal{T}^{\Delta}_{m}, \mathbb{T}^{\Delta}_n\Big]&={p^{-\Delta(n+m)}\big(p^{n+1}-p^{m+1}\big)\over \big(-p^{-1}+q\big)}[2]^Q_{p,q}\mathbb{T}^{2\Delta}_{m+n}+\mathbb{H}^{\Delta}_{p,q}(m,n)\nonumber\\
&-\bigg(\frac{q^{-\Delta(n+1)}\big(p^{n+1}-q^{-m-1}\big)}{p^{\Delta\,n}\big(-p^{-1}+q\big)}+ \frac{q^{-\Delta(m+1)}\big(p^{m+1}-q^{-n-1}\big)}{p^{\Delta\,m}\big(-p^{-1}+q\big)}\bigg)\mathbb{T}^{\Delta}_{m+n},
\end{align*}
where
\begin{align*}
\mathbb{H}^{\Delta(1,1)}_{p,q}(m,n)&=q^{-2\Delta(m+n+1)}\bigg([-\Delta\,m]^Q_{p,q}\,[n(1-\Delta)+1]^Q_{p,q}\nonumber\\&-[-\Delta\,n]^Q_{p,q}\,[m(1-\Delta)+1]^Q_{p,q}\bigg).
\end{align*}
Furthermore, the super $n$-bracket is defined by:
\begin{align*}
\big[\mathcal{T}^{a\Delta}_{m_1},\mathcal{T}^{a\Delta}_{m_2},\cdots, \mathbb{T}^{a\Delta}_{m_n}\big]&:=\sum_{j=0}^{n-1}(-1)^{n-1+j}\epsilon^{i_1\ldots i_{n-1}}_{12\cdots n-1}\mathcal{T}^{a\Delta}_{m_{i_1}}\cdots \mathcal{T}^{a\Delta}_{m_{i_j}}\nonumber\\&\times
\mathbb{T}^{a\Delta}_{m_{n}}\mathcal{T}^{a\Delta}_{m_{i_{j+1}}}\cdots \mathcal{T}^{a\Delta}_{m_{i_{n-1}}}.
\end{align*}
Putting $a=b$ in the relation (\ref{hnscrto}),  we obtain:
\begin{align*}
\Big[\mathcal{T}^{a\Delta}_{m}, \mathbb{T}^{a\Delta}_n\Big]&={p^{-a\Delta(n+m)}\big(p^{a(n+1)}-p^{-a(m+1)}\big)\over -p^{-a}+q^{a}}[2]^Q_{p^a,q^a}\mathbb{T}^{2a\Delta}_{m+n}\nonumber\\ &-\frac{1}{\big(-p^{-1}+q\big)}\bigg(\frac{q^{-a\Delta(n+1)}}{p^{a\Delta\,n}}\big(p^{a(n+1)}-q^{-a(m+1)}\big)\nonumber\\&- \frac{q^{-a\Delta(m+1)}}{p^{a\Delta\,m}}\big(p^{a(m+1)}-q^{-a(n+1)}\big)\bigg)\mathbb{T}^{a\Delta}_{m+n}+\mathbb{H}^{\Delta(a)}_{p,q}(m,n),
\end{align*}
where \begin{align*}
\mathbb{H}^{\Delta(a)}_{p,q}(m,n)&=q^{-2a\Delta(m+n+1)}\bigg([-\Delta\,m]^Q_{p^{a},q^{a}}[n(1-\Delta)+1]^Q_{p^{a},q^{a}}\nonumber\\&-[-\Delta\,n]^Q_{p^{a},q^{a}}[m(1-\Delta)+1]^Q_{p^{a},q^{a}}\bigg)
\end{align*}
The generalized $q$-Quesne conformal super Witt $n$-algebra is given  by :
\begin{align*}
\Big[{\mathcal T}^{a\Delta}_{m_1},\cdots,{\mathbb T}^{a\Delta}_{m_n}
\Big]&={(-1)^{n+1}\over \big(-p^{-a}+q^{a}\big)^{n-1}}\Big( \mathbb {V}^n_{a\Delta}[n]^Q_{p^{a},q^{a}}{\mathbb T}^{n\,a\Delta}_{\bar{m}}\nonumber\\ &- [n-1]^Q_{p^{a},q^{a}}\big(\mathbb{D}^n_{a\Delta}+ \mathbb{N}^n_{a\Delta}\big){\mathbb T}^{a(n-1)\Delta}_{\bar{m}}\Big)\nonumber\\&+\mathbb{H}^{\Delta}_{p,q}(m_1,\ldots,m_n),
\end{align*}
where 
\begin{align*}
\mathbb{V}^n_{a\Delta}&= p^{a(n-1)(1-\Delta )\bar{m}}\Big(\big(-p^{-a}+q^{a}\big)^{n\choose 2}\prod_{1\leq j < k \leq n}\Big([m_j+1]^Q_{p^{a},q^{a}}\nonumber\\&-[m_k+1]^Q_{p^{a},q^{a}}\Big)+\prod_{1\leq j < k \leq n}\Big(q^{-a(m_j+1)}-q^{-a(m_k+1)}\Big)\Big),
\end{align*}
\begin{align*}
\mathbb{D}^n_{a\Delta}
&=\Big(\big(-p^{-a}+q^{a}\big)^{n\choose 2}\prod_{1\leq j < k \leq n}\frac{q^{-a\Delta(m_k+1)}}{p^{a\Delta\,m_k}}\Big([m_k+1]^Q_{p^{a},q^{a}}-[m_j+1]^Q_{p^{a},q^{a}}\\&+q^{-a(m_k+1)}-p^{a(m_j+1)}\Big),
\end{align*}
\begin{align*}
\mathbb{N}^n_{a\Delta}&=(-1)^{n+1}
\Big(\big(-p^{-a}+q^{a}\big)^{n\choose 2}\prod_{1\leq j < k \leq n}\frac{q^{-a\Delta(m_j+1)}}{p^{a\Delta\,m_j}}\Big([m_j+1]^Q_{p^{a},q^{a}}\nonumber\\&-[m_k+1]^Q_{p^{a},q^{a}}+q^{-a(m_j+1)}-p^{a(m_k+1)}\Big),
\end{align*}
and
\begin{align*}
\mathbb{H}^{\Delta(a)}_{p,q}(m_1,\ldots,m_n)&=q^{-an\Delta(\bar{m}+1)}\prod_{1\leq j < k \leq n}\bigg([-\Delta\,m_j]^Q_{p^{a},q^{a}}[m_k(1-\Delta)+1]^Q_{p^{a},q^{a}}\nonumber\\&-[-\Delta\,m_k]^Q_{p^{a},q^{a}}[m_j(1-\Delta)+1]^Q_{p^{a},q^{a}}\bigg).
\end{align*}
\subsection{A toy model for $\mathcal{R}(p,q)$-super Virasoro constraints with conformal dimension $(\Delta\neq 0,1)$} We study a toy model for the $\mathcal{R}(p,q)$-super Virasoro constraints with conformal dimension $(\Delta\neq 0,1).$ We use step by step  the same procedure investigated in the section\eqref{sub4.2}.
Then,  from the constraints on the partition function,
\begin{align}
\,{\mathcal T}^{a\Delta}_m\,Z^{(toy)}(t)&=0,\quad m\geq 0,\\
\,\mathbb{T}^{a\Delta}_m\,Z^{(toy)}(t)&=0,\quad m\geq 0,
\end{align}
we obtain the following lemma:
\begin{lemma}
	The $\mathcal{R}(p,q)$-super conformal  differential operators given by the relation \eqref{Rpqop} and 
	\begin{align*}
	{\mathbb T}^{a\Delta}_m&=\theta\bigg([\Delta(m+1)+\gamma]_{{\mathcal R}(p^{a},q^{a})}\,m!\, \tau^{a\,m}_1\,{\partial\over \partial t_m}\nonumber\\ &+ h(p^{a},q^{a}){\tau^{a(\Delta(m+1)+\gamma)}_2\over \tau^{a}_1 - \tau^{a}_2}\sum_{k=1}^{\infty}{(m+k)!\over k!}\\&\times B_k(t^{a}_1,\cdots,t^{a}_k){\partial\over \partial t_{m+k}}\bigg).
	\end{align*}
\end{lemma}
Considering the relation\eqref{eq7},\eqref{rpqsop1} and \eqref{rpqsop2}, the $\mathcal{R}(p,q)$-super Witt algebra is generated by the operators \eqref{rpqopb} and
\begin{eqnarray}\label{rpqopc}
{\mathbb T}^{a\Delta}_m=\theta\bigg([x\partial_x+\Delta(m+1)-m]_{{\mathcal R}(p^{a},q^{a})}\bigg)\,x^{m}.
\end{eqnarray} 
Then, the following lemma holds.
\begin{lemma}
	The $\mathcal{R}(p,q)$-super conformal  differential operators given by the relation \eqref{rpqop2} and 
	\begin{align*}
	{\mathbb T}^{a\Delta}_m&=\theta\bigg([x\partial_x+\Delta(m+1)-m]_{{\mathcal R}(p^{a},q^{a})}\,m!\, \tau^{a\,m}_1\,{\partial\over \partial t_m}\nonumber\\ &+ h(p^{a},q^{a}){\tau^{a(\Delta(m+1)+\gamma)}_2\over \tau^{a}_1 - \tau^{a}_2}\sum_{k=1}^{\infty}{(m+1+\Delta(k-1))!\over k!}\\&\times B_k(t^{a}_1,\cdots,t^{a}_k){\partial\over \partial t_{m+1+\Delta(k-1)}}\bigg).
	\end{align*}
\end{lemma}
\begin{remark} For some particular cases of super conformal dimension $\Delta,$ we get:
	\begin{enumerate}
		\item[(a)]	For $\Delta=0,$ we obtain the operators \eqref{Rpqop0} and \begin{align*}
		{\mathbb T}_m&=\theta\bigg([x\partial_x-m]_{{\mathcal R}(p^{a},q^{a})}\,m!\, \tau^{-a\,m}_1\,{\partial\over \partial t_m}\nonumber\\ &+ h(p^{a},q^{a}){\tau^{a(\gamma)}_2\over \tau^{a}_1 - \tau^{a}_2}\sum_{k=1}^{\infty}{(m+1)!\over k!} B_k(t^{a}_1,\cdots,t^{a}_k){\partial\over \partial t_{m+1}}\bigg).
		\end{align*}
		\item[(b)]  $\Delta=1,$ we obtain, the  $\mathcal{R}(p,q)$-super conformal differential operators \eqref{Rpqop1}
		and
		\begin{align*}
		{\mathbb T}^{a}_m&=\theta\bigg([x\partial_x+1]_{{\mathcal R}(p^{a},q^{a})}\,m!\, \tau^{-a\,m}_1\,{\partial\over \partial t_m}\\&+ h(p^{a},q^{a}){\tau^{a(m+1+\gamma)}_2\over \tau^{a}_1 - \tau^{a}_2}\sum_{k=1}^{\infty}{(m+k)!\over k!} B_k(t^{a}_1,\cdots,t^{a}_k){\partial\over \partial t_{m+k}}\bigg).
		\end{align*}
	\end{enumerate}
\end{remark}
\begin{remark}
	The super conformal differential operators corresponding to the quantum algebras known in the literature are derived as:
	\begin{enumerate}
		\item[(a)] The $q$-super conformal differential operator is provided by the relation \eqref{qop} and 
		\begin{align*}
		{\mathbb T}^{a\Delta}_m&=\theta\bigg([\Delta(m+1)+\gamma]_{q^{a}}\,m!\, q^{-a\,m}\,{\partial\over \partial t_m}\nonumber\\ &+ {q^{-a(\Delta(m+1)+\gamma)}\over q^{a} - q^{-a}}\sum_{k=1}^{\infty}{(m+k)!\over k!} B_k(t^{a}_1,\cdots,t^{a}_k){\partial\over \partial t_{m+k}}\bigg).
		\end{align*}
		\item[(b)]The $(p,q)$-super conformal differential operator is given by the relation \eqref{pqop} and  \begin{align*}
		{\mathbb T}^{a\Delta}_m&=\theta\bigg([\Delta(m+1)+\gamma]_{p^{a},q^{a}}\,m!\, p^{-a\,m}\,{\partial\over \partial t_m}\nonumber\\ &+ {q^{a(\Delta(m+1)+\gamma)}\over p^{a} - q^{a}}\sum_{k=1}^{\infty}{(m+k)!\over k!} B_k(t^{a}_1,\cdots,t^{a}_k){\partial\over \partial t_{m+k}}\bigg).
		\end{align*}
		\item[(b)]The $(p^{-1},q)$-super conformal differential operator is given by the relation \eqref{cjop} and  \begin{align*}
		{\mathbb T}^{a\Delta}_m&=\theta\bigg([\Delta(m+1)+\gamma]_{p^{-a},q^{a}}\,m!\, p^{a\,m}\,{\partial\over \partial t_m}\nonumber\\ &+ {q^{a(\Delta(m+1)+\gamma)}\over p^{-a} - q^{a}}\sum_{k=1}^{\infty}{(m+k)!\over k!} B_k(t^{a}_1,\cdots,t^{a}_k){\partial\over \partial t_{m+k}}\bigg).
		\end{align*}
		\item[(d)]The generalized $q$-Quesne super conformal differential operator is determined by the relation \eqref{qQop} and \begin{align*}
		{\mathbb T}^{a\Delta}_m&=\theta\bigg([\Delta(m+1)+\gamma]^Q_{p^{a},q^{a}}\,m!\, p^{-a\,m}\,{\partial\over \partial t_m}\nonumber\\ &+ {q^{-a(\Delta(m+1)+\gamma)}\over p^{a} - q^{-a}}\sum_{k=1}^{\infty}{(m+k)!\over k!} B_k(t^{a}_1,\cdots,t^{a}_k){\partial\over \partial t_{m+k}}\bigg).
		\end{align*}
	\end{enumerate}
\end{remark}
\section{Generalized  matrix model with conformal dimension $(\Delta\neq 0,1)$}
In this section, we characterize the  matrix model with conformal dimension $(\Delta\neq 0,1)$ from the $\mathcal{R}(p,q)$-deformed quantum algebra. Moreover, we derive particular cases induced by quantum algebra known in the literature.

We consider now the following relation:
\begin{eqnarray}
\left \{
\begin{array}{l}
F(z)=z, \\
\\
G(P,Q)=\frac{q^{Q}-p^{P}}{p^{Q}\mathcal{R}(p^P,q^Q)},\quad \mbox{if}\quad l>0,
\end{array}
\right .
\end{eqnarray}
where $l$ is given in the relation \eqref{r10}.
\begin{definition}
	The  $\mathcal{R}(p,q)$-Pochhammer symbol is given by:
	\begin{equation}
	\big(z;\mathcal{R}(p,q)\big)_{0}=1,\quad \big(z;\mathcal{R}(p,q)\big)_{n}:=\prod\limits_{k=0}^{n-1}\left( 1-F\big(\frac{q^{k}}{ p^k}\,z\big)G(P,Q)\right),\: n\in \mathbb{N},
	\end{equation}
	and the $\mathcal{R}(p,q)$-theta function $\theta(z; \mathcal{R}(p,q))$ as follows:
	\begin{equation}\label{rpqtheta}
	\theta(z; \mathcal{R}(p,q))=\prod\limits_{k=0}^{\infty}\bigg( 1-F\big(\frac{q^{k}}{ p^k}z\big)G(P,Q)\bigg)\prod\limits_{k=0}^{\infty}\bigg( 1-F\big(\frac{q^{k+1}}{ p^{k+1}}z^{-1}\big)G(P,Q)\bigg),
	\end{equation}
\end{definition}
Note that the $\mathcal{R}(p,q)$-theta function \eqref{rpqtheta} satisfies the following properties:
\begin{equation}\label{p1theta}
\theta\big(qz; \mathcal{R}(p,q)\big) = \theta \big((pz)^{-1}; \mathcal{R}(p,q)\big),
\end{equation}
and
\begin{equation}\label{p2theta}
\theta\big(q^{-1}z; \mathcal{R}(p,q)\big) =  \frac{p}{q} z^2~ \theta\big((pz)^{-1}; \mathcal{R}(p,q)\big).
\end{equation}
\begin{definition}
	The  generating function for the $\mathcal{R}(p,q)$-elliptic hermitian matrix model is defined as follows:
	\begin{eqnarray} 
	Z_N^{\rm ell}(\{ t\})= \oint \prod_{i=1}^N \frac{d z_i}{z_i} \prod_{i<j}\theta\bigg(\frac{z_i}{z_j};\mathcal{R}(p,q)\bigg)
	\theta\bigg(\frac{z_j}{z_i};\mathcal{R}(p,q)\bigg)e^{\sum\limits_{k=0}^{\infty} \frac{t_k}{k!}\sum\limits_{i=1}^N z_i^k},
	\label{wilson:loop}
	\end{eqnarray}
	where $\{t\}=\{t_k|k\in\mathbb{N}\}.$ The integration is over the counter around the origin.
\end{definition}
Note that, 
taking $\mathcal{R}(x,1)=\frac{x-1}{x},$ we obtain the well known $q$-Pochhammer symbol:
\begin{equation*}
\big(z;q\big)_{0}=1,\quad \big(z;q\big)_{n}:=\prod\limits_{k=0}^{n-1}\left( 1-q^{k}\,z\right),\: n\in \mathbb{N},
\end{equation*}
the $q$-theta function
\begin{equation*}
\theta(z, q)=\prod\limits_{k=0}^{\infty}\left( 1-q^{k}\,z\right)\prod\limits_{k=0}^{\infty}\left( 1-q^{k+1}\,z^{-1}\right) .
\end{equation*}
and the generating function for the elliptic hermitian matrix model \cite{NZ}:
\begin{equation*} 
Z_N^{\rm ell}(\{ t\})= \oint \prod_{i=1}^N \frac{d z_i}{z_i} \prod_{i<j}\theta(\frac{z_i}{z_j};q)
\theta(\frac{z_j}{z_i};q)e^{\sum\limits_{k=0}^{\infty} \frac{t_k}{k!}\sum\limits_{i=1}^N z_i^k}.
\end{equation*}
%In this section, the matrix models associated to the $\mathcal{R}(p,q)$- deformation are investigated and particular cases corresponding to quantum deformations existing in literature are derived.

\begin{definition}
	The $\mathcal{R}(p,q)$-differential opertor with conformal dimension $(\Delta\neq 0,1)$ is defined by:
	\begin{eqnarray}
	T^{\Delta}_n := -\sum\limits_{l=1}^N z_l^{(1-\Delta)(n+1)}D_{\mathcal{R}(p,q)}\,z_l^{\Delta(n+1)},\label{rpqVir-matr-op}
	\end{eqnarray}
	which acts on the functions of $N$ variables  and $ D_{\mathcal{R}(p,q)}$ is $\mathcal{R}(p,q)$-derivative \eqref{rpqder} with respect to the $z_l$-variable. 
\end{definition}
Note that, the conformal operators \eqref{rpqVir-matr-op} satisfy the commutation relation \eqref{p1a}.
\begin{proposition}
	The $\mathcal{R}(p,q)$-conformal operator \eqref{rpqVir-matr-op} can be given by:
	\begin{align}
	T_n^{\Delta}&=\frac{K(p,q)}{p-q}\bigg[(\frac{q}{p})^{\Delta(n+1)-N} \sum\limits_{l=0}^\infty \frac{(l+n-2N)!}{l!} B_l(\tilde{t}_1,\ldots , \tilde{t}_l)\nonumber\\&\times
	{\mathcal D}_N \frac{\partial}{\partial t_{l+n-2N}}- p^{\Delta(n+1)} n!\frac{\partial}{\partial t_n}\bigg].
	\label{rpqcoperator}
	\end{align}
	where ${\it D}_{N}$ is a differential operator\eqref{diffcop}.
\end{proposition}
\begin{proof}
	Putting these operators under the contour integral (\ref{wilson:loop}), we obtain naturally zero. 
	Now we have to evaluate how these differential operators act on the integrand. 
	From the relations \eqref{p1theta} and \eqref{p2theta} concerning the properties of the generalized $\theta$ function and by setting
	\begin{eqnarray}
	g(z)=\prod_{i<j}\theta\big(\frac{z_i}{z_j};\mathcal{R}(p,q)\big)
	\theta\big(\frac{z_j}{z_i};\mathcal{R}(p,q)\big)
	\end{eqnarray}
	and $$f(z_l)=z^{\Delta(n+1)}_l,$$
	we have:  
	\begin{align}\label{useful-id}
	A&:=\sum_{l=1}^N D_{\mathcal{R}(p,q)}\bigg(f(z_l)\prod_{i<j}\theta\big(\frac{z_i}{z_j};\mathcal{R}(p,q)\big)\theta\big(\frac{z_j}{z_i};\mathcal{R}(p,q)\big)\bigg)\nonumber\\&= 
	\sum_{l=1}^N\frac{K(P,Q)}{(p-q)z_l}\bigg(\frac{f(qz_l)}{f(pz_l)}\prod_{j\neq l}\frac{p}{q}\frac{z_j^2}{z_l^2}-1\bigg)
	f(pz_l)\nonumber\\&\times\prod_{i<j}\theta\big(\frac{z_i}{z_j};\mathcal{R}(p,q)\big)\theta\big(\frac{z_j}{z_i};\mathcal{R}(p,q)\big)\nonumber\\
	&= 
	\sum_{l=1}^N\frac{K(p,q)}{p-q}\bigg(\big(\frac{q}{p}\big)^{\Delta(n+1)}\prod_{j\neq l}\frac{p}{q}\frac{z_j^2}{z_l^2}-1\bigg)
	p^{\Delta(n+1)}z^{\Delta(n+1)-1}_l\nonumber\\&\times\prod_{i<j}\theta\big(\frac{z_i}{z_j};\mathcal{R}(p,q)\big)\theta\big(\frac{z_j}{z_i};\mathcal{R}(p,q)\big).
	\end{align}
	Thus,
	\begin{align}
	T^{\Delta}_ng(z)&=\sum_{l=1}^N\frac{K(p,q)}{p-q}\bigg(\big(\frac{q}{p}\big)^{\Delta(n+1)}\prod_{j\neq l}\frac{p}{q}\frac{z_j^2}{z_l^2}-1\bigg)
	p^{\Delta(n+1)}z^{n}_l\nonumber\\&\times\prod_{i<j}\theta\big(\frac{z_i}{z_j};\mathcal{R}(p,q)\big)\theta\big(\frac{z_j}{z_i};\mathcal{R}(p,q)\big).
	\end{align}
	The $n$th complete Bell polynomial $B_n$  satisfy the following relations:
	\begin{eqnarray}\label{Bell-pr3}
	B_l (\tilde{t}_1,\ldots ,\tilde{t}_l) = \sum\limits_{\nu=0}^l \tau_1^{l-\nu}\tau_2^{\nu} \binom{l}{\nu} B_{\nu} (t_1, \ldots, t_{\nu}) B_{n-\nu} (-t_1, \ldots , - t_{n-\nu}),
	\end{eqnarray}
	where $\tilde{t}_k = (\tau_1^k-\tau_29^k) t_k,$ 
	and 
	\begin{align}
	\exp\left(\sum\limits_{k=1}^{\infty}\frac{t_k}{k!}\tau^k z_i^k\right)&=
	\sum_{k=0}^{\infty}\frac{1}{k!}B_k\left(\tilde{t}_1,\dots,\tilde{t}_k\right)x^k\exp\left(\sum\limits_{l=1}^{\infty}\frac{t_l}{l!} z_i^l\right)
	.\label{Bells-useful}
	\end{align}
	From the formulas (\ref{useful-id}) and (\ref{Bells-useful}),  the $\mathcal{R}(p,q)$-conformal operator (\ref{rpqVir-matr-op})
	can be rewritten in the simpler form:
	\begin{align}
	T^{\Delta}_ng(z)&=\frac{K(p,q)}{p-q}\bigg[\prod_{j=1}^{N} z_j^2 \sum_{l=1}^{N}\sum_{k,\nu=0}^{\infty}\big(\frac{q}{p}\big)^{\Delta(n+1)-N}q^{k}\frac{1}{k!\nu!}B_k(t_1,\dots,t_k)\nonumber\\&\times
	B_{\nu}(-t_1,\dots,-t_{\nu})z_l^{k+\nu+n-2N} - p^{\Delta(n+1)}\,\sum_{l=1}^{N} z_l^n\bigg].\label{gen-terms-elvir}
	\end{align}
	Note that, the expectation value of these terms could be zero. 
	Let us now   generate these terms by using the  
	$t$-derivatives of the integrand of the relation (\ref{wilson:loop}). Thus, to do it, 
	we can rewrite $\prod\limits_{i=1}^N z_i$ in terms of sums  $\sum\limits_{i=1}^N z_i^k$ 
	by using the Newton's identities, 
	\begin{eqnarray}
	\prod_{i=1}^N z_i=\frac{1}{N!}\left|\begin{array}{cccccc}
	\nu_1 & 1 & 0 & \dots &  &  \\
	\nu_2  & \nu_1 & 2 & 0 & \dots & \\
	\dots   &  \dots  &\dots &  \dots  &\dots &\\
	\nu_{N-1} & \nu_{N-2} & \dots  &\dots& \nu_1 & N-1\\
	\nu_{N} & \nu_{N-1} & \dots  &\dots& \nu_2 & \nu_1 
	\end{array}\right|,
	\end{eqnarray}
	where $\nu_k\equiv \sum\limits_{i=1}^{N}z_{i}^k ,$   the terms $\sum\limits_{i=1}^{N}z_{i}^k$ may be generated by taking 
	the derivatives with respect to $t.$ We consider the following differential operator
	\begin{eqnarray}\label{diffcop}
	{\mathcal D}_N =\frac{1}{N!}\left|\begin{array}{cccccc}
	2! \frac{\partial}{\partial t_2} & 1 & 0 & \dots &  &  \\
	4!\frac{\partial}{\partial t_4}  &2! \frac{\partial}{\partial t_2} & 2 & 0 & \dots & \\
	\dots   &  \dots  &\dots &  \dots  &\dots &\\
	(2N-2)! \frac{\partial}{\partial t_{2N-2}} & (2N-4)!\frac{\partial}{\partial t_{2N-4}} & \dots  &\dots& 2!\frac{\partial}{\partial t_2} & N-1\\
	(2N)! \frac{\partial}{\partial t_{2N}} & (2N-2)!\frac{\partial }{\partial t_{2N-2}} & \dots  &\dots& 4!\frac{\partial}{\partial t_4} & 2! \frac{\partial}{\partial t_2} 
	\end{array}\right|,
	\end{eqnarray}
	with 
	\begin{eqnarray}
	\prod_{j=1}^{N} z_j^2 e^{\sum\limits_{k=0}^{\infty} \frac{t_k}{k!}\sum\limits_{i=1}^N z_i^k  }=
	{\mathcal D}_N  \left (e^{\sum\limits_{k=0}^{\infty} \frac{t_k}{k!}\sum\limits_{i=1}^N z_i^k  } \right ).
	\end{eqnarray}
	By using all together we derive $\mathcal{R}(p,q)$-conformal operator as follows:
	\begin{align}
	T^{\Delta}_n &= \frac{K(P,Q)}{p-q}\bigg[  \sum_{k,\nu=0}^{\infty}\big(\frac{q}{p}\big)^{\Delta(n+1)-N}q^{k}\frac{(k+\nu+n-2N)!}{k!\nu!}B_k(t_1,\dots,t_k)\nonumber\\&\times
	B_{\nu}(-t_1,\dots,-t_{\nu}) {\mathcal D}_N \frac{\partial}{\partial t_{k+\nu+n-2N}} - 
	p^{\Delta(n+1)}\,n!\frac{\partial}{\partial t_n}\bigg].\label{rpqVir-exp-matrix}
	\end{align}
	Note that, the above operator \eqref{rpqVir-exp-matrix}
	annihilates the generating function $Z_N^{\rm ell}(\{ t\})$. From the property of the Bell number, the relation \eqref{rpqcoperator} follows and the proof is achieved. 
\end{proof}
Taking $\Delta=1$ and $\mathcal{R}(x,y)=(p-q)^{-1}(x-y),$ we recovered the $q$-differential operator given in \cite{NZ}.
\begin{remark}
	Particular cases of Pochhammer symbol, theta function, and the generating function for the elliptic hermitian matrix model  related to quantum algebras known in the literature are deduced as:
	\begin{enumerate}
		\item[(a)]
		the  $(p,q)$-Pochhammer symbol:
		\begin{equation*}
		\big(z;q\big)_{0}=1,\quad \big(z;q\big)_{n}:=\prod\limits_{k=0}^{n-1}\left( 1-q^{-2k}\,z\right),\: n\in \mathbb{N},
		\end{equation*}
		the $q$- theta function
		\begin{equation*}
		\theta(z, q)=\prod\limits_{k=0}^{\infty}\left( 1-q^{-2k}\,z\right)\prod\limits_{k=0}^{\infty}\left( 1-q^{-2(k+1)}\,z^{-1}\right) .
		\end{equation*}
		and the generating function for the elliptic hermitian matrix model:
		\begin{equation*} 
		Z_N^{\rm ell}(\{ t\})= \oint \prod_{i=1}^N \frac{d z_i}{z_i} \prod_{i<j}\theta(\frac{z_i}{z_j};q)
		\theta(\frac{z_j}{z_i};q)e^{\sum\limits_{k=0}^{\infty} \frac{t_k}{k!}\sum\limits_{i=1}^N z_i^k}.
		\end{equation*}
		Moreover, the $q$-conformal operator \begin{eqnarray*}
			T^{\Delta}_n := -\sum\limits_{l=1}^N z_l^{(1-\Delta)(n+1)}D_{q}\,z_l^{\Delta(n+1)},
		\end{eqnarray*} can be given by:
		\begin{align*}
		T_n^{\Delta}&=\frac{K(q)}{q-q^{-1}}\bigg[q^{-2\Delta(n+1)-N} \sum\limits_{l=0}^\infty \frac{(l+n-2N)!}{l!} B_l(\tilde{t}_1,\ldots , \tilde{t}_l)\nonumber\\&\times
		{\mathcal D}_N \frac{\partial}{\partial t_{l+n-2N}}- q^{\Delta(n+1)} n!\frac{\partial}{\partial t_n}\bigg].
		\end{align*}
		\item[(b)]
		The  $(p,q)$-Pochhammer symbol:
		\begin{equation*}
		\big(z;p,q\big)_{0}=1,\quad \big(z;p,q\big)_{n}:=\prod\limits_{k=0}^{n-1}\left( p^{k}-q^{k}\,z\right),\: n\in \mathbb{N},
		\end{equation*}
		the $(p,q)$- theta function
		\begin{equation*}
		\theta(z;p, q)=\prod\limits_{k=0}^{\infty}\left( p^{k}-q^{k}\,z\right)\prod\limits_{k=0}^{\infty}\left( p^{k+1}-q^{k+1}\,z^{-1}\right) .
		\end{equation*}
		and the generating function for the elliptic hermitian matrix model:
		\begin{equation*} 
		Z_N^{\rm ell}(\{ t\})= \oint \prod_{i=1}^N \frac{d z_i}{z_i} \prod_{i<j}\theta(\frac{z_i}{z_j};p,q)
		\theta(\frac{z_j}{z_i};p,q)e^{\sum\limits_{k=0}^{\infty} \frac{t_k}{k!}\sum\limits_{i=1}^N z_i^k}.
		\end{equation*}
		Moreover, the $(p,q)$-conformal operator \begin{eqnarray*}
			T^{\Delta}_n := -\sum\limits_{l=1}^N z_l^{(1-\Delta)(n+1)}D_{p,q}\,z_l^{\Delta(n+1)},
		\end{eqnarray*} can be given by:
		\begin{align*}
		T_n^{\Delta}&=\frac{K(p,q)}{p-q}\bigg[(\frac{q}{p})^{\Delta(n+1)-N} \sum\limits_{l=0}^\infty \frac{(l+n-2N)!}{l!} B_l(\tilde{t}_1,\ldots , \tilde{t}_l)\nonumber\\&\times
		{\mathcal D}_N \frac{\partial}{\partial t_{l+n-2N}}- p^{\Delta(n+1)} n!\frac{\partial}{\partial t_n}\bigg].
		\end{align*}
	\end{enumerate}
\end{remark}
\section{Concluding  remarks}
The $\mathcal{R}(p,q)$-super Virasoro algebra, the $\mathcal{R}(p,q)$-conformal Virasoro $n$ algebra, the $\mathcal{R}(p,q)$-conformal super Virasoro $n$-algebra ($n$ even) have been constructed.  and discuss A toy model for the $\mathcal{R}(p,q)$-conformal Virasoro constraints and $\mathcal{R}(p,q)$-conformal super Virasoro constraints and  the $\mathcal{R}(p,q)$-elliptic hermitian matrix model with an arbitrary conformal dimension $\Delta$ have been discussed. Also, particular cases induced by quantum algebras existing in the literature have been derived. Note that the relation between the $\mathcal{R}(p,q)$-super Virasoro algebra and the nonlinear integrable system(super KdV equations) can be investigated in the futur work.  
\section*{Acknowledgments}
This research was partly supported by the SNF Grant No. IZSEZ0\_206010.

\end{document}